\documentclass[runningheads]{llncs}

\usepackage[dvipsnames]{xcolor}

\usepackage{xspace}

\usepackage{microtype}%if unwanted, comment out or use option "draft"
\usepackage{todonotes}
\usepackage{amsmath}
\usepackage{amsmath,amsfonts,amssymb} % Math packages
\usepackage{url}
\usepackage{stmaryrd}
\usepackage{mathrsfs}
\usepackage{wasysym}
\usepackage{bm}
\usepackage{alltt}

\usepackage[inline,shortlabels]{enumitem}
\usepackage{todonotes}
\usepackage{wrapfig}
\usepackage{subcaption}
\usepackage{pgfplots}
\usepackage{multirow}
\usepackage{verbatim}
\usepackage{mathtools}
\usepackage{booktabs}
\usepackage{extarrows}
\usepackage{mathtools}

\usepackage{algorithm}
\usepackage{algorithmic}
\usepackage{listings}
\usepackage[title]{appendix}
\usepackage{marginnote}
\usepackage{listings}
\usepackage{url}
\usepackage{wrapfig}
\usepackage{graphicx}
\usepackage{colortbl}
\usepackage{pifont}
\usepackage{hhline}
\usepackage[colorlinks=true,citecolor=blue,linkcolor=blue,urlcolor=blue]{hyperref}
\usepackage{graphicx}
\usepackage{adjustbox}
\usepackage[autostyle]{csquotes}

\definecolor{Blueberry}{RGB}{4,51,255}

\usepackage{filecontents}

\usepackage{tikz, tikzscale, pgfplots}
\usetikzlibrary{arrows,arrows.meta,calc,automata,positioning,decorations.pathreplacing,shapes.geometric,shapes.misc,graphs,backgrounds,shadows.blur,snakes}
\tikzset{align at top/.style={baseline=(current bounding box.north)}}
\tikzstyle{every node}=[font=\scriptsize]
\tikzstyle{state} = [draw,fill=white,ellipse,thick,align=center,inner sep=0pt,minimum size=4.5mm]
\tikzstyle{vvert} = [draw,fill=white,ellipse,thick,align=center,inner sep=-2pt,minimum size=8mm]
\tikzstyle{rvert} = [draw,fill=white,rectangle,thick,align=center,inner sep=3pt,minimum size=7mm]
\tikzstyle{dot} = [fill,circle,inner sep=0mm,minimum size=1.25mm,line width=0mm]

\makeatletter

\def\orcidID#1{\smash{\href{http://orcid.org/#1}{\protect\raisebox{-1.25pt}{\protect\includegraphics{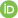}}}}}
\makeatother

\allowdisplaybreaks

\newcommand{\Nat}{\ensuremath{\mathbb{N}}}
\newcommand{\Reals}{\ensuremath{\mathbb{R}}}
\newcommand{\Realsnn}{\ensuremath{\Reals_{\geq0}}}

\newcommand{\NP}{\mathsf{NP}}
\newcommand{\coNP}{\mathsf{coNP}}

\newcommand{\powerset}{\mathscr{P}}% requires package mathrsfs
\newcommand{\powersetf}{\powerset_{\mathsf{f}}}

\newcommand{\last}{\textit{last}}

\newcommand{\mult}{\ }

\newcommand{\proj}{\mathrm{pr}}

\newcommand{\charac}{\mathbf{1}}

\newcommand{\Finally}{\Diamond}
\newcommand{\Until}{\mathbin{\mathsf{U}}}

\newcommand{\limp}{\Rightarrow}

\newcommand{\abs}[1]{\lvert{#1}\rvert}
\newcommand{\bvec}[1]{\boldsymbol{#1}}
\newcommand{\Poly}{\mathsf{Poly}}
\newcommand{\Simp}{\mathsf{Simp}}
\newcommand{\DPoly}{\mathsf{DPoly}}
\newcommand{\DSimp}{\mathsf{DSimp}}
\newcommand{\Triang}{\mathsf{Triang}}

\newcommand{\vertices}{\mathbb{V}}

\newcommand{\Dist}{\mathsf{Dist}}
\newcommand{\Dirac}{\delta}

\newcommand{\Salg}{\bm{\Sigma}}
\newcommand{\PMeas}{\mathsf{PMeas}}
\newcommand{\diff}{\mathbf{d}}
\newcommand{\sect}[1]{{|_{#1}}}

\newcommand{\Prob}{\mathbb{P}}
\newcommand{\Expect}{\mathbb{E}}

\newcommand{\maxplay}{{\Box}}
\newcommand{\minplay}{{\Diamond}}

\newcommand{\psgtrans}{\Theta}
\newcommand{\psgnodes}{\sgnodes}
\newcommand{\psgnmax}{\sgnmax}
\newcommand{\psgnmin}{\sgnmin}

\newcommand{\StochK}{\mathcal{K}}

\newcommand{\reward}{\mathit{r}}
\newcommand{\sgtrans}{\theta}
\newcommand{\sgnodes}{\mathcal{S}}
\newcommand{\sgnmax}{\sgnodes_\maxplay}
\newcommand{\sgnmin}{\sgnodes_\minplay}

\newcommand{\sgactions}{\mathcal{A}}

\newcommand{\enabled}{\sgactions}

\newcommand{\StochG}{\mathcal{G}}
\newcommand{\StochH}{\mathcal{H}}

\newcommand{\GamePaths}{\mathsf{Paths}}
\newcommand{\FinGamePaths}{\mathsf{FPaths}}

\newcommand{\textrew}{\mathsf{rew}}

\newcommand{\Rewards}[1]{\textrew_{#1}}
\newcommand{\TRewards}{\Rewards{\mathsf{t}}}
\newcommand{\ARewards}{\Rewards{\mathsf{a}}}
\newcommand{\DRewards}[1]{\Rewards{#1}}
\newcommand{\GRewards}{\Rewards{\fgen}}
\newcommand{\discfactor}{\gamma}
\newcommand{\fgen}{\mathsf{f}}

\newcommand{\FRewards}[1]{\widehat{\textrew}_{#1}}
\newcommand{\FTRewards}{\FRewards{\mathsf{t}}}
\newcommand{\FARewards}{\FRewards{\mathsf{a}}}
\newcommand{\FDRewards}[1]{\FRewards{#1}}
\newcommand{\FGRewards}{\FRewards{\fgen}}

\newcommand{\goal}{\mathit{G}}

\newcommand{\StochGK}{{\StochG_{\StochK}}}
\newcommand{\StochGKp}{{\StochG_{\StochK'}}}
\newcommand{\StochHK}{{\StochH_{\StochK}}}
\newcommand{\strat}{\pi}
\newcommand{\starredstrat}{\strat^*}
\newcommand{\supv}[1]{{#1}^{\mathsf{v}}}
\newcommand{\supx}[1]{{#1}^{\mathsf{x}}}
\newcommand{\stratv}{\supv{\strat}}
\newcommand{\stratx}{\supx{\strat}}
\newcommand{\memorylessidx}{\textsl{M}}
\newcommand{\deterministicidx}{\textsl{D}}
\newcommand{\semimarkovidx}{\textsl{S}}
\newcommand{\extremeidx}{\textsl{X}}

\newcommand{\Strategies}[1]{\Pi_{#1}}
\newcommand{\MemorylessStrats}[1]{\Strategies{#1}^{\memorylessidx}}
\newcommand{\SemiMarkovStrats}[1]{\Strategies{#1}^{\semimarkovidx}}

\newcommand{\XSemiMarkovStrats}[1]{\Strategies{#1}^{\extremeidx\semimarkovidx}}
\newcommand{\DetMemorylessStrats}[1]{\Strategies{#1}^{\memorylessidx\deterministicidx}}
\newcommand{\XDetMemorylessStrats}[1]{\Strategies{#1}^{\extremeidx\memorylessidx\deterministicidx}}
\def\val{\mathop{\textup{val}}}

\def\StochH{\mathcal{H}}

\def\val#1#2{\mathit{Val}_{#1}(#2)}

\mathchardef\mhyphen="2D

\def\PRISM{\textsf{PRISM}\xspace}
\def\PRISMGAMES{\textsf{PRISM-Games}\xspace}

\def\LTL{\textsf{LTL}}

\spnewtheorem*{proofofclaim}{Proof of claim}{\itshape}{\rmfamily}
\def\qedclaim{\hfill\emph{(End of claim)}\qed}

\usepackage{color}

\definecolor{lightblue}{RGB}{220,220,255}
\definecolor{lightred}{RGB}{255,224,224}
\definecolor{lightgreen}{RGB}{224,255,224}
\definecolor{lightyellow}{RGB}{255,255,224}
\definecolor{lightpurple}{RGB}{255,224,255}
\definecolor{darkerred}{RGB}{64,0,0}
\definecolor{darkred}{RGB}{128,0,0}
\definecolor{darkblue}{RGB}{0,0,128}
\definecolor{darkgreen}{RGB}{0,128,0}
\definecolor{darkpurple}{RGB}{128,0,128}
\definecolor{black}{RGB}{0,0,0}

\makeatletter
\def\THICKhrulefill{\leavevmode \leaders \hrule height 5pt\hfill \kern \z@}
\makeatother

\definecolor{codegreen}{rgb}{0,0.6,0}
\definecolor{codegray}{rgb}{0.5,0.5,0.5}
\definecolor{codepurple}{rgb}{0.58,0,0.82}
\definecolor{backcolour}{rgb}{0.95,0.95,0.92}

\lstdefinestyle{polyprism}{
    commentstyle=\color{codegreen},
    keywordstyle=\color{magenta},
    numberstyle=\tiny\color{codegray},
    stringstyle=\color{codepurple},
    basicstyle=\ttfamily\scriptsize,
    breakatwhitespace=false,         
    breaklines=true,                 
    captionpos=b,                    
    keepspaces=true,                 
    numbers=none,                    
    numbersep=5pt,                  
    showspaces=false,                
    showstringspaces=false,
    showtabs=false, 
    columns=flexible,
    morecomment = [l]{//},                 
    tabsize=2,
    escapechar=$
}

\urldef{\mailsa}\path|pedro.dargenio@unc.edu.ar|
\urldef{\mailsb}\path|castro@dc.exa.unrc.edu.ar|

\graphicspath{{./Figs/}}%helpful if your graphic files are in another directory.

\title{Polytopal Stochastic Games
  \thanks{This work was supported by
    Agencia {I+D+i} PICT 2019-03134,
    SeCyT-UNC 33620230100384CB (MECANO), and
    EU Grant agreement ID: 101008233 (MISSION).}}

\titlerunning{\vspace{-3cm}Polytopal Stochastic Games}

\author{
Pablo F. Castro \inst{1,3} \orcidID{0000-0002-5835-4333}\and
Pedro R. D'Argenio \inst{2,3} \orcidID{0000-0002-8528-9215}%\and
}
\authorrunning{P.F. Castro et al.}
\institute{Universidad Nacional de 
  R\'{\i}o Cuarto, FCEFQyN, Departamento de Computaci\'on,
  R\'{\i}o Cuarto, C\'ordoba,
  Argentina,
  \mailsb
  \and 
  Universidad Nacional de C\'ordoba, FAMAF,
  C\'ordoba,
  Argentina,
  \mailsa
  \and
  Consejo Nacional de Investigaciones Cient\'ificas y T\'ecnicas (CONICET), Argentina}

\begin{document}
\maketitle

\begin{abstract}
  In this paper we introduce \emph{polytopal stochastic games}, an
  extension of two-player, zero-sum, turn-based stochastic games, in
  which we may have uncertainty over the transition probabilities.  In
  these games the uncertainty over the probability distributions is
  captured via linear (in)equalities whose space of solutions forms a
  polytope.  We give a formal definition of these games and prove
  their basic properties: determinacy and existence of optimal
  memoryless and deterministic strategies.  We do this for
  reachability and different types of reward objectives and show that
  the solution exists in a finite representation of the game.  We also
  state that the corresponding decision problems are in
  $\NP\cap\coNP$.  We motivate the use of polytopal stochastic games
  via a simple example.
%%   Finally, we report some experiments we
%%   performed with a prototype tool.
\end{abstract}

\section{Introduction}

    In the last decades,  stochastic systems have become ubiquitous in computer science: communication and security protocols, fault analysis in critical systems,  autonomous devices,  to name a few examples, typically use techniques coming from probability theory.  Furthermore,  well-known techniques in artificial intelligence, such as reinforcement learning \cite{ReinfLearning}, are based on stochastic models.
    In view of this, the verification and formal analysis of stochastic systems is one of the most active areas of research in software verification.
    Christel Baier and Joost-Pieter Katoen's book \cite{BaierK08} is considered a standard reference in the area,  it introduces common concepts and techniques for model checking probabilistic systems, this includes algorithms for verifying temporal assertions over Markov chains (MCs) and Markov Decision processes (MDPs). The latter can be considered as one player stochastic games, in which the system has to select strategies to solve non-determinism in stochastic settings.  
%\remarkPC{quizas acá también podemos destacar otros trabajos de Baier}
In general,  game theory offers a powerful mathematical framework for specifying and verifying computing systems. The idea is appealing,  a computing system can be thought of as a player playing against an environment, or another system, while trying to achieve certain goals.
For instance, a security system can be seen as a player that selects different countermeasures to possibly different types of maneuvers executed by an attacker (a second player) each of which may succeed with certain probabilities.  The objective of the defense system is to minimize the probability that the attack succeeds while the attacker wants to maximize it.
%\remarkPRD{Pablo: fijate si te gusta este otro ejemplo}
%% For instance, consider a communicating protocol,  we may think of it  as a player that wants  to send a message, while the possible faults that may occur during the communication could be thought of as another player.  The objective of the system is to guarantee certain probability of successfully sending the message to the receiver.
%
This scenario can be modeled as a stochastic game, and then analysed using techniques coming from game theory.   Examples of applications of game theory  to the analysis of systems can be found almost everywhere in the last years: self-driving cars~\cite{WangWTGS21},  robotics~\cite{DBLP:conf/qest/Junges0KTZH18},  UAVs~\cite{DBLP:journals/tase/FengWHT16},  security~\cite{DBLP:conf/csfw/AslanyanNP16},  etc.  Furthermore,  in recent years,  some model checkers have been extended  to provide support for stochastic games,  e.g., this is the case of \PRISMGAMES~\cite{DBLP:conf/tacas/ChenFKPS13}, which offers support for several versions of stochastic games. % any other supporting games?

In this paper we focus on two-player, zero-sum,  turn-based perfect-information stochastic games.  Intuitively,  they are non-deterministic probabilistic transition systems in which the vertices are partitioned into two sets:  vertices belonging to player $\maxplay$ and vertices belonging to player $\minplay$.  When the current state belongs to a given player, say $\maxplay$, she performs an action by selecting one of the non-deterministic outgoing transitions which would lead to different states with some given probabilities.
%% In this paper we focus on two-player, zero-sum,  turn-based perfect-information stochastic games.  Intuitively,  they are graph games in which the vertices are partitioned into two sets:  vertices belonging to player 1 and vertices belonging to player 2.  Players' actions are graph transitions, thus the players take turns to move a token from one vertex to another.
Typically,  the players want to fulfill or maximize/minimize some objectives. Standard quantitative objectives are discounted sum (the players collect an amount of rewards during the play which are multiplied by a discount factor in each step), total sum (the players want to maximize/minimize the cumulative sum of the rewards collected during a play), mean-payoff (the objective is to maximize or minimize the long-run average reward), or simply a reachability objective, that is, they aim to  maximize/minimize the probability of reaching certain subset of states. These kinds of objectives can be used and combined to model different kinds of systems,  e.g., the case of a self-driving car intending to maximize the probability of reaching some zone in a city can be seen as a multiobjective game \cite{DBLP:conf/qest/ChenKSW13}.   
     
Most of the time, when modeling stochastic systems, one assumes that the probability distributions are exactly  known,  which may not always be the case due to measurement inaccuracies, lack of data, or other issues.  In this paper we propose an extension of stochastic games that adds the possibility of having uncertainty over the probabilities.  Games with some kinds of uncertainty have been considered for $1\frac{1}{2}$-player games, i.e.,  Markov Decision Processes (MDPs). For instance, Interval-valued Discrete-Time Markov Chains (IDTMCs)~\cite{JonssonL91,KozineU02,SenVA06}, Interval Markov Decision Process (IMDP)~\cite{SenVA06}, and Convex MDPs~\cite{DBLP:conf/cav/PuggelliLSS13}.  To the best of our knowledge,  these approaches  have not been extended to stochastic games (i.e., $2\frac{1}{2}$-player games).  A key challenge for doing so is that in multiplayer games one needs to prove determinacy results,  this ensures that the games possess  a well-defined value,  which does not depend on the players' knowledge.
In the aforementioned approaches the notion of uncertainty is usually adversarially resolved,  that is,  each time a state is visited,  the adversary 
picks a transition distribution that respects the constraints, and takes a probabilistic step according to the chosen distribution. However, it is interesting to note that, in two-player games we may adopt two ways of resolving uncertainty:   a controllable one, in which the actual player resolves the uncertainty following her goals; and an adversarial one in which the adversary resolves the uncertainty in her favor. The former approach  is useful in those scenarios in which the uncertainty affects the adversary as she does not precisely know the possible movements of  our player; while the latter is helpful to reason in worst-case scenarios.

We therefore introduce \emph{polytopal stochastic games (PSG)}. PSGs, as defined in Section \ref{sec:polytopal-games}, allow one to model uncertainties over probability distributions using linear (closed) inequalities.  Geometrically, these linear inequalities correspond to polytopes,  i.e., bounded polyhedra.  As PSGs are two-player games,  both ways of resolving uncertainty are possible: the adversarial approach and the controllable one.  Furthermore, we show that in all the cases  these kinds of games preserve some good properties of standard stochastic games for several objectives: reachability, total rewards, average sum,  and mean payoff.  In particular we show that these games are determined and admit optimal memoryless and deterministic strategies.
We also show that these inherently infinite games can be reduced to equivalent finite stochastic games that traverse exclusively through the vertices of the original polytopes.  As such, they are amenable to standard algorithmic solutions.
Finally, we prove that the complexity of these games for the aforementioned objectives remain  in $\NP \cap \coNP$, that is,  they stay in the standard complexity class of simple stochastic games, even when polytopal games support for  an uncountable number of actions for the players and the discretization may grow exponentially.

\paragraph{Related work.}
Definitions of infinite stochastic games do exist (see, for instance, \cite{Kucera11}) though they are of discrete nature, contrary to the type of games presented here.  In fact,  PSGs are related to IDTMCs~\cite{JonssonL91,KozineU02,SenVA06}, IMDPs~\cite{SenVA06}, and Convex MDPs~\cite{DBLP:conf/cav/PuggelliLSS13},  but they are variants of MDPs and hence they are $1\frac{1}{2}$-games.  In particular, PSGs adopt a semantics similar to IMDPs and Convex MDPs~\cite{DBLP:conf/cav/PuggelliLSS13} in which the uncertainty introduced by the polytope is interpreted as an uncountable non-deterministic branching.
While  in \cite{DBLP:conf/cav/PuggelliLSS13}  interior-point algorithms are used to solve Convex MDPs, we use a discretization through the vertices of polytopes to solve PSGs.  Though this has an exponential impact, this is very mild in practice as we will show later.
A much simpler variant of PSG was used in~\cite{CastroDPD23} to provide an algorithmic solution for a fault tolerant measure.  This incipient idea served as the starting point for the generalization presented here.

Somewhat related are the stochastic timed games (STGs)~\cite{BouyerF09,AkshayBKMT16}.  However, the continuous non-determinism introduced by the time in STGs is resolved by uniform and exponential distributions and the remnant non-determinism (resolved by the strategies) is still discrete.  This does not make these models simpler since undecidability has been shown for games with at least 3 clocks~\cite{BouyerF09}.

\paragraph{Outline of the paper.}
Section~\ref{sec:roborta} presents a motivating example.  Section~\ref{sec:preliminaries} introduces the background needed for tackling the rest of the paper.  The definition of PSGs, their semantics and basic properties are given in Section~\ref{sec:polytopal-games}.  The main results are presented in Section~\ref{sec:discretazition}.  
%Finally,  we describe some experiments performed with a prototype tool  that extends {\PRISMGAMES} to support polytopal stochastic games, and draw some conclusions. 
Full proofs are gathered in the Appendix.

\section{Roborta vs. Rigoborto in the land of uncertainties} \label{sec:roborta}
%
% This is a Tikz draw for an example of grid

%
%
\begin{wrapfigure}[15]{r}{.40\textwidth}
\vspace{-.8cm}
\resizebox{.4\textwidth}{!}{\input{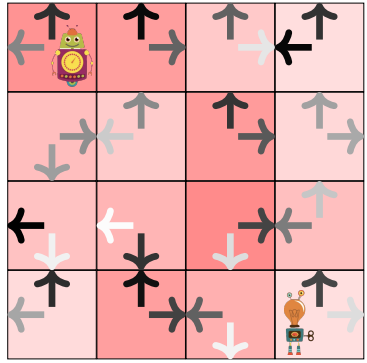}\unskip}
\vspace*{-.5cm}
\caption{An example of a grid for the Roborta vs Rigoborto game.}\label{fig:rob-vs-rig}
\end{wrapfigure}     
We illustrate our approach by means of a simple example.  Consider a field represented as a bidimensional grid and two robots --which we call Roborta and Rigoborto-- that navigate it.   Roborta can move sideways and forward,  Rigoborto can move sideways and backward. The robots start at a certain initial  position.  Roborta intends to reach the end of the grid, i.e., she wants to reach position $(i,n+1)$ for any $i$, whereas Rigoborto wants to stop Roborta.  He can achieve this by reaching Roborta's location.
The robots play in turns.
%% The robots play in turns (i.e., it is a turn-based game).
%
The objective of Roborta is to maximize the probability of reaching the exit, while the objective of Rigoborto is to minimize this value.
We spice up this example by considering the terrain quality (which depends on factors like,  e.g., stones, mud or grass) and slope,  which may cause imprecisions and uncertainties in the robots mobility, probably making them slide towards some undesired direction.  The terrain quality and slope may vary in each grid position.  In Fig.~\ref{fig:rob-vs-rig}, we show an example of such a scenario.  Therein, the robots start at the corners,  the arrows indicate the slopes in the terrain, and the colors in the cells indicate the terrain quality.  Darker arrows correspond to sharper slopes.  Similarly, cells with lower quality are colored with stronger red colors.

More precisely, for each $(x,y)$-cell, the terrain quality $q_{xy}\in[0,0.5]$ gives the uncertainty factor, where $q_{xy}=0$ means that probabilities are completely determined, and, as $q_{xy}$ grows, the probability values become increasingly fuzzier.   In addition, we consider two factors associated with the terrain slopes: $l_{xy},f_{xy}\in[-1,1]$, representing the inclination of the lateral and frontal slopes respectively.  Thus, as $l_{xy}$ get closer to $1$ ($-1$), the likelihood of shifting to the right (left) increases, with $l_{xy}=0$ not favouring any particular side.  Similarly $f_{xy}>1$ ($f_{xy}<-1$) biases the robot towards the front (back).
Let $p_c$ be the probability that the robot command is successful (that is, that it moves in the intended direction), and let $p_l$, $p_r$, $p_f$, and $p_b$ be the probabilities that the command is unsuccessful and the robot uncontrollably slides respectively to the left, right, front, and back.  Then, the space of all probability values can be defined by the following set of inequalities:
\begin{align*}
  1   & \ = \    p_c + p_l + p_r + p_f + p_b \\
  p_c & \ \geq \ 0,\ \ \ p_l \ \geq \ 0,\ \ \ p_r \ \geq \ 0,\ \ \ p_f \ \geq \ 0,\ \ \ p_b \ \geq \ 0 \\
  p_c & \ \leq \ 1-(q_{xy}+{\textstyle\frac{1}{2}}\cdot(1-(1-|l_{xy}|) \cdot (1-|f_{xy}|))) \\
  0   & \ \leq \ (1-\max(0,-l_{xy})) \cdot p_l-(1-q_{xy}) \cdot (1-\max(0,l_{xy})) \cdot p_r \\
  0   & \ \leq \ (1-\max(0,l_{xy})) \cdot p_r-(1-q_{xy}) \cdot (1-\max(0,-l_{xy})) \cdot p_l \\
  0   & \ \leq \ (1-\max(0,f_{xy})) \cdot p_f-(1-q_{xy}) \cdot (1-\max(0,-f_{xy})) \cdot p_b \\
  0   & \ \leq \ (1-\max(0,-f_{xy})) \cdot p_b-(1-q_{xy}) \cdot (1-\max(0,f_{xy})) \cdot p_f
\end{align*}
Note that if $q_{xy}=0$, the system has a unique solution. If, in addition, $l_{xy}>0$, $1/(1-l_{xy})=p_r/p_l$ giving the likelihood ratio of sliding towards the right.

\begin{figure}[t]
%\vspace*{-0.5cm}
\begin{subfigure}{\textwidth}
\lstset{style=polyprism}
\begin{lstlisting}{language=Java}
// action specification for Roborta moving to the left
[robl] (turn = 0) & (roby<L) & !Collision -> (rob_mov'=1) & (turn'=1)$\medskip$
[robl-cont] (turn = 1) & (rob_mov = 1) ->
  //The first four probabilistic options  correspond to environments setbacks
  pl : (robx'=max(0,robx-1)) & (rob_mov'=0) + pr : (robx'=min(W-1,robx+1)) & (rob_mov'=0) 
  + pf : (roby'=roby+1) & (rob_mov'=0) + pb : (roby'=max(0,roby+1)) & (rob_mov'=0)
  + pc : (robx'=max(0,robx-1)) & (rob_mov'=0) 
  {// inequations for uncertainty
    1-(Q[robx,roby]+(1-(1-abs(L[robx,roby]))*(1-abs(F[robx,roby])))/2) >= pc,
    (1-max(0,-L[robx,roby]))*pl - (1-Q[robx,roby])*(1-max(0,L[robx,roby]))*pr >= 0,
    (1-max(0,L[robx,roby]))*pr - (1-Q[robx,roby])*(1-max(0,-L[robx,roby]))*pl >= 0,
    (1-max(0,F[robx,roby]))*pf - (1-Q[robx,roby])*(1-max(0,-F[robx,roby]))*pb >= 0,
    (1-max(0,-F[robx,roby]))*pb - (1-Q[robx,roby])*(1-max(0,F[robx,roby]))*pf >= 0  
  };
\end{lstlisting}
\vspace*{-0.3cm}
\caption{Roborta moves left}\label{fig:roborta-code}
\end{subfigure}

\smallskip

\begin{subfigure}{\textwidth}
\lstset{style=polyprism}
\begin{lstlisting}{language=Java}
// action specification for Rigoborto moving to the left
[rigl] (turn = 1) & (rob_mov = 0) & (rigy<L) & !Collision ->
  //The first four probabilistic options  correspond to environments setbacks
  pl : (rigx'=max(0,rigx-1)) & (turn'=0) & (Collision'=(robx=rigx && roby=rigy))
  + pr : (rigx'=min(W-1,rigx+1)) & (turn'=0) & (Collision'=(robx=rigx && roby=rigy))
  + pf : (rigy'=rigy+1) & (turn'=0) & (Collision'=(robx=rigx && roby=rigy))
  + pb : (rigy'=max(0,rigy+1)) & (turn'=0) & (Collision'=(robx=rigx && roby=rigy))
  + pc : (rigx'=max(0,rigx-1)) & (turn'=0) & (Collision'=(robx=rigx && roby=rigy)) 
  {// inequations for uncertainty
    1-(Q[rigx,rigy]+(1-(1-abs(L[rigx,rigy]))*(1-abs(F[rigx,rigy])))/2) >= pc,
    (1-max(0,-L[rigx,rigy]))*pl - (1-Q[rigx,rigy])*(1-max(0,L[rigx,rigy]))*pr >= 0,
    (1-max(0,L[rigx,rigy]))*pr - (1-Q[rigx,rigy])*(1-max(0,-L[rigx,rigy]))*pl >= 0,
    (1-max(0,F[rigx,rigy]))*pf - (1-Q[rigx,rigy])*(1-max(0,-F[rigx,rigy]))*pb >= 0,
    (1-max(0,-F[rigx,rigy]))*pb - (1-Q[rigx,rigy])*(1-max(0,F[rigx,rigy]))*pf >= 0  
  };
\end{lstlisting}
\vspace*{-0.3cm}
\caption{Rigoborto moves left}\label{fig:rigoborto-code}
\end{subfigure}
%\vspace*{-0.3cm}
\caption{Fragment of code for Roborta vs Rigoborto}\label{fig:robvsrig-code}
%\vspace*{-0.2cm}
\end{figure}

Our aim is to find the best strategy for Roborta to win against all
odds.  This implies that the terrain uncertainty behaves adversarially
to Roborta but favourably to Rigoborto.  Thus,  in our model,  
Rigoborto controls the non-determinism introduced by the terrain
uncertainty.
Assuming an extension of the \PRISMGAMES language, the code could
look like in Fig.~\ref{fig:robvsrig-code}, where
subfigures~\ref{fig:roborta-code} and~\ref{fig:rigoborto-code} show
the decisions to move left by Roborta and Rigoborto respectively.

Variable \verb"turn" indicates who is the next player to move (with
\verb"0" for Roborta and \verb"1" for Rigoborto).  If it is Roborta's
turn (see first line in Fig.~\ref{fig:roborta-code}) and she decides
to move left, she indicates it by setting \verb"rob_mov'=1" (\verb"1"
indicates a left move while \verb"2", \verb"3", and \verb"4" are used
for the other directions, and \verb"0" to indicate that Roborta is not
moving).  At the same time, she yields her turn by setting
\verb"turn'=1".
Notice that the action is not yet complete: the reaction of the
terrain to the move is encoded in the next line (action
\verb"robl-cont" in Fig.~\ref{fig:roborta-code}).  Notice that this
action happens in a state in which \verb"turn=1", making the terrain
uncertainty --defined by the polytope-- adversarial to Roborta.
Here, variables \verb"robx" and \verb"roby" correspond to Roborta's
coordinates $x$ and $y$ and constant matrices \verb"Q", \verb"L", and
\verb"F" contain the respective values for $q_{xy}$, $l_{xy}$, and
$f_{xy}$. The rest of the variables are as expected.  Once this step is
taken, variable \verb"rob_mov" is set to \verb"0", thus enabling
Rigoborto's move.
Rigoborto's decision to move left is given in
Fig~\ref{fig:rigoborto-code}.  Notice that this is performed in a
single action since we assume that the terrain uncertainty plays in
favour of him.  Something particular to this transition is the setting
of variable \verb"Collision" to indicate whether Rigoborto has
caught Roborta.

\section{Preliminaries}\label{sec:preliminaries}
In this section we introduce notation and basic concepts of polytopes and games. Interested readers are referred to \cite{Ziegler95,Kucera11}.

In the following $\powerset(S)$ denotes the powerset of set $S$, and $\powerset_f(S)$ denotes the set of finite subsets of set $S$.
A \emph{convex polytope} in $\Reals^n$ is a bounded set
$K=\{\bvec{x}\in\Reals^n \mid A\bvec{x}\leq b\}$, with
$A\in\Reals^{m\times n}$ and $b\in\Reals^m$, for some $m\in\Nat$.  By
\emph{bounded} we mean (in our case) that there exists $M\in\Realsnn$ such that
$\sum_{i=1}^n\abs{x_i}\leq M$ for all $\bvec{x}\in K$ ($x_i$ denotes
the $i$th element of $\bvec{x}$).
%% \remarkPRD{In fact, more generally it is requested that
%%   $\norm{\bvec{x}}\leq M$ for all $\bvec{x}\in K$ and some given norm
%%   $\norm{\_}$.}
%
Let $S$ be a finite set. As functions in $\Reals^S$ can be
equivalently seen as vectors in $\Reals^{|S|}$, we will in general
refer to polytopes in $\Reals^S$.  Let $\Poly(S)$ be the set of all
convex polytopes in $\Reals^S$.
Notice that the set of all probability functions on $S$ form the
convex polytope
$\Dist(S)=\{\mu\in\Reals^S\mid \sum_{s\in S}\mu(s)=1 \text{ and } \forall {s\in S} \colon {\mu(s)\geq 0} \}$.
Let $\DPoly(S)=\{K\cap\Dist(S)\mid K\in\Poly(S)\}$.  Thus,
$K\in\DPoly(S)$ is a convex polytope whose elements are also
probability functions on $S$ and therefore its defining set of
inequality $A\bvec{x}\leq b$ already encodes the inequalities
$\sum_{s\in S}x_s=1$ and $x_s\geq 0$ for $s\in S$.

Any convex polytope $K\in\Poly(S)$ can alternatively be characterized
as the convex hull of its finite set of vertices.   Let $\vertices(K)$ denote the set of all
vertices of polytope $K$.
If $\vertices(K)=\{\bvec{v}^1,\ldots,\bvec{v}^k\}$, then every
$\bvec{x}\in K$ is a convex combination of
$\{\bvec{v}^1,\ldots,\bvec{v}^k\}$, that is,
$\bvec{x} = \sum_{i=1}^k\lambda_i\bvec{v}^k$ with $\lambda_i\geq 0$,
for $i\in[1..k]$, and $\sum_{i=1}^k\lambda_i=1$.
A \emph{simplex} is any convex polytope $K\in\Poly(S)$ whose set of
vertices $\vertices(K)$ is affinely independent, that is, for any
family $\{\lambda_{\bvec{v}}\in\Reals\}_{\bvec{v}\in\vertices(K)}$
such that $\sum_{\bvec{v}\in\vertices(K)}\lambda_{\bvec{v}}=0$,
$\sum_{\bvec{v}\in\vertices(K)}\lambda_{\bvec{v}}\bvec{v}=0$ implies
that $\lambda_{\bvec{v}}=0$ for all $\bvec{v}\in\vertices(K)$.
This implies that for every $\bvec{x}\in K$, with $K$ being a simplex,
the convex combination $\bvec{x} = \sum_{i=1}^k\lambda_i\bvec{v}^k$ is
unique.
We also remark that any convex polytope $K$ can be expressed as the
union of a (finite) set of simplices $\{K_i\}_{i\in I}$ so that
$\vertices(K)=\bigcup_{i\in I}\vertices(K_i)$ (this is a consequence
of Charath\'eodory's Theorem~\cite{Ziegler95,McMullen2020}).  We will
call such decomposition a \emph{vertex-preserving triangulation}.
% The last claim is spelled out as Theorem 1.11 in "disgeoIII_notes.pdf"
% and corresponds to Theorem 2.15(4) in Ziegler95.
% I like Theorem 1B4 of Caratheodory's theorem in McMullen2020.
%
Let $\Simp(S)$ denote the set of all simplices in $\Reals^S$ and
$\DSimp(S)=\Simp(S)\cap\DPoly(S)$.

A \emph{stochastic game} \cite{Shapley53,Condon92,FilarV96}
is a tuple
$\StochG = (\sgnodes, (\sgnmax, \sgnmin), \sgactions, \sgtrans)$,
where $\sgnodes$ is a finite set of \emph{states} with
$\sgnmax,\sgnmin\subseteq\sgnodes$ being a partition of $\sgnodes$,
$\sgactions$ is a (finite) set of \emph{actions},
and
$\sgtrans : \sgnodes \times \sgactions \times \sgnodes \rightarrow [0,1]$
is a \emph{probabilistic transition function} such that for every
$s\in\sgnodes$ and $a\in\sgactions$,
$\sgtrans(s,a,\cdot)\in\Dist(\sgnodes)$ or $\sgtrans(s,a,\sgnodes)=0$.
Let
$\enabled(s)=\{{a\in\sgactions}\mid{\sgtrans(s,a,\sgnodes)=1}\}$ be
the set of actions enabled at state $s$.
If $\sgnmax=\emptyset$ or $\sgnmin=\emptyset$,
then $\StochG$ is a \emph{Markov decision process} (or MDP).
If, in addition, $|\enabled(s)|=1$ for all $s\in\sgnodes$,
$\StochG$ is a \emph{Markov chain} (or MC).
A \emph{path} in the game $\StochG$ is an infinite sequence of states
$\rho=s_0 s_1 \dots$ such that,  for every $k\in\Nat$,  there is an $a\in\sgactions$ with $\sgtrans(s_k, a, s_{k+1})>0$.  For $i\geq0$, $\rho_i$
indicates the $i$th state in the path $\rho$ (notice that $\rho_0$ is
the first state in $\rho$). $\GamePaths_{\StochG}$ denotes the set of
all paths, and $\FinGamePaths_{\StochG}$ denotes the set of finite
prefixes of paths.  Similarly, $\GamePaths_{\StochG,s}$ and
$\FinGamePaths_{\StochG,s}$ denote the set of paths and the set of
finite paths starting at state $s$.
A \emph{strategy} for the $i$-player (for $i\in\{\maxplay,\minplay\}$)
in a game $\StochG$ is a function
$\strat_{i}:{\sgnodes^*\sgnodes_i}\to{\Dist(\sgactions)}$ that assigns
a probabilistic distribution to each finite sequence of states such
that $\strat_{i}(\hat{\rho}s)(a) > 0$ only if $a\in \enabled(s)$.  The
set of all strategies for the $i$-player is named
$\Strategies{i}$. Whenever convenient, we indicate that the set of
strategies $\Strategies{i}$ belongs to the game $\StochG$ by writing
by $\Strategies{\StochG,i}$
A strategy $\strat_{i}$ is said to be \emph{pure} or
\emph{deterministic} if, for every
$\hat{\rho}s\in\sgnodes^*\sgnodes_i$, $\strat_{i}(\hat{\rho} s)$ is a
Dirac distribution (that is a distribution $\Dirac_a$ s.t.,
$\Dirac_a(a)=1$ and $\Dirac_a(b)=0$ for all $b\neq a$), and it is
called \emph{memoryless} if $\strat_{i}(\hat{\rho} s) =
\strat_{i}(s)$, for every $\hat{\rho}\in\sgnodes^*$.
%% Let
%% $\MemorylessStrats{i}$ and $\DetStrats{i}$ be respectively the set of
%% all memoryless strategies and the set of all deterministic strategies
%% for the $i$-player.  $\DetMemorylessStrats{i} = \MemorylessStrats{i}
%% \cap \DetStrats{i}$ is the set of all its deterministic and memoryless
%% strategies.
Let $\MemorylessStrats{i}$ be the set of all memoryless strategies for
the $i$-player and $\DetMemorylessStrats{i}$ be the set of all its
deterministic and memoryless strategies.  Note that the definition of strategy given above works for  set of actions that are finite,  in Section \ref{sec:polytopal-games} we define strategies for uncountable sets of actions.

Given strategies $\strat_{\maxplay} \in \Strategies{\maxplay}$ and
$\strat_{\minplay} \in \Strategies{\minplay}$, and an initial state
$s$, the \emph{result} of the game is a Markov chain
\cite{ChatterjeeH12}, denoted
$\StochG^{\strat_{\maxplay},\strat_{\minplay}}_s$.
The Markov chain $\StochG^{\strat_{\maxplay},\strat_{\minplay}}_s$
defines a probability measure
$\Prob^{\strat_{\maxplay},\strat_{\minplay}}_{\StochG,s}$ on the Borel
$\sigma$-algebra generated by the cylinders of $\GamePaths_{\StochG,s}$.
If $\xi$ is a measurable set in such a Borel $\sigma$-algebra,
$\Prob^{\strat_{\maxplay},\strat_{\minplay}}_{\StochG,s}(\xi)$ is the
probability that strategies $\strat_{\maxplay}$ and
$\strat_{\minplay}$ follow a path in $\xi$ starting from state $s$.
We use {\LTL} notation to represent specific set of paths, in particular,
$D \Until^n C =
\{{\rho \in \sgnodes^\omega} \mid {{{\rho_n\in C} \wedge {\forall{j<n}\colon{\rho_j\in D}}}}\} =
D^n\times C\times \sgnodes^\omega$
is the set of paths that reach $C\subseteq\sgnodes$ in exactly $n\geq0$ steps
traversing before only states in $D\subseteq\sgnodes$;
$\Finally^n C = \sgnodes\Until^n C$ is the set of all paths reaching
states in $C$ in exactly $n$ steps; and
$\Finally C = \bigcup_{n\geq0}(\sgnodes\setminus C)\Until^nC$
is the set of all paths that reach a state in $C$.

%% A stochastic game is said to be \emph{almost surely
%% stopping}~\cite{Condon92} if for all pair of memoryless and deterministic strategies\remarkPRD{por qu\'e solo memoryless y deterministic?}
%% $\strat_{\maxplay}$, $\strat_{\minplay}$ the probability of reaching a
%% terminal state is~$1$.
%% %
%% A state $s$ is \emph{terminal} if $\sgtrans(s,a,s)=1$, for all
%% $a\in\enabled(s)$.
%% %\DetMemorylessStrats{i}
%% In other words, a game is stopping if
%% $\inf_{\strat_\minplay\in\DetMemorylessStrats{\minplay}}\inf_{\strat_\maxplay\in\DetMemorylessStrats{\maxplay}}\Prob^{\strat_{\maxplay},\strat_{\minplay}}_{s}(\Finally T)=1$,
%% where $T\subseteq\sgnodes$ is the set of terminal states.
%% %
%% A stochastic game is \emph{irreducible}~\cite{FilarV96} if for all
%% pair of memoryless and deterministic strategies, the probability of reaching a state from any other
%% state is positive, that is, if
%% $\inf_{\strat_\minplay\in\DetMemorylessStrats{\minplay}}\inf_{\strat_\maxplay\in\DetMemorylessStrats{\maxplay}}\Prob^{\strat_{\maxplay},\strat_{\minplay}}_{s}(\Finally s')>0$
%% for all pair of states $s,s'\in\sgnodes$.
%
%
A stochastic game is said to be \emph{almost surely
stopping}~\cite{Condon92,FilarV96} if for all pair of
%memoryless
strategies
$\strat_{\maxplay}$, $\strat_{\minplay}$ the probability of reaching a
terminal state is~$1$.
A state $s$ is \emph{terminal} if $\sgtrans(s,a,s)=1$, for all
$a\in\enabled(s)$.
%\MemorylessStrats{i}
In other words, a game is stopping if
$\inf_{\strat_\minplay\in\Strategies{\minplay}}\inf_{\strat_\maxplay\in\Strategies{\maxplay}}\Prob^{\strat_{\maxplay},\strat_{\minplay}}_{s}(\Finally T)=1$,
%$\inf_{\strat_\minplay\in\MemorylessStrats{\minplay}}\inf_{\strat_\maxplay\in\MemorylessStrats{\maxplay}}\Prob^{\strat_{\maxplay},\strat_{\minplay}}_{s}(\Finally T)=1$,
where $T\subseteq\sgnodes$ is the set of terminal states.
A stochastic game is \emph{irreducible}~\cite{FilarV96} if for all
pair of
%memoryless
strategies,
%% strategies%
%% \footnote{Normally the definition of irreducible stochastic games is
%% limited to memoryless and deterministic strategies (which coincide
%% with the general definition for finite models).  Since we are
%% introducing a variant with uncountably many transitions, we opt for
%% the general case.},
%
the probability of reaching a state from any other state is positive,
that is, if
$\inf_{\strat_\minplay\in\Strategies{\minplay}}\inf_{\strat_\maxplay\in\Strategies{\maxplay}}\Prob^{\strat_{\maxplay},\strat_{\minplay}}_{s}(\Finally
s')>0$
%% $\inf_{\strat_\minplay\in\MemorylessStrats{\minplay}}\inf_{\strat_\maxplay\in\MemorylessStrats{\maxplay}}\Prob^{\strat_{\maxplay},\strat_{\minplay}}_{s}(\Finally s')>0$
for all pair of states $s,s'\in\sgnodes$.

A \emph{quantitative objective} or \emph{payoff function} is a
measurable function $f: \sgnodes^{\omega} \to \Reals$.  Let
$\Expect^{\strat_{\maxplay},\strat_{\minplay}}_{\StochG,s}[f]$ be the
expectation of measurable function $f$ under probability
$\Prob^{\strat_{\maxplay},\strat_{\minplay}}_{\StochG,s}$.
The goal of the $\maxplay$-player is to maximize this value whereas the
goal of the $\minplay$-player is to minimize it.  Sometimes quantitative
objective functions can be defined via \emph{rewards}. These are
assigned by a \emph{reward function} $\reward:S \to \Reals^+$.  We
usually consider stochastic games augmented with a reward
function.  Moreover, we assume that for every terminal state $s$,
$\reward(s) = 0$.
The value of the game for the $\maxplay$-player at state $s$ under
strategy $\strat_{\maxplay}$ is defined as the infimum over all the
values resulting from the $\minplay$-player strategies in that state,
i.e.,
$\inf_{\strat_{\minplay} \in \Strategies{\minplay}} \Expect^{\strat_{\maxplay},\strat_{\minplay}}_{\StochG,s}[f]$.
The \emph{value of the game} for the $\maxplay$-player is defined as the
supremum of the values of all the $\maxplay$-player strategies, i.e.,
$\sup_{\strat_{\maxplay} \in \Strategies{\maxplay}} \inf_{\strat_{\minplay} \in \Strategies{\minplay}} \Expect^{\strat_{\maxplay},\strat_{\minplay}}_{\StochG,s}[f]$.
Similarly, the value of the game for the $\minplay$-player under
strategy $\strat_{\minplay}$ and the value of the game for the
$\minplay$-player are defined as
$\sup_{\strat_{\maxplay} \in \Strategies{\maxplay}}  \Expect^{\strat_{\maxplay},\strat_{\minplay}}_{\StochG,s}[f]$
and
$\inf_{\strat_{\minplay} \in \Strategies{\minplay}} \sup_{\strat_{\maxplay} \in \Strategies{\maxplay}}  \Expect^{\strat_{\maxplay},\strat_{\minplay}}_{\StochG,s}[f]$,
respectively.
We say that a game is \emph{determined} if both values are the same,
that is,
$\sup_{\strat_{\maxplay} \in \Strategies{\maxplay}} \inf_{\strat_{\minplay} \in \Strategies{\minplay}} \Expect^{\strat_{\maxplay},\strat_{\minplay}}_{\StochG,s}[f]
=
\inf_{\strat_{\minplay} \in \Strategies{\minplay}} \sup_{\strat_{\maxplay} \in \Strategies{\maxplay}} \Expect^{\strat_{\maxplay},\strat_{\minplay}}_{\StochG,s}[f]$.
%
%% Martin \cite{Martin98} proved the determinacy of stochastic games for
%% Borel and bounded objective functions.

% The following definitions are taken from Puterman 1994, Sec. 5.1,
% p. 120-121
%
In this paper we focus on \emph{total accumulated reward}, where the
payoff function is defined by
$\TRewards(\rho)=\lim_{n\to\infty}\sum^n_{i=0} \reward(\rho_i)$,
\emph{total discounted reward}, defined by
$\DRewards{\discfactor}(\rho)=\lim_{n\to\infty}\sum^n_{i=0}\discfactor^i\reward(\rho_i)$,
where $\discfactor\in(0,1)$ is the discount factor,
and \emph{average reward}, defined by
$\ARewards(\rho)=\lim_{n\to\infty}\frac{1}{n+1}\sum^n_{i=0}\reward(\rho_i)$.
By taking, respectively, $\fgen(i,n)=1$, $\fgen(i,n)=\discfactor^i$, or
$\fgen(i,n)=\frac{1}{n+1}$, we refer simultaneously to the above
payoff functions with the single function
$\GRewards(\rho)=\lim_{n\to\infty}\sum^n_{i=0}\fgen(i,n)\reward(\rho_i)$.

We also focus on \emph{reachability objective}.  In this case, the
goal of the $\maxplay$-player is to maximize the probability of reaching a
state on a goal set $\goal\subseteq\sgnodes$ whereas the goal of the
$\minplay$-player is to minimize it.  Therefore, similar to quantitative
objectives, the \emph{value of the reachability game for the
$\maxplay$-player} is defined by
$\sup_{\strat_{\maxplay} \in \Strategies{\maxplay}} \inf_{\strat_{\minplay} \in \Strategies{\minplay}} \Prob^{\strat_{\maxplay},\strat_{\minplay}}_{\StochG,s}(\Finally\goal)$
and the \emph{value of the reachability game for the $\minplay$-player}
is defined by
$\inf_{\strat_{\minplay} \in \Strategies{\minplay}} \sup_{\strat_{\maxplay} \in \Strategies{\maxplay}}  \Prob^{\strat_{\maxplay},\strat_{\minplay}}_{\StochG,s}(\Finally\goal)$,
and the game is \emph{determined} if both values are the same.

\section{Polytopal Stochastic Games} \label{sec:polytopal-games}

A polytopal stochastic game is characterized through a structure that
contains a finite set of states divided into two sets, each owned by
a different player.  In addition, each state has assigned a finite
set of convex polytopes of probability distributions over states.
The formal definition is as follows.
\begin{definition}\label{def:psg}
  A \emph{polytopal stochastic game}
  %% (PSG, for short\footnote{The
  %% choice of the acronym is just coincidental.  No association should
  %% be made to any football club that may have not apropriately welcomed
  %% the greatest football player of all time.})
  (PSG, for short)
  is a structure
  $\StochK=(\psgnodes,(\psgnmax,\psgnmin),\psgtrans)$ such that
  $\psgnodes$ is a finite set of states partitioned into 
  $\psgnodes=\psgnmax\uplus\psgnmin$ and
  $\psgtrans:\psgnodes\to\powersetf(\DPoly(\psgnodes))$.
  If, in particular,
  $\psgtrans:\psgnodes\to\powersetf(\DSimp(\psgnodes))$, we call
  $\StochK$ a \emph{simplicial stochastic game} (SSG for short).
\end{definition}

The idea of a PSG is as expected: in a state $s\in\psgnodes_i$
($i\in\{\maxplay,\minplay\}$), player $i$ chooses to play a polytope
$K\in\psgtrans(s)$ and a distribution $\mu\in K$.  The next state $s'$
is sampled according to distribution $\mu$ and the game continues from
$s'$ repeating the same process.

%% \begin{example}
%%   One can imagine a stochastic game variant of interval Markov
%%   decision processes (IMDP)~\cite{JonssonL91,KozineU02}.  In this
%%   case, every polytope $K\in\psgtrans(s)$, for all $s\in\psgnodes$, is
%%   defined by $\mu\in K$ iff $\sum_{s'\in\psgnodes}\mu(s')=1$ and, for
%%   all $s'\in\psgnodes$ and some fixed $0\leq l_{s'}\leq u_{s'}\leq 1$,
%%   $l_{s'}\leq\mu(s')\leq u_{s'}$ (note that the intervals need to be
%%   closed).
%% \end{example}

As a particular example, one can devise a stochastic game variant of
Interval Markov Decision Processes (IMDPs)~\cite{JonssonL91,KozineU02}.
This type of games can be interpreted as a PSG where every polytope
$K\in\psgtrans(s)$, for all $s\in\psgnodes$, is defined by $\mu\in K$
iff $\sum_{s'\in\psgnodes}\mu(s')=1$ and, for all $s'\in\psgnodes$ and
some fixed $0\leq l_{s'}\leq u_{s'}\leq 1$,
$l_{s'}\leq\mu(s')\leq u_{s'}$ (note that the intervals need to be closed).

The behaviour of a polytopal stochastic game is formally interpreted
in terms of a stochastic game where the number of transitions outgoing
the players' states may be uncountably large.
We choose a controllable view on the uncertainty introduced by the polytope since the adversarial alternative can be encoded as was shown in Sec.~\ref{sec:roborta}.
Formally, the
interpretation of a PSG is as follows.
\begin{definition}\label{def:interpretation}
  The \emph{interpretation of the polytopal stochastic game $\StochK$}
  is defined by the stochastic game
  $\StochGK = (\sgnodes, (\sgnmax,\sgnmin), \sgactions, \sgtrans)$,
  where
  $\sgactions=\left(\bigcup_{s\in\sgnodes}\psgtrans(s)\right)\times\Dist(\sgnodes)$
  and
  \[\sgtrans(s,(K,\mu),s')=
    \begin{cases}
      \mu(s') & \text{if } K\in\psgtrans(s) \text{ and } \mu \in K \\
      0 & \text{otherwise}
    \end{cases}
  \]
\end{definition}

Notice that the set of actions $\sgactions$ can be uncountably large,
as well as each set $\enabled(s)=\bigcup_{K\in\psgtrans(s)}\{K\}{\times}K$.
%as well as each set $\enabled(s)$, $s\in\sgnodes$.
Therefore we need to extend the strategies to this uncountable domain
which should be properly endowed with a $\sigma$-algebra.
For this we make use of a standard construction to provide a
$\sigma$-algebra to $\Dist(\sgnodes)$~\cite{Giry82}:
$\Salg_{\Dist(\sgnodes)}$ is defined as the smallest $\sigma$-algebra
containing the sets $\{\mu\in\Dist(\sgnodes)\mid\mu(S)\geq p\}$ for
all $S\subseteq\sgnodes$ and $p\in[0,1]$.
Now, we endow $\sgactions$ with the product $\sigma$-algebra
$\Salg_\sgactions=\powerset\left(\bigcup_{s\in\sgnodes}\psgtrans(s)\right)\otimes\Salg_{\Dist(\sgnodes)}$
(i.e., the smallest $\sigma$-algebra containing all rectangles
$\boldsymbol{K}\times \boldsymbol{M}$ with
$\boldsymbol{K}\subseteq\bigcup_{s\in\sgnodes}\psgtrans(s)$ and
$\boldsymbol{M}\in\Salg_{\Dist(\sgnodes)}$) and let $\PMeas(\sgactions)$ be the set
of all probability measures on $\Salg_\sgactions$.  It is not
difficult to check that each set of enabled actions $\enabled(s)$ is
measurable (i.e., $\enabled(s)\in\Salg_\sgactions$) and that function
$\sgtrans(s,\cdot,s')$ is measurable (i.e.,
$\{{a\in\sgactions}\mid{\sgtrans(s,a,s')\leq p}\}\in\Salg_\sgactions$
for all $p\in[0,1]$).

We extend the definition of \emph{strategy} for the $i$-player
($i\in\{\maxplay,\minplay\}$) in $\StochGK$ to be a function
$\strat_{i}:{\sgnodes^*\sgnodes_i}\to{\PMeas(\sgactions)}$ that
assigns a probability measure to each finite sequence of states such
that $\strat_{i}(\hat{\rho}s)(\enabled(s)) = 1$.  All other concepts
on strategies defined in~Sec.~\ref{sec:preliminaries} apply to this
new definition as well.

In the following we present the formal definition of
$\Prob^{\strat_{\maxplay},\strat_{\minplay}}_{\StochGK,s}$.
First, for each $n\geq 0$ and $s\in \sgnodes$, define 
$\Prob^{\strat_{\maxplay},\strat_{\minplay},n}_{\StochGK,s}:\sgnodes^{n+1}\to[0,1]$
for all $s'\in \sgnodes$ and $\hat{\rho}\in \sgnodes^{n+1}$ inductively as follows:
\begin{align*}
  &\Prob^{\strat_{\maxplay},\strat_{\minplay},0}_{\StochGK,s}(s') =  \Dirac_{s}(s') \\
  &\Prob^{\strat_{\maxplay},\strat_{\minplay},n+1}_{\StochGK,s}(\hat{\rho} s') =% {} \\
  %&\qquad
  \begin{dcases}
    \Prob^{\strat_{\maxplay},\strat_{\minplay},n}_{\StochGK,s}(\hat{\rho})\int_{\sgactions}\sgtrans(\last(\hat{\rho}),\cdot,s')\ \diff(\strat_{\maxplay}(\hat{\rho})(\cdot)) & \text{if } \last(\hat{\rho})\in \sgnmax \\%[1ex]
    \Prob^{\strat_{\maxplay},\strat_{\minplay},n}_{\StochGK,s}(\hat{\rho})\int_{\sgactions}\sgtrans(\last(\hat{\rho}),\cdot,s')\ \diff(\strat_{\minplay}(\hat{\rho})(\cdot)) & \text{if } \last(\hat{\rho})\in \sgnmin
  \end{dcases}
\end{align*}
%
%% \begin{align*}
%%   &\Prob^{\strat_{\maxplay},\strat_{\minplay},0}_{\StochGK,s}(s') =  \Dirac_{s}(s') \\
%%   &\Prob^{\strat_{\maxplay},\strat_{\minplay},n+1}_{\StochGK,s}(\hat{\rho} s') = 
%%    \Prob^{\strat_{\maxplay},\strat_{\minplay},n}_{\StochGK,s}(\hat{\rho})\int_{\sgactions}\sgtrans(\last(\hat{\rho}),\cdot,s')\ \diff(\strat_{i}(\hat{\rho})(\cdot))\\
%%   & & \hspace{-5em}\text{if } \last(\hat{\rho})\in \sgnodes_i \text{ with } i\in\{\maxplay,\minplay\}
%% \end{align*}
%
and extend $\Prob^{\strat_{\maxplay},\strat_{\minplay},n}_{\StochGK,s}:\powerset(\sgnodes^{n+1})\to[0,1]$ to sets as the sum of the points.

Let $\Salg_{\sgnodes}$ denote the discrete $\sigma$-algebra on $\sgnodes$
and $\Salg_{\sgnodes^\omega}$ the usual product $\sigma$-algebra on
$\sgnodes^\omega$.
By Carath\'eodory extension theorem~\cite{AshDoleans99},
$\Prob^{\strat_{\maxplay},\strat_{\minplay}}_{\StochGK,s}:\Salg_{\sgnodes^\omega}\to[0,1]$
is defined as the unique probability measure such that for all
$n\geq 0$, and $S_i\in \Salg_{\sgnodes}$, $0\leq i\leq n$,
\[\Prob^{\strat_{\maxplay},\strat_{\minplay}}_{\StochGK,s}(S_0\times\cdots\times S_n\times\sgnodes^\omega) = \Prob^{\strat_{\maxplay},\strat_{\minplay},n}_{\StochGK,s}(S_0\times\cdots\times S_n)\]

%\medskip

The notions of \emph{deterministic} and \emph{memoryless} extends
directly to this type of strategies.
In addition, a strategy $\strat_{i}$, $i\in\{\maxplay,\minplay\}$, is
\emph{semi-Markov} if for every $\hat{\rho},\hat{\rho}'\in\sgnodes^*$
and $s\in \sgnodes_i$, $|\hat{\rho}|=|\hat{\rho}'|$ implies
$\strat_{i}(\hat{\rho}s)=\strat_{i}(\hat{\rho}'s)$, that is, the
decisions of $\strat_{i}$ depend only on the length of the run and its
last state. Thus, we write $\strat_{i}(n,s)$ instead of
$\strat_{i}(\hat{\rho}s)$ whenever $|\hat{\rho}|=n$.  Let
$\SemiMarkovStrats{i}$ denote the set of all semi-Markov strategies
for the $i$-player.
Also, we say that a strategy $\strat_{i}\in\Strategies{i}$ is
\emph{extreme} if for all $\hat{\rho}\in\sgnodes^*$,
$\strat_{i}(\hat{\rho}s)(\{(K,\mu)\in\sgactions(s)\mid\mu\in\vertices(K)\})=1$.
Notice that extreme strategies only selects transitions on vertices
of polytopes.
Let $\XSemiMarkovStrats{i}$ and $\XDetMemorylessStrats{i}$ be,
respectively, the set of all extreme semi-Markov strategies and the
set of all extreme deterministic and memoryless srategies for the
$i$-player.

Polytopal stochastic games can be translated into
simplicial stochastic games preserving all the stochastic behaviour.
More precisely, for every PSG $\StochK$ there is a SSG $\StochK'$ such
that for every pair of strategies for $\StochK$ in a particular class
(i.e., memoryless, semi-Markov, etc.), there is a pair of strategies
for $\StochK'$ in the same class that yields the same probability
measure and vice versa.
Let $\Triang\colon\DPoly\to\powerset(\DSimp)$ be a function that assigns a
vertex-preserving triangulation $\Triang(K)$ to each polytope $K$. 
Then:

\begin{proposition}\label{prop:PSG:SSG}
  Let $\StochK=(\psgnodes,(\psgnmax,\psgnmin),\psgtrans)$ be a PSG and
  define the SSG $\StochK'=(\psgnodes,(\psgnmax,\psgnmin),\psgtrans')$
  such that $\psgtrans'(s) = \bigcup_{K\in\psgtrans(s)}\Triang(K)$.
  Let $\StochGK$ and $\StochGKp$ be their respective interpretations.
  Then,
  \begin{enumerate}
  \item\label{item:PSG:SSG:i}%
    for all pair of strategies $\strat_{\maxplay}$ and
    $\strat_{\minplay}$ for $\StochGK$ there is a pair of strategies
    $\strat'_{\maxplay}$ and $\strat'_{\minplay}$ for $\StochGKp$ such
    that
    \begin{enumerate*}
    \item%
      $\Prob^{\strat_{\maxplay},\strat_{\minplay}}_{\StochGK,s}=\Prob^{\strat'_{\maxplay},\strat'_{\minplay}}_{\StochGKp,s}$
      for all $s\in\sgnodes$, and
    \item%
      if $\strat_{i}$, $i\in\{\maxplay,\minplay\}$, is memoryless (resp.
      deterministic, semi-Markov or extreme) then so is $\strat'_{i}$;
    \end{enumerate*}
    and
  \item\label{item:PSG:SSG:ii}%
    the same holds with the roles of $\StochGK$ and $\StochGKp$
    exchanged.
  \end{enumerate}  
\end{proposition}
\begin{proof}[Sketch]
  Let $\StochGK=(\sgnodes,(\sgnmax,\sgnmin),\sgactions,\sgtrans)$
%%   with
%%   $\sgactions=\left(\bigcup_{s\in\sgnodes}\psgtrans(s)\right)\times\Dist(\sgnodes)$
  and
  $\StochGKp=(\sgnodes,(\sgnmax,\sgnmin),\sgactions',\sgtrans')$.
%%   with
%%   $\sgactions'=\left(\bigcup_{s\in\sgnodes}\psgtrans'(s)\right)\times\Dist(\sgnodes)=\left(\bigcup_{\substack{s\in\sgnodes,~\\K\in\psgtrans(s)}}\Triang(K)\right)\times\Dist(\sgnodes)$.
%
  To prove item~\ref{item:PSG:SSG:i}, the new strategies are defined
  so that they preserve the same measure on the probability part of
  the labels in $\sgactions'$ as the one the old strategies measure on
  the probability part of $\sgactions$ while properly distributing the
  probabilities on the simplices of the triangulation of the original
  polytopes.
  For this, first fix a function $f_K:\Triang(K)\to\powerset(K)$ for
  each polytope $K\in\DPoly(\sgnodes)$ satisfying
  \begin{enumerate*}[(i)]
  \item%
    $\forall {K'\in\Triang(K)}\colon{f_K(K')\subseteq K'}$,
  \item%
    $\bigcup_{K'\in\Triang(K)}f_K(K') = K$, and
  \item%
    $\forall {K'_1,K'_2\in\Triang(K)}\colon {{f_K(K'_1)\cap f_K(K'_2)\neq\emptyset} \limp {K'_1=K'_2}}$.
  \end{enumerate*}
  Thus, $f_K(K')$ is almost the simplex $K'$ but ensuring that
  distributions on the faces of $K'$ are exactly in one of the
  $f_K(K'')$, $K''\in\Triang(K)$.
  
  Given strategies $\strat_{i}$, $i\in\{\maxplay,\minplay\}$,
  for $\StochGK$ define $\strat'_{i}$ for $\StochGKp$, for all
  $\hat{\rho}\in\sgnodes^*$, $s\in\sgnodes_{i}$, and
  $A'\in\Salg_{\sgactions'}$ by
  \begin{equation}\label{eq:def:stratp:ssg:main}\textstyle
    \strat'_{i}(\hat{\rho}s)(A') = \sum_{K\in\psgtrans(s)} \sum_{K'\in\Triang(K)} \strat_{i}(\hat{\rho}s)(\{K\}\times({A'\sect{K'}}\cap f_K(K')))
  \end{equation}
  where ${A'\sect{K'}}=\{\mu\mid (K',\mu)\in A'\}$ is the $K'$ section
  of the measurable set $A'$.
  Notice that $f_K$ ensures that the faces of each $K'\in\Triang(K)$
  are considered in exactly one summand of the inner summation
  of~(\ref{eq:def:stratp:ssg:main}).
  %
%%   Thus, $\strat'_{i}(\hat{\rho})$ is a well defined probability
%%   measure and hence $\strat'_{\maxplay}$ and $\strat'_{\minplay}$ is a
%%   pair of strategies for $\StochGKp$.

  For item~\ref{item:PSG:SSG:ii},
  %% as for~\ref{item:PSG:SSG:i},
  the new strategies preserve the same measure on the probability part
  of $\sgactions$ as the old strategies while gathering the
  probability of the simplices in the original polytope.
  So, for each state $s\in\sgnodes$, fix
  %% $f_s:\psgtrans(s)\to\bigcup_{K\in\psgtrans(s)}\Triang(K)$
  $f_s:{\psgtrans(s)\to\powerset(\DSimp(\sgnodes)})$
  such that
  \begin{enumerate*}[(i)]
  \item%
    $\forall {K\in\psgtrans(s)}\colon {f_s(K)\subseteq\Triang(K)}$,
  \item%
    $\bigcup_{K\in\psgtrans(s)}f_s(K)=\bigcup_{K\in\psgtrans(s)}\Triang(K)$, and
  \item%
    $\forall {K_1,K_2\in\psgtrans(s)}\colon {{f_s(K_1)\cap f_s(K_2)\neq\emptyset} \limp {K_1=K_2}}$.
  \end{enumerate*}
  Given strategies $\strat'_{i}$, $i\in\{\maxplay,\minplay\}$, for
  $\StochGKp$ define $\strat_{i}$ for $\StochGK$, for all
  $\hat{\rho}\in\sgnodes^*$, $s\in\sgnodes_{i}$, and
  $A\in\Salg_{\sgactions}$ by
  \begin{equation}\label{eq:def:stratp:psg:main}\textstyle
    \strat_{i}(\hat{\rho}s)(A) = \sum_{K\in\psgtrans(s)} \sum_{K'\in f_{s}(K)} \strat'_{i}(\hat{\rho}s)(\{K'\}\times{A\sect{K}})
  \end{equation}
  Notice that, by definition, $K'\in\psgtrans'(s)$.
  Moreover, notice that $f_s$ ensures that a simplex in a
  triangulation of a polytope outgoing $s$ is considered in exactly one
  summand of~(\ref{eq:def:stratp:psg:main}).
  %
%%   Thus $\strat_{i}(\hat{\rho})$ is a well defined probability measure
%%   on $\sgnodes$ and hence $\strat_{\maxplay}$ and $\strat_{\minplay}$
%%   is a pair of strategies for $\StochGK$.

  In both cases, it requires some straightforward calculations to
  check that the properties of memoryless, semi-Markov, deterministic,
  and extreme are preserved by the new strategies. Also in both cases,
  to prove that
  $\Prob^{\strat_{\maxplay},\strat_{\minplay}}_{\StochGK,s}=\Prob^{\strat'_{\maxplay},\strat'_{\minplay}}_{\StochGKp,s}$
  it sufficies to state that
  $\Prob^{\strat_{\maxplay},\strat_{\minplay},n}_{\StochGK,s}=\Prob^{\strat'_{\maxplay},\strat'_{\minplay},n}_{\StochGKp,s}$
  for all $n\geq0$ which is done by induction using results from
  measure theory.
  \qed
\end{proof}

\section{Discretizing Polytopal Stochastic Games} \label{sec:discretazition}

In this section we show that a PSG can be solved by translating it
into a finite stochastic game that is just like the original PSG but
it only has the transitions corresponding to the vertices of the
polytopes.  We focus on reachability games, and the reward games
introduced above: total accumulated reward, total discounted reward,
and average reward.

The first lemma we introduce states that the calculation of the
expected values of the different reward games only depend on the
probability of reaching each state and the reward collected in each
state regardless the path that lead to such states.  In particular,
Lemma~\ref{lm:encoding:reach:and:expectation}.\ref{lm:encoding:reach:and:expectation:i}
refers to the reward collected in a finite number of steps while
Lemma~\ref{lm:encoding:reach:and:expectation}.\ref{lm:encoding:reach:and:expectation:ii}
refers to the general case stated before.

For $k\geq 0$ define
$\Finally^{k} s = \sgnodes^k\times\{s\}\times\sgnodes^\omega$
to be the set of all runs in which $s\in\sgnodes$ is reached in
exactly $k$ steps.
Let
$\FGRewards^n(\hat{\rho})=\sum^{n}_{i=0}\fgen(i,n)\reward(\hat{\rho}_i)$
for all $\hat{\rho}\in\sgnodes^{n+1}$.
Then $\GRewards(\rho)=\lim_{n\to\infty}\FGRewards^n(\rho[..n+1])$ where
$\rho[..n+1]$ is the $(n+1)$th prefix of $\rho$, i.e.,
$\rho[..n+1]=\rho_0\rho_1\rho_2...\rho_n$.

\begin{lemma}\label{lm:encoding:reach:and:expectation}
  Let $\StochGK$ be a stochastic game resulting from interpreting a
  PSG $\StochK$.  For all strategies
  $\strat_{\maxplay}\in\Strategies{\maxplay}$ and
  $\strat_{\minplay}\in\Strategies{\minplay}$,
  \begin{enumerate}
  \item\label{lm:encoding:reach:and:expectation:i}%
    $%\displaystyle
    \sum_{\hat{\rho} \in \sgnodes^{n+1}} \Prob^{\strat_{\maxplay},\strat_{\minplay},n}_{\StochGK,s}(\hat{\rho})\mult\FGRewards^n(\hat{\rho}) =
    \sum^{n}_{i=0} \sum_{s' \in \sgnodes} \Prob^{\strat_{\maxplay},\strat_{\minplay}}_{\StochGK,s}(\Finally^i s')\mult\fgen(i,n)\mult\reward(s')$, for all $n\geq 0$,
    and\smallskip
  \item\label{lm:encoding:reach:and:expectation:ii}%
    $%\displaystyle
    \Expect^{\strat_{\maxplay},\strat_{\minplay}}_{\StochGK,s}[\GRewards] =
    \lim_{n\to\infty} \sum^{n}_{i=0} \sum_{s' \in \sgnodes} \Prob^{\strat_{\maxplay},\strat_{\minplay}}_{\StochGK,s}(\Finally^i s')\mult\fgen(i,n)\mult\reward(s')$.    
  \end{enumerate}
\end{lemma}

The proof of
Lemma~\ref{lm:encoding:reach:and:expectation}.\ref{lm:encoding:reach:and:expectation:i}
follows by induction on $n$ while
Lemma~\ref{lm:encoding:reach:and:expectation}.\ref{lm:encoding:reach:and:expectation:ii}
can be calculated using the first item.

The next lemma states that if the $\minplay$-player plays a
semi-Markov strategy, the $\maxplay$-player can achieve equal results
whether she plays an arbitrary strategy or limits to playing only
semi-Markov strategies.
%% The first two items of the lemma focus on
%% reachability while the second one on the expected value of the
%% different reward games.
%
\begin{lemma}\label{lm:semimarkov}
  Let $\StochGK$ be a stochastic game resulting from interpreting a
  PSG $\StochK$.
  If $\strat_{\minplay} \in \SemiMarkovStrats{\minplay}$
  is a semi-Markov strategy,
  then, for any $\strat_{\maxplay} \in \Strategies{\maxplay}$,
  there is a semi-Markov strategy
  $\starredstrat_{\maxplay} \in \SemiMarkovStrats{\maxplay}$
  such that:
  \begin{enumerate}
  \item\label{lm:semimarkov:i}%
    $\Prob^{\strat_{\maxplay},\strat_{\minplay}}_{\StochGK,s}(D \Until^n s') =
    \Prob^{\starredstrat_{\maxplay},\strat_{\minplay}}_{\StochGK,s}(D \Until^n s')$,
    for all $n\geq 0$, $D\subseteq\sgnodes$ and $s'\in \sgnodes$;
  \item\label{lm:semimarkov:ii}%
    $\Prob^{\strat_{\maxplay},\strat_{\minplay}}_{\StochGK,s}(\Finally C) =
    \Prob^{\starredstrat_{\maxplay},\strat_{\minplay}}_{\StochGK,s}(\Finally C)$,
    for all $C\subseteq\sgnodes$; and
  \item\label{lm:semimarkov:iii}%
    $\Expect^{\strat_{\maxplay},\strat_{\minplay}}_{\StochGK,s}[\GRewards] =
    \Expect^{\starredstrat_{\maxplay},\strat_{\minplay}}_{\StochGK,s}[\GRewards]$.
  \end{enumerate}
  Similarly, if $\strat_{\maxplay} \in \SemiMarkovStrats{\maxplay}$
  % is a semi-Markov strategy,
  then, for any $\strat_{\minplay} \in \Strategies{\minplay}$,
  there exists % a semi-Markov strategy
  $\starredstrat_{\minplay} \in \SemiMarkovStrats{\minplay}$
  satisfying, mutatis mutandis, the same equalities.
\end{lemma}
\begin{proof}[Sketch]
  To prove item~\ref{lm:semimarkov:i}, we define the new strategy
  $\starredstrat_{\maxplay}$ so that the probability of choosing from
  $A\in\Salg_\sgactions$ after a path of length $n$ ending on a state
  $s$ with the original strategy is uniformly distributed among the
  paths of this type in the new strategy.
  %% To prove item~\ref{lm:semimarkov:i}, we define the new strategy
  %% $\starredstrat_{\maxplay}$ so that the probability of choosing
  %% from $A\in\Salg_\sgactions$ after finite paths of equal lengths, say
  %% $n$, ending on the same state $s$ uniformly distributes the total
  %% probability of chosing from $A$ induced by the original strategy
  %% after $n$ steps.
  %
  Thus, $\starredstrat_{\maxplay}$ is formally defined as follows.
  For $\hat{\rho}\in\sgnodes^*$, $s'\in\sgnodes$, and $A\in\Salg_\sgactions$,
  such that
  $\Prob^{\strat_{\maxplay},\strat_{\minplay}}_{\StochGK,s}(D \Until^n s') > 0$
  and $|\hat{\rho}| = n\geq 0$, let
  \begin{align*}
  \starredstrat_{\maxplay}(\hat{\rho}s')(A)
  & = 
  \frac{\sum_{\hat{\rho}'\in D^n}\Prob^{\strat_{\maxplay},\strat_{\minplay},n}_{\StochGK,s}(\hat{\rho}'s')\mult\strat_{\maxplay}(\hat{\rho}' s')(A)}{\Prob^{\strat_{\maxplay},\strat_{\minplay}}_{\StochGK,s}(D \Until^n s')}
  \end{align*}
  For $s'\in\sgnodes$ with
  $\Prob^{\strat_{\maxplay},\strat_{\minplay}}_{\StochGK,s}(D \Until^n s') = 0$
  and $|\hat{\rho}s'| = n$, define
  $\starredstrat_{\maxplay}(\hat{\rho}s')$ to be
  $\Dirac_{\textsf{f}(s')}$ for a globally fixed function $\textsf{f}$
  such that $\textsf{f}(s')\in\enabled(s')$.
  Notice that
  $\starredstrat_{\maxplay}\in\SemiMarkovStrats{\maxplay}$.
%%   Therefore, we write $\starredstrat_{\maxplay}(n,s')$ for
%%   $\starredstrat_{\maxplay}(\hat{\rho}s')$ whenever
%%   $|\hat{\rho}|=n$.

  Then, the proof of item~\ref{lm:semimarkov:i} follows by induction
  with particular care in the case of
  $\Prob^{\strat_{\maxplay},\strat_{\minplay}}_{\StochGK,s}(D \Until^n s') = 0$.
  Item~\ref{lm:semimarkov:ii} follows straightforwardly from
  item~\ref{lm:semimarkov:i} and item~\ref{lm:semimarkov:iii} follows
  directly from item~\ref{lm:semimarkov:ii} using
  Lemma~\ref{lm:encoding:reach:and:expectation}.\ref{lm:encoding:reach:and:expectation:ii}.
  The proof can be replicated mutatis mutandi with
  $\maxplay$ and $\minplay$ exchanged yielding the last part of the
  lemma.
  \qed
\end{proof}

Since $\psgtrans(s)$ is finite,  there can be finitely many
polytopes $K$ such that $(K,\mu)\in\enabled(s)$.  Besides, the set of
vertices $\vertices(K)$ of $K$ is finite.  Therefore the set
$\{(K,\mu)\in\sgactions(s)\mid\mu\in\vertices(K)\}$ is also finite
and, as a consequence, extreme strategies only resolve with discrete
(finite) probability distributions.  That is, if $\strat_{i}$ is
extreme, $\strat_{i}(\hat{\rho}s)$ has finite support for all
$\hat{\rho}\in\sgnodes^*$ and $s\in\sgnodes$.

It turns out that Lemma~\ref{lm:semimarkov} can be strengthened to
obtain \emph{extreme} semi-Markov strategies.
%% Thus if one of the players plays a semi-Markov strategy, the other
%% can achieve equal results to the general case if she limits herself
%% to playing only extreme semi-Markov strategies, i.e., she only
%% chooses actions $(K,\mu)$ such that the distribution $\mu$ is a
%% vertex of the simplex $K$.
%
We first prove this new lemma for simplicial stochastic games since
simplices have the particular property that any vector in a simplex can
be uniquely defined as a convex combination of the simplex vertices
which is crucial for the proof of the lemma.

\begin{lemma}\label{lm:xsemimarkov}
  Let $\StochGK$ be a stochastic game resulting from interpreting a
  SSG $\StochK$.
  If $\strat_{\minplay} \in \SemiMarkovStrats{\minplay}$
  is a semi-Markov strategy,
  then, for any $\strat_{\maxplay} \in \SemiMarkovStrats{\maxplay}$,
  there is an extreme semi-Markov strategy
  $\starredstrat_{\maxplay} \in \XSemiMarkovStrats{\maxplay}$
  such that:
  \begin{enumerate}
  \item\label{lm:xsemimarkov:i}%
    $\Prob^{\strat_{\maxplay},\strat_{\minplay}}_{\StochGK,s}(D \Until^n s') =
    \Prob^{\starredstrat_{\maxplay},\strat_{\minplay}}_{\StochGK,s}(D \Until^n s')$,
    for all $n\geq 0$, $D\subseteq\sgnodes$ and $s'\in \sgnodes$;
  \item\label{lm:xsemimarkov:ii}%
    $\Prob^{\strat_{\maxplay},\strat_{\minplay}}_{\StochGK,s}(\Finally C) =
    \Prob^{\starredstrat_{\maxplay},\strat_{\minplay}}_{\StochGK,s}(\Finally C)$,
    for all $C\subseteq\sgnodes$; and
  \item\label{lm:xsemimarkov:iii}%
    $\Expect^{\strat_{\maxplay},\strat_{\minplay}}_{\StochGK,s}[\GRewards] =
    \Expect^{\starredstrat_{\maxplay},\strat_{\minplay}}_{\StochGK,s}[\GRewards]$.
  \end{enumerate}
  Similarly, if $\strat_{\maxplay} \in \SemiMarkovStrats{\maxplay}$
  then, for any $\strat_{\minplay} \in \SemiMarkovStrats{\minplay}$,
  there exists
  $\starredstrat_{\minplay} \in \XSemiMarkovStrats{\minplay}$
  satisfying, mutatis mutandis, the same equalities.
\end{lemma}
\begin{proof}[Sketch]
  \newcommand{\convp}{\mathsf{p}}%
  For any $K\in\DSimp(\sgnodes)$, $\mu\in K$ and
  $\hat{\mu}\in\vertices(K)$ define $\convp^K(\mu,\hat{\mu})\in[0,1]$
  such that
  $\sum_{\hat{\mu}\in\vertices(K)}\convp^K(\mu,\hat{\mu})\mult\hat{\mu}=\mu$.
  That is, all $\convp^K(\mu,\hat{\mu})$, $\hat{\mu}\in\vertices(K)$,
  are the unique factors that define the convex combination for $\mu$
  in the simplex $K$.
%%   Therefore, $\convp^K(\mu,\hat{\mu})$ is well
%%   defined for all $K\in\DSimp(\sgnodes)$, $\mu\in K$ and
%%   $\hat{\mu}\in\vertices(K)$.
  In any other case, let $\convp^K(\mu,\hat{\mu})=0$.

  Let $\convp((K,\mu),(K,\hat{\mu}))=\convp^K(\mu,\hat{\mu})$ for all
  $K\in\DSimp(\sgnodes)$, $\mu\in K$ and $\hat{\mu}\in\vertices(K)$,
  and let $\convp(a,b)=0$ for any other $a,b\in\sgactions$.
  For every $(K,\mu)\in\sgactions$ such that $\mu\in K$, let
  $\vertices(K,\mu)=\{(K,\hat{\mu})\mid\hat{\mu}\in\vertices(K)\}$ and
  let $\vertices(K,\mu)=\emptyset$ otherwise.
  Thus, for every $s\in\sgnodes$ and $a\in\sgactions$,
%%   \begin{equation}\label{eq:xsemimarkov:sgtrans:convp:main}
%%     \sgtrans(s,a,\cdot) =
%%     \sum_{b\in\vertices(a)}\convp(a,b)\mult\sgtrans(s,b,\cdot).
%%   \end{equation}
  $\sgtrans(s,a,\cdot) =
  \sum_{b\in\vertices(a)}\convp(a,b)\mult\sgtrans(s,b,\cdot)$.

  We also extend $\convp$ to measurable sets $B\in\Salg_\sgactions$
  and $a\in\sgactions$ by
  $\convp(a,B)=\sum_{b\in B\cap\vertices(a)}\convp(a,b)$.

  For every $\hat{\rho}\in\sgnodes^*$, $s'\in\sgnodes$ and
  $B\in\Salg_\sgactions$, define $\starredstrat_\maxplay$ by
  \[\starredstrat_\maxplay(\hat{\rho}s')(B) =
  \int_{\sgactions} \convp(\cdot,B)\ \diff(\strat_\maxplay(\hat{\rho}s')(\cdot)).\]
  $\starredstrat_\maxplay(\hat{\rho}s')$ is defined so that it assigns
  to each vertex of a simplex the weighted contribution (according to
  $\strat_\maxplay(\hat{\rho}s')$) of each distribution (in the said
  simplex) to such vertex.

  Because $\strat_\maxplay$ is semi-Markov, so is
  $\starredstrat_\maxplay$.  Moreover, notice that if $b$ is not a
  vertex label, then $\convp(a,b)=0$ (and hence $\convp(a,B)>0$ only
  if $B$ contains vertices).  This should hint that
  $\starredstrat_\maxplay$ is also extreme.

  Item~\ref{lm:xsemimarkov:i} proceeds by induction on $n$.
  Item~\ref{lm:xsemimarkov:ii} follows straightforwardly
  using~\ref{lm:xsemimarkov:i}, and item~\ref{lm:xsemimarkov:iii}
  follows from~item~\ref{lm:xsemimarkov:ii} using
  Lemma~\ref{lm:encoding:reach:and:expectation}.\ref{lm:encoding:reach:and:expectation:ii}.
  The proof can be replicated mutatis mutandi with
  $\maxplay$ and $\minplay$ exchanged which yields the last part of the
  lemma.
  \qed
\end{proof}

Because of Proposition~\ref{prop:PSG:SSG}, Lemma~\ref{lm:xsemimarkov}
extends immediately to PSG. Moreover, by applying
Lemma~\ref{lm:xsemimarkov} twice and Proposition~\ref{prop:PSG:SSG},
we have the next corollary.

\begin{corollary}\label{cor:xsemimarkov}
  Let $\StochGK$ be a stochastic game resulting from interpreting a
  PSG $\StochK$.
  For all semi-Markov strategies
  $\strat_{\minplay} \in \SemiMarkovStrats{\minplay}$ and
  $\strat_{\maxplay} \in \SemiMarkovStrats{\maxplay}$,
  there are extreme semi-Markov strategies
  $\starredstrat_{\minplay} \in \XSemiMarkovStrats{\minplay}$ and
  $\starredstrat_{\maxplay} \in \XSemiMarkovStrats{\maxplay}$
  such that
  \begin{enumerate}
  \item\label{cor:xsemimarkov:i}%
    $\Prob^{\strat_{\maxplay},\strat_{\minplay}}_{\StochGK,s}(\Finally C) =
    \Prob^{\starredstrat_{\maxplay},\starredstrat_{\minplay}}_{\StochGK,s}(\Finally C)$,
    for all $C\subseteq\sgnodes$; and
  \item\label{cor:xsemimarkov:ii}%
    $\Expect^{\strat_{\maxplay},\strat_{\minplay}}_{\StochGK,s}[\GRewards] =
    \Expect^{\starredstrat_{\maxplay},\starredstrat_{\minplay}}_{\StochGK,s}[\GRewards]$.
  \end{enumerate}
\end{corollary}

Given $\StochGK$, define the \emph{extreme interpretation of
$\StochK$} as the stochastic game
$\StochHK = (\sgnodes, (\sgnmax,\sgnmin), \vertices(\sgactions), \sgtrans_\StochHK)$
where $\sgtrans_\StochHK$ is the restriction of
$\sgtrans$ to actions in
$\vertices(\sgactions) = \{{(K,\mu)\in\sgactions}\mid{\mu\in\vertices(K)}\}$,
that is,
$\sgtrans_\StochHK(s,a,s)=\sgtrans(s,a,s')$ for all $s,s'\in\sgnodes$
and $a\in\vertices(\sgactions)$.
Since $\vertices(\sgactions)$ is finite, $\StochHK$ is a finite
stochastic game.

Given an extreme semi-Markov strategy
$\strat_i\in\XSemiMarkovStrats{\StochGK,i}$ for the $i$-player in the 
stochastic game $\StochGK$, $i\in\{\maxplay,\minplay\}$, define
$\stratv_i(\hat{\rho}s)(A)=\strat_i(\hat{\rho}s)(A)$ for all
$\hat{\rho}\in\sgnodes^*$, $s\in\sgnodes$, and
$A\subseteq\vertices(\sgactions)$ ($A\in\Salg_\sgactions$ since it is
finite).
Notice that
$\stratv_i(\hat{\rho}s)(\enabled_\StochHK(s))=\strat_i(\hat{\rho}s)(\vertices(\enabled(s)))=1$.
Therefore $\stratv_i\in\SemiMarkovStrats{\StochHK,i}$ is a semi-Markov
strategy in $\StochHK$.
Conversely, for a semi-Markov strategy
$\strat_i\in\SemiMarkovStrats{\StochHK,i}$ for the $i$-player in the
stochastic game $\StochHK$, define
$\stratx_i(\hat{\rho}s)(A)=\strat_i(\hat{\rho}s)(A\cap\vertices(\sgactions))$
for all $\hat{\rho}\in\sgnodes^*$, $s\in\sgnodes$, and
$A\in\Salg_\sgactions$.
$\stratx_i\in\XSemiMarkovStrats{\StochGK,i}$ is a well defined extreme
semi-Markov strategy in $\StochGK$ since
$\stratx_i(\hat{\rho}s)(\vertices(\enabled(s)))=\strat_i(\hat{\rho}s)(\enabled_\StochHK(s))=1$
and
$\stratx_i(\hat{\rho}s)(\sgactions\setminus\vertices(\sgactions))=\strat_i(\hat{\rho}s)(\emptyset)=0$.
Then, it can be calculated by induction on $n$ that
$\Prob^{\strat_{\maxplay},\strat_{\minplay},n}_{\StochGK,s} = \Prob^{\stratv_{\maxplay},\stratv_{\minplay},n}_{\StochHK,s}$
and
$\Prob^{\stratx_{\maxplay},\stratx_{\minplay},n}_{\StochGK,s} = \Prob^{\strat_{\maxplay},\strat_{\minplay},n}_{\StochHK,s}$
which yield
\begin{equation}\label{eq:StochGK:StochHK}
  \Prob^{\strat_{\maxplay},\strat_{\minplay}}_{\StochGK,s} = \Prob^{\stratv_{\maxplay},\stratv_{\minplay}}_{\StochHK,s}
  \text{ and }
  \Prob^{\stratx_{\maxplay},\stratx_{\minplay}}_{\StochGK,s} = \Prob^{\strat_{\maxplay},\strat_{\minplay}}_{\StochHK,s}.
\end{equation}
This suggests that the solution of a PSG under extreme semi-Markov
strategies is equivalent to the solution the game on its extreme
interpretation limited to semi-Markov strategies, which is stated in
the following:

\begin{proposition}\label{prop:infsup:supinf:StochGK:StochHK}
  Let $\StochGK$ and $\StochHK$ be respectively the interpretation and
  the extreme interpretation of $\StochK$. Then, the following
  equalities hold
  \begin{enumerate}
  \item\label{prop:infsup:supinf:StochGK:StochHK:i}%
    $
    \inf_{\strat_\minplay\in\XSemiMarkovStrats{\StochGK,\minplay}}\!\sup_{\strat_\maxplay\in\XSemiMarkovStrats{\StochGK,\maxplay}}\!\Prob^{\strat_{\maxplay},\strat_{\minplay}}_{\StochGK,s}(\Finally C) =
    \inf_{\strat_\minplay\in\SemiMarkovStrats{\StochHK,\minplay}}\!\sup_{\strat_\maxplay\in\SemiMarkovStrats{\StochHK,\maxplay}}\!\Prob^{\strat_{\maxplay},\strat_{\minplay}}_{\StochHK,s}(\Finally C)$
  \item\label{prop:infsup:supinf:StochGK:StochHK:ii}%
    $
    \sup_{\strat_\maxplay\in\XSemiMarkovStrats{\StochGK,\maxplay}}\!\inf_{\strat_\minplay\in\XSemiMarkovStrats{\StochGK,\minplay}}\!\Prob^{\strat_{\maxplay},\strat_{\minplay}}_{\StochGK,s}(\Finally C) =
    \sup_{\strat_\maxplay\in\SemiMarkovStrats{\StochHK,\maxplay}}\!\inf_{\strat_\minplay\in\SemiMarkovStrats{\StochHK,\minplay}}\!\Prob^{\strat_{\maxplay},\strat_{\minplay}}_{\StochHK,s}(\Finally C)$
  \item\label{prop:infsup:supinf:StochGK:StochHK:iii}%
    $
    \inf_{\strat_\minplay\in\XSemiMarkovStrats{\StochGK,\minplay}}\!\sup_{\strat_\maxplay\in\XSemiMarkovStrats{\StochGK,\maxplay}}\!\Expect^{\strat_{\maxplay},\strat_{\minplay}}_{\StochGK,s}[\GRewards] =
    \inf_{\strat_\minplay\in\SemiMarkovStrats{\StochHK,\minplay}}\!\sup_{\strat_\maxplay\in\SemiMarkovStrats{\StochHK,\maxplay}}\!\Expect^{\stratv_{\maxplay},\stratv_{\minplay}}_{\StochHK,s}[\GRewards]$
  \item\label{prop:infsup:supinf:StochGK:StochHK:iv}%
    $
    \sup_{\strat_\maxplay\in\XSemiMarkovStrats{\StochGK,\maxplay}}\!\inf_{\strat_\minplay\in\XSemiMarkovStrats{\StochGK,\minplay}}\!\Expect^{\strat_{\maxplay},\strat_{\minplay}}_{\StochGK,s}[\GRewards] =
    \sup_{\strat_\maxplay\in\SemiMarkovStrats{\StochHK,\maxplay}}\!\inf_{\strat_\minplay\in\SemiMarkovStrats{\StochHK,\minplay}}\!\Expect^{\stratv_{\maxplay},\stratv_{\minplay}}_{\StochHK,s}[\GRewards]$
  \end{enumerate}
\end{proposition}

The next proposition, whose proof also uses (\ref{eq:StochGK:StochHK}),
%% Proposition~\ref{prop:StochGK:StochHK},
provides necessary conditions
for the polytopal stochastic game to be almost surely stopping or
irreducible in terms of the extreme interpretation.

\begin{proposition}\label{prop:stopping:irreducible:StochGK:StochHK}
  Let $\StochGK$ and $\StochHK$ be respectively the interpretation and
  the extreme interpretation of $\StochK$.
  Then,
  \begin{enumerate*}[(1)]
  \item\label{prop:stopping:irreducible:StochGK:StochHK:i}%
    if $\StochGK$ is almost surely stopping, so is $\StochHK$, and
  \item\label{prop:stopping:irreducible:StochGK:StochHK:ii}%
    if $\StochGK$ is irreducible, so is $\StochHK$.
  \end{enumerate*}
  %% Then, the following two
  %% statements hold:
  %% \begin{enumerate}
  %% \item\label{prop:stopping:irreducible:StochGK:StochHK:i}%
  %%   if $\StochGK$ is almost surely stopping, so is $\StochHK$;
  %% \item\label{prop:stopping:irreducible:StochGK:StochHK:ii}%
  %%   if $\StochGK$ is irreducible, so is $\StochHK$.
  %% \end{enumerate}
\end{proposition}

Notice that by fixing one strategy in $\StochHK$ to be the memoryless,
the remaining structure is a Markov decision process.  Then the
statements in the following proposition are consequences of standard
results in MDP~\cite{Puterman94}.

\begin{proposition}\label{prop:mdp:results}
  For all
  $\starredstrat_\maxplay\in\DetMemorylessStrats{\StochHK,\maxplay}$ and
  $\starredstrat_\minplay\in\DetMemorylessStrats{\StochHK,\minplay}$,
  \begin{enumerate}
  \item\label{prop:mdp:results:i}%
    $%\displaystyle
%%     \sup_{\strat_\maxplay\in\Strategies{\StochHK,\maxplay}}\Prob^{\strat_{\maxplay},\starredstrat_{\minplay}}_{\StochHK,s}(\Finally C)
%%     =
    \sup_{\strat_\maxplay\in\SemiMarkovStrats{\StochHK,\maxplay}}\Prob^{\strat_{\maxplay},\starredstrat_{\minplay}}_{\StochHK,s}(\Finally C)
    =
    \sup_{\strat_\maxplay\in\DetMemorylessStrats{\StochHK,\maxplay}}\Prob^{\strat_{\maxplay},\starredstrat_{\minplay}}_{\StochHK,s}(\Finally C)$;
  \item\label{prop:mdp:results:ii}%
    $%\displaystyle
%%     \inf_{\strat_\minplay\in\Strategies{\StochHK,\minplay}}\Prob^{\starredstrat_{\maxplay},\strat_{\minplay}}_{\StochHK,s}(\Finally C)
%%     =
    \inf_{\strat_\minplay\in\SemiMarkovStrats{\StochHK,\minplay}}\Prob^{\starredstrat_{\maxplay},\strat_{\minplay}}_{\StochHK,s}(\Finally C)
    =
    \inf_{\strat_\minplay\in\DetMemorylessStrats{\StochHK,\minplay}}\Prob^{\starredstrat_{\maxplay},\strat_{\minplay}}_{\StochHK,s}(\Finally C)$;
  \item\label{prop:mdp:results:iii}%
    $%\displaystyle
%%     \sup_{\strat_\maxplay\in\Strategies{\StochHK,\maxplay}}\Expect^{\strat_{\maxplay},\starredstrat_{\minplay}}_{\StochHK,s}(\GRewards)
%%     =
    \sup_{\strat_\maxplay\in\SemiMarkovStrats{\StochHK,\maxplay}}\Expect^{\strat_{\maxplay},\starredstrat_{\minplay}}_{\StochHK,s}(\GRewards)
    =
    \sup_{\strat_\maxplay\in\DetMemorylessStrats{\StochHK,\maxplay}}\Expect^{\strat_{\maxplay},\starredstrat_{\minplay}}_{\StochHK,s}(\GRewards)$,
    provided
    $\Expect^{\strat_{\maxplay},\starredstrat_{\minplay}}_{\StochHK,s}(\GRewards)$
    is defined for all
    $\strat_\maxplay\in\SemiMarkovStrats{\StochHK,\maxplay}$; and
  \item\label{prop:mdp:results:iv}%
    $%\displaystyle
%%     \inf_{\strat_\minplay\in\Strategies{\StochHK,\minplay}}\Expect^{\starredstrat_{\maxplay},\strat_{\minplay}}_{\StochHK,s}(\GRewards)
%%     =
    \inf_{\strat_\minplay\in\SemiMarkovStrats{\StochHK,\minplay}}\Expect^{\starredstrat_{\maxplay},\strat_{\minplay}}_{\StochHK,s}(\GRewards)
    =
    \inf_{\strat_\minplay\in\DetMemorylessStrats{\StochHK,\minplay}}\Expect^{\starredstrat_{\maxplay},\strat_{\minplay}}_{\StochHK,s}(\GRewards)$,
    provided
    $\Expect^{\starredstrat_{\maxplay},\strat_{\minplay}}_{\StochHK,s}(\GRewards)$
    is defined for all
    $\strat_\minplay\in\SemiMarkovStrats{\StochHK,\minplay}$.
  \end{enumerate}
\end{proposition}

We are now in conditions to present our main result.  The following
theorem is two folded.  On the one hand, it states that the polytopal
stochastic games of all quantitative objectives of interest in this
paper --namely, quantitative reachability, expected total accumulated
reward, expected discounted accumulated rewards, and expected average
rewards-- are determined.
%
%% On the other hand, it states that any of these quantiative objectives
%% for PSG can be equivalently solved in its extreme interpretation.
%
On the other hand, it states that these objectives for PSG can be
equivalently solved in its extreme interpretation.

\begin{theorem}\label{th:determinacy:and:discretazation}%
  Let $\StochGK$ and $\StochHK$ be respectively the interpretation and
  the extreme interpretation of $\StochK$.  Then,
  \begin{enumerate}
  \item\label{th:determinacy:and:discretazation:i}%
    $\displaystyle
    \inf_{\strat_\minplay\in\Strategies{\StochGK,\minplay}}\sup_{\strat_\maxplay\in\Strategies{\StochGK,\maxplay}}\Prob^{\strat_{\maxplay},\strat_{\minplay}}_{\StochGK,s}(\Finally C)
    =
    \inf_{\strat_\minplay\in\DetMemorylessStrats{\StochHK,\minplay}}\sup_{\strat_\maxplay\in\DetMemorylessStrats{\StochHK,\maxplay}}\Prob^{\strat_{\maxplay},\strat_{\minplay}}_{\StochHK,s}(\Finally C) = {}$\newline
    \mbox{}\hfill
    $\displaystyle
    {} =
    \sup_{\strat_\maxplay\in\DetMemorylessStrats{\StochHK,\maxplay}}\inf_{\strat_\minplay\in\DetMemorylessStrats{\StochHK,\minplay}}\Prob^{\strat_{\maxplay},\strat_{\minplay}}_{\StochHK,s}(\Finally C)
    =
    \sup_{\strat_\maxplay\in\Strategies{\StochGK,\maxplay}}\inf_{\strat_\minplay\in\Strategies{\StochGK,\minplay}}\Prob^{\strat_{\maxplay},\strat_{\minplay}}_{\StochGK,s}(\Finally C)$\newline
    for all $C\subseteq\sgnodes$; and
  \item\label{th:determinacy:and:discretazation:ii}%
    $\displaystyle
    \inf_{\strat_\minplay\in\Strategies{\StochGK,\minplay}}\sup_{\strat_\maxplay\in\Strategies{\StochGK,\maxplay}}\Expect^{\strat_{\maxplay},\strat_{\minplay}}_{\StochGK,s}(\GRewards)
    =
    \inf_{\strat_\minplay\in\DetMemorylessStrats{\StochHK,\minplay}}\sup_{\strat_\maxplay\in\DetMemorylessStrats{\StochHK,\maxplay}}\Expect^{\strat_{\maxplay},\strat_{\minplay}}_{\StochHK,s}(\GRewards)
    = {}$\newline
    \mbox{}\hfill
    $\displaystyle
    {} =
    \sup_{\strat_\maxplay\in\DetMemorylessStrats{\StochHK,\maxplay}}\inf_{\strat_\minplay\in\DetMemorylessStrats{\StochHK,\minplay}}\Expect^{\strat_{\maxplay},\strat_{\minplay}}_{\StochHK,s}(\GRewards)
    =
    \sup_{\strat_\maxplay\in\Strategies{\StochGK,\maxplay}}\inf_{\strat_\minplay\in\Strategies{\StochGK,\minplay}}\Expect^{\strat_{\maxplay},\strat_{\minplay}}_{\StochGK,s}(\GRewards)$,\newline
    provided:
    \begin{enumerate*}
    \item\label{th:determinacy:and:discretazation:ii:cond:total}%
      $\StochGK$ is almost surely stopping whenever
      $\GRewards=\TRewards$, and
    \item\label{th:determinacy:and:discretazation:ii:cond:average}%
      $\StochGK$ is irreducble whenever $\GRewards=\ARewards$.
    \end{enumerate*}
%%     $\Expect^{\strat_{\maxplay},\strat_{\minplay}}_{\StochHK,s}(\GRewards)$
%%     is defined for all
%%     $\strat_\maxplay\in\Strategies{\StochHK,\maxplay}$, and
%%     $\strat_\minplay\in\Strategies{\StochHK,\minplay}$.
%%     \remarkPRD{cuidado, quiz\'as haya que poner condiciones
%%       particulares para cada modo}
  \end{enumerate}
\end{theorem}
\begin{proof}
  For item~\ref{th:determinacy:and:discretazation:ii} we calculate as follows:
  \begin{align*}
    \textstyle
    \inf_{\strat_\minplay\in\Strategies{\StochGK,\minplay}}\sup_{\strat_\maxplay\in\Strategies{\StochGK,\maxplay}}\Expect^{\strat_{\maxplay},\strat_{\minplay}}_{\StochGK,s}(\GRewards)
    \hspace{-14em} &
    \\
    & \textstyle \leq
    \inf_{\strat_\minplay\in\SemiMarkovStrats{\StochGK,\minplay}}\sup_{\strat_\maxplay\in\Strategies{\StochGK,\maxplay}}\Expect^{\strat_{\maxplay},\strat_{\minplay}}_{\StochGK,s}(\GRewards)
    \tag{$\SemiMarkovStrats{\StochGK,\minplay}\subseteq\Strategies{\StochGK,\minplay}$}\\
    & \textstyle =
    \inf_{\strat_\minplay\in\SemiMarkovStrats{\StochGK,\minplay}}\sup_{\strat_\maxplay\in\SemiMarkovStrats{\StochGK,\maxplay}}\Expect^{\strat_{\maxplay},\strat_{\minplay}}_{\StochGK,s}(\GRewards)
    \tag{by Lemma~\ref{lm:semimarkov}.\ref{lm:semimarkov:iii}}\\
    & \textstyle =
    \inf_{\strat_\minplay\in\XSemiMarkovStrats{\StochGK,\minplay}}\sup_{\strat_\maxplay\in\XSemiMarkovStrats{\StochGK,\maxplay}}\Expect^{\strat_{\maxplay},\strat_{\minplay}}_{\StochGK,s}(\GRewards)
    \tag{by Corollary~\ref{cor:xsemimarkov}.\ref{cor:xsemimarkov:ii}}\\
    & \textstyle =
    \inf_{\strat_\minplay\in\SemiMarkovStrats{\StochHK,\minplay}}\sup_{\strat_\maxplay\in\SemiMarkovStrats{\StochHK,\maxplay}}\Expect^{\strat_{\maxplay},\strat_{\minplay}}_{\StochHK,s}(\GRewards)
    \tag{by Prop.~\ref{prop:infsup:supinf:StochGK:StochHK}.\ref{prop:infsup:supinf:StochGK:StochHK:iii}}\\
    & \textstyle \leq
    \inf_{\strat_\minplay\in\DetMemorylessStrats{\StochHK,\minplay}}\sup_{\strat_\maxplay\in\SemiMarkovStrats{\StochHK,\maxplay}}\Expect^{\strat_{\maxplay},\strat_{\minplay}}_{\StochHK,s}(\GRewards)
    \tag{$\DetMemorylessStrats{\StochHK,\minplay}\subseteq\SemiMarkovStrats{\StochHK,\minplay}$}\\
    & \textstyle =
    \inf_{\strat_\minplay\in\DetMemorylessStrats{\StochHK,\minplay}}\sup_{\strat_\maxplay\in\DetMemorylessStrats{\StochHK,\maxplay}}\Expect^{\strat_{\maxplay},\strat_{\minplay}}_{\StochHK,s}(\GRewards)
    \tag{by Prop.~\ref{prop:mdp:results}.\ref{prop:mdp:results:iii}}\\
    & \textstyle =
    \sup_{\strat_\maxplay\in\DetMemorylessStrats{\StochHK,\maxplay}}\inf_{\strat_\minplay\in\DetMemorylessStrats{\StochHK,\minplay}}\Expect^{\strat_{\maxplay},\strat_{\minplay}}_{\StochHK,s}(\GRewards)
    \tag{*}\\
    & \textstyle =
    \sup_{\strat_\maxplay\in\DetMemorylessStrats{\StochHK,\maxplay}}\inf_{\strat_\minplay\in\SemiMarkovStrats{\StochHK,\minplay}}\Expect^{\strat_{\maxplay},\strat_{\minplay}}_{\StochHK,s}(\GRewards)
    \tag{by Prop.~\ref{prop:mdp:results}.\ref{prop:mdp:results:iv}}\\
    & \textstyle \leq
    \sup_{\strat_\maxplay\in\SemiMarkovStrats{\StochHK,\maxplay}}\inf_{\strat_\minplay\in\SemiMarkovStrats{\StochHK,\minplay}}\Expect^{\strat_{\maxplay},\strat_{\minplay}}_{\StochHK,s}(\GRewards)
    \tag{$\DetMemorylessStrats{\StochHK,\maxplay}\subseteq\SemiMarkovStrats{\StochHK,\maxplay}$}\\
    & \textstyle =
    \sup_{\strat_\maxplay\in\XSemiMarkovStrats{\StochGK,\maxplay}}\inf_{\strat_\minplay\in\XSemiMarkovStrats{\StochGK,\minplay}}\Expect^{\strat_{\maxplay},\strat_{\minplay}}_{\StochGK,s}(\GRewards)
    \tag{by Prop.~\ref{prop:infsup:supinf:StochGK:StochHK}.\ref{prop:infsup:supinf:StochGK:StochHK:iv}}\\
    & \textstyle =
    \sup_{\strat_\maxplay\in\SemiMarkovStrats{\StochGK,\maxplay}}\inf_{\strat_\minplay\in\SemiMarkovStrats{\StochGK,\minplay}}\Expect^{\strat_{\maxplay},\strat_{\minplay}}_{\StochGK,s}(\GRewards)
    \tag{by Corollary~\ref{cor:xsemimarkov}.\ref{cor:xsemimarkov:ii}}\\
    & \textstyle =
    \sup_{\strat_\maxplay\in\SemiMarkovStrats{\StochGK,\maxplay}}\inf_{\strat_\minplay\in\Strategies{\StochGK,\minplay}}\Expect^{\strat_{\maxplay},\strat_{\minplay}}_{\StochGK,s}(\GRewards)
    \tag{by Lemma~\ref{lm:semimarkov}.\ref{lm:semimarkov:iii}}\\
    & \textstyle 
    \leq
    \sup_{\strat_\maxplay\in\Strategies{\StochGK,\maxplay}}\inf_{\strat_\minplay\in\Strategies{\StochGK,\minplay}}\Expect^{\strat_{\maxplay},\strat_{\minplay}}_{\StochGK,s}(\GRewards)
    \tag{$\SemiMarkovStrats{\StochGK,\maxplay}\subseteq\Strategies{\StochGK,\maxplay}$}\\
    & \textstyle \leq
    \inf_{\strat_\minplay\in\Strategies{\StochGK,\minplay}}\sup_{\strat_\maxplay\in\Strategies{\StochGK,\maxplay}}\Expect^{\strat_{\maxplay},\strat_{\minplay}}_{\StochGK,s}(\GRewards)
    \tag{by prop. of $\sup$ and $\inf$}
  \end{align*}
  Since the last term is equal to the first term in the calculation,
  item~\ref{th:determinacy:and:discretazation:ii} is concluded.
  In particular, step (*) is justified as follows, depending on
  $\GRewards$:
  %
%%   \begin{itemize}
%%   \item%
    For $\GRewards=\TRewards$, (*) follows by
    \cite[Theorem~4.2.6]{FilarV96} since, by
    Proposition~\ref{prop:stopping:irreducible:StochGK:StochHK}.\ref{prop:stopping:irreducible:StochGK:StochHK:i},
    the game $\StochHK$ is also almost surely stopping.
%%   \item%
    For $\GRewards=\DRewards{\gamma}$ (*) follows by
    \cite[Theorem~4.3.2]{FilarV96}.
%%   \item%
    For $\GRewards=\ARewards$ (*) follows by
    \cite[Theorem~5.1.5]{FilarV96} since, by
    Proposition~\ref{prop:stopping:irreducible:StochGK:StochHK}.\ref{prop:stopping:irreducible:StochGK:StochHK:ii},
    the game $\StochHK$ is also irreducible.
%%   \end{itemize}

  Item~\ref{th:determinacy:and:discretazation:i} of the theorem
  follows similarly. In each step, propositions, lemmas and
  corollaries are the same only differing on the item, while
  step~(*) follows from~\cite[Lemma 6]{Condon92}.
  \qed
\end{proof}

Since extreme interpretations are finite, the values of the different
games can be calculated following known
algorithms~\cite{Condon92,FilarV96}.  Thus,
Theorem~\ref{th:determinacy:and:discretazation} immediately provides
an algorithmic solution for PSGs.

The number of vertices of a polytope grows exponentially in the
dimension of the polytope~\cite{KaibelP03}. More precisely if $d$ is
the dimension of a polytope $K$ and $m$ is the number of inequalities
that defines it, $\vertices(K)\sim\Omega(m^{\lfloor{d/2}\rfloor})$.
This implies that the extreme interpretation $\StochHK$ grows
exponentially on the largest size of the support sets of the
distributions involved in the original PSG $\StochK$ which we expect
not to be too large.  (In our example of Sec.~\ref{sec:roborta},
$\lfloor{d/2}\rfloor=2$)

Condon~\cite{Condon92} showed that deciding reachability in stochastic
games is in $\NP \cap \coNP$.  Despite the exponential grow, this is
still our case as we show in the following.
Let $\val{s}{\StochK}$ denote the value of the game at state $s$, that
is, it is equal to
$\sup_{\strat_{\maxplay} \in \Strategies{\maxplay}} \inf_{\strat_{\minplay} \in \Strategies{\minplay}} \Prob^{\strat_{\maxplay},\strat_{\minplay}}_{\StochGK,s}(\Finally\goal)$,
or
$\sup_{\strat_{\maxplay} \in \Strategies{\maxplay}} \inf_{\strat_{\minplay} \in \Strategies{\minplay}} \Expect^{\strat_{\maxplay},\strat_{\minplay}}_{\StochGK,s}[\GRewards]$.
The problem is then to decide whether $\val{s}{\StochK} \geq q$, for a
given $q \in \mathbb{Q}$ and $s \in \psgnodes$
Since for all the cases (total reward, discounted reward, average
reward and reachability objectives under the respective conditions)
the value $\val{s}{\StochK}$ of the game can be achieved with an extreme
memoryless and deterministic strategies, we can reason as follows:
\begin{enumerate*}[(i)]
\item%
  guess a memoryless and deterministic strategy for each player,
\item%
  on the resulting Markov chain compute the corresponding measure
  (i.e. total reward, discounted reward, average reward or
  reachability) on the respective set of linear equations, which can
  be done in polynomial time (for $\ARewards$ an extra linear
  summation is needed)~\cite{Kulkarni17}, 
\item%
  verify if it is a fixpoint of Bellman equations (for reachability, discounted, or total reward), or a fixpoint of the Alg.~5.1.1 of \cite{FilarV96}, in the case of average reward, and
\item  check whether $\val{s}{\StochK} \geq q$.
\end{enumerate*}
This puts our problem in $\NP$.
With the same process we can check whether $\val{s}{\StochK} < q$
which puts the problem also in $\coNP$.  Hence we have the next
theorem.
\begin{theorem}
  For any PSG $\StochK$, $q \in \mathbb{Q}$, and $s \in \psgnodes$,
  the problem of deciding whether $\val{s}{\StochK} \geq q$ is in
  ${\NP} \cap {\coNP}$.
  %
%%   For $\GRewards=\TRewards$ and $\GRewards=\ARewards$
  For $\GRewards\in\{\TRewards, \ARewards\}$
  the decision problem is restricted to
  $\StochGK$ being almost surely stopping and irreducible, respectively.
\end{theorem}

\section{Concluding remarks}

We believe that polytopal games may have several applications in practice,  particularly, in scenarios where the probabilities are not exact but can be characterized with linear equations.  We observe that one may expect that  the number of vertices of the polytopes keep small in practical examples, hence the game discretization may have no impact on the runtime of a tool implementing the approach described in the paper.  However,  we leave as further work the implementation of such a tool and an in-depth evaluation of it.

In addition, it would be also be of interest to explore other types of objectives, including $\omega$-regular objectives as already study for standard stochastic games in~\cite{ChatterjeeAH05} or even solving stochastic games for conditional probabilities of temporal properties or conditional expectations of rewards models as widely studied by Christel Baier and her team in the context of Markov decision processes~\cite{BaierKKM14,Baier0KW17,MarckerB0K17,PiribauerB19}.

%% as well as the use of other kinds of objectives including multiobjectives or lexicographic ones.

\bibliographystyle{splncs04}% the mandatory bibstyle
\bibliography{references}

\appendix

\section{Full proofs}

\begin{proof}[of Proposition~\ref{prop:PSG:SSG}]
  Let $\StochGK=(\sgnodes,(\sgnmax,\sgnmin),\sgactions,\sgtrans)$ with
  $\sgactions=\left(\bigcup_{s\in\sgnodes}\psgtrans(s)\right)\times\Dist(\sgnodes)$
  and
  $\StochGKp=(\sgnodes,(\sgnmax,\sgnmin),\sgactions',\sgtrans')$ with
  $\sgactions'=\left(\bigcup_{s\in\sgnodes}\psgtrans'(s)\right)\times\Dist(\sgnodes)=\left(\bigcup_{\substack{s\in\sgnodes,~\\K\in\psgtrans(s)}}\Triang(K)\right)\times\Dist(\sgnodes)$,
  be the respective interpretations of $\StochK$ and $\StochK'$.

  To prove item~\ref{item:PSG:SSG:i}, first fix a function
  $f_K:\Triang(K)\to\powerset(K)$ for each polytope
  $K\in\DPoly(\sgnodes)$ satisfying
  \begin{enumerate}[(i)]
  \item%
    $\forall {K'\in\Triang(K)}\colon{f_K(K')\subseteq K'}$,
  \item%
    $\bigcup_{K'\in\Triang(K)}f_K(K') = K$, and
  \item%
    $\forall {K'_1,K'_2\in\Triang(K)}\colon {{f_K(K'_1)\cap f_K(K'_2)\neq\emptyset} \limp {K'_1=K'_2}}$.
  \end{enumerate}
  Thus, $f_K(K')$ is almost the simplex $K'$ but ensuring that
  distributions on the faces of $K'$ are exactly in one of the
  $f_K(K'')$, $K''\in\Triang(K)$.

  Now, let $\strat_{\maxplay}$ and $\strat_{\minplay}$ be a pair of
  strategies for $\StochGK$.  Define $\strat'_{i}$,
  $i\in\{\maxplay,\minplay\}$, for all $\hat{\rho}\in\sgnodes^*$,
  $s\in\sgnodes_{i}$, and $A'\in\Salg_{\sgactions'}$ by
  \begin{equation}\label{eq:def:stratp:ssg}
    \strat'_{i}(\hat{\rho}s)(A') = \sum_{K\in\psgtrans(s)} \sum_{K'\in\Triang(K)} \strat_{i}(\hat{\rho}s)(\{K\}\times({A'\sect{K'}}\cap f_K(K')))
  \end{equation}
  where ${A'\sect{K'}}=\{\mu\mid (K',\mu)\in A'\}$ is the $K'$ section
  of the measurable set $A'$.
  Notice that $f_K$ ensures that the faces of each $K'\in\Triang(K)$
  are considered in exactly one summand of the inner summation
  of~(\ref{eq:def:stratp:ssg}).  Thus, $\strat'_{i}(\hat{\rho})$ is a
  well defined probability measure and hence $\strat'_{\maxplay}$ and
  $\strat'_{\minplay}$ is a pair of strategies for $\StochGKp$.

  It is straightforward to check that if $\strat_{i}$ is memoryless or
  semi-Markov, so is $\strat'_{i}$.  Suppose $\strat_{i}$ is extreme,
  then we have that
  \begin{align}
    \strat'_{i}(\hat{\rho}s)(\{(K',\mu)\in\sgactions'(s)\mid\mu\in\vertices(K')\}) \hspace{-12em}
    \notag\\
    & =
    \sum_{K\in\psgtrans(s)} \sum_{K'\in\Triang(K)} \strat_{i}(\hat{\rho}s)(\{K\}\times(\vertices(K')\cap f_K(K')))
    \label{eq:strat:ssg:extreme:i}\\
    & =
    \sum_{K\in\psgtrans(s)} \strat_{i}(\hat{\rho}s)\bigg(\{K\}\times\bigg(\bigcup_{K'\in\Triang(K)}\vertices(K')\cap f_K(K')\bigg)\bigg)
    \label{eq:strat:ssg:extreme:ii}\\
    & =
    \sum_{K\in\psgtrans(s)} \strat_{i}(\hat{\rho}s)(\{K\}\times\vertices(K))
    \label{eq:strat:ssg:extreme:iii}\\
    & =
    \strat_{i}(\hat{\rho}s)\bigg(\bigcup_{K\in\psgtrans(s)}\{K\}\times\vertices(K)\bigg)
    \label{eq:strat:ssg:extreme:iv}\\
    & =
    \strat_{i}(\hat{\rho}s)(\{(K,\mu)\in\sgactions(s)\mid\mu\in\vertices(K)\})
    \label{eq:strat:ssg:extreme:v}\\
    & = 1
    \label{eq:strat:ssg:extreme:vi}
  \end{align}
  Equality~(\ref{eq:strat:ssg:extreme:i}) follows
  by~(\ref{eq:def:stratp:ssg}) and (\ref{eq:strat:ssg:extreme:ii}) is
  a consequence of $\strat_{i}$ being a measure and the fact that
  $f_K$ guarantees the disjointness of sets in the union.
  $f_K$ also guarantees that no vertix of $K$ is lost, hence
  (\ref{eq:strat:ssg:extreme:iii}).
  (\ref{eq:strat:ssg:extreme:iv}) is a consequence of $\strat_{i}$
  being a measure and (\ref{eq:strat:ssg:extreme:v}) by definition of
  $\sgactions(s)$ and $\sgtrans(s)$.
  Finally, (\ref{eq:strat:ssg:extreme:vi}) follows from $\strat_{i}$
  being extreme.

  Let $\strat_{i}$ be deterministic and suppose
  $\strat_{i}(\hat{\rho}s)(\{(K_\star,\mu)\})=1$ for $K_\star\in\psgtrans(s)$ and
  $\mu\in K_\star$.  Besides, suppose that $K'_\star\in\Triang(K_\star)$ such that
  $\mu\in f_{K_\star}(K'_\star)$.  Then
  \begin{align}
    \strat'_{i}(\hat{\rho}s)(\{(K'_\star,\mu)\}) \hspace{-5em}\notag\\
    & =
    \sum_{K\in\psgtrans(s)} \sum_{K'\in\Triang(K)} \strat_{i}(\hat{\rho}s)(\{K\}\times({\{(K'_\star,\mu)\}\sect{K'}}\cap f_K(K')))
    \label{eq:strat:ssg:det:i}\\
    & =
    \strat_{i}(\hat{\rho}s)(\{(K_\star,\mu)\})
    \label{eq:strat:ssg:det:ii}\\
    & =
    1
    \label{eq:strat:ssg:det:iii}
  \end{align}
  Equality~(\ref{eq:strat:ssg:det:i}) follows
  by~(\ref{eq:def:stratp:ssg}). (\ref{eq:strat:ssg:det:ii}) follows
  from the fact that all summands are $0$ except for the one in which
  $K=K_\star$ and $K'=K'_\star$.  Finally (\ref{eq:strat:ssg:det:iii})
  follows because $\strat_{i}$ is deterministic with
  $\strat_{i}(\hat{\rho}s)(\{(K_\star,\mu)\})=1$ by assumption.

  To prove that
  $\Prob^{\strat_{\maxplay},\strat_{\minplay}}_{\StochGK,s}=\Prob^{\strat'_{\maxplay},\strat'_{\minplay}}_{\StochGKp,s}$
  it sufficies to state that
  $\Prob^{\strat_{\maxplay},\strat_{\minplay},n}_{\StochGK,s}=\Prob^{\strat'_{\maxplay},\strat'_{\minplay},n}_{\StochGKp,s}$
  for all $n\geq0$ which we show by induction in the following.

  For $n=0$,
  $\Prob^{\strat_{\maxplay},\strat_{\minplay},0}_{\StochGK,s}(s')=\Dirac_{s}(s')=\Prob^{\strat'_{\maxplay},\strat'_{\minplay},0}_{\StochGKp,s}(s')$.
  For $n+1>0$ we calculate as follows.  Suppose
  $\hat{\rho}\in\sgnodes^n$, $s''\in\sgnodes$ and $s'\in\sgnodes_{i}$.
  Then,
  \begin{align}
    &\!\!
    \Prob^{\strat'_{\maxplay},\strat'_{\minplay},n+1}_{\StochGKp,s}(\hat{\rho} s'' s')
    \notag\\
    & =
    \Prob^{\strat'_{\maxplay},\strat'_{\minplay},n}_{\StochGKp,s}(\hat{\rho}s'')\int_{\sgactions'}\sgtrans'(s'',\cdot,s')\ \diff(\strat'_{i}(\hat{\rho}s'')(\cdot))
    \label{eq:strat:ssg:prob:i}\\
%%     & =
%%     \Prob^{\strat'_{\maxplay},\strat'_{\minplay},n}_{\StochGKp,s}(\hat{\rho}s'')\sum_{K\in\psgtrans(s'')}\sum_{K'\in\Triang(K)}\int_{\{K'\}\times f_K(K')}\sgtrans'(s'',\cdot,s')\ \diff(\strat'_{i}(\hat{\rho}s'')(\cdot))
%%     \label{eq:strat:ssg:prob:ii}\\
    & =
    \Prob^{\strat'_{\maxplay},\strat'_{\minplay},n}_{\StochGKp,s}(\hat{\rho}s'')\sum_{\substack{K\in\psgtrans(s'')\\K'\in\Triang(K)}}\int_{\{K'\}\times f_K(K')}\sgtrans'(s'',\cdot,s')\ \diff(\strat'_{i}(\hat{\rho}s'')(\cdot))
    \label{eq:strat:ssg:prob:ii}\\
%%     & =
%%     \Prob^{\strat'_{\maxplay},\strat'_{\minplay},n}_{\StochGKp,s}(\hat{\rho}s'')\sum_{K\in\psgtrans(s'')}\sum_{K'\in\Triang(K)}\int_{f_K(K')}\sgtrans'(s'',(K',\cdot),s')\ \diff(\strat'_{i}(\hat{\rho}s'')((K',\cdot)))
%%     \label{eq:strat:ssg:prob:iii}\\
    & =
    \Prob^{\strat'_{\maxplay},\strat'_{\minplay},n}_{\StochGKp,s}(\hat{\rho}s'')\sum_{\substack{K\in\psgtrans(s'')\\K'\in\Triang(K)}}\int_{f_K(K')}\sgtrans'(s'',(K',\cdot),s')\ \diff(\strat'_{i}(\hat{\rho}s'')(K',\cdot))
    \label{eq:strat:ssg:prob:iii}\\
%%     & =
%%     \Prob^{\strat_{\maxplay},\strat_{\minplay},n}_{\StochGK,s}(\hat{\rho}s'')\sum_{K\in\psgtrans(s'')}\sum_{K'\in\Triang(K)}\int_{f_K(K')}\sgtrans(s'',(K,\cdot),s')\ \diff(\strat_{i}(\hat{\rho}s'')((K,\cdot)))
%%     \label{eq:strat:ssg:prob:iv}\\
    & =
    \Prob^{\strat_{\maxplay},\strat_{\minplay},n}_{\StochGK,s}(\hat{\rho}s'')\sum_{\substack{K\in\psgtrans(s'')\\K'\in\Triang(K)}}\int_{f_K(K')}\sgtrans(s'',(K,\cdot),s')\ \diff(\strat_{i}(\hat{\rho}s'')(K,\cdot))
    \label{eq:strat:ssg:prob:iv}\\
    & =
    \Prob^{\strat_{\maxplay},\strat_{\minplay},n}_{\StochGK,s}(\hat{\rho}s'')\sum_{K\in\psgtrans(s'')}\int_{_{\left(\bigcup_{K'\in\Triang(K)}f_K(K')\right)}}\hspace{-1.7em}\sgtrans(s'',(K,\cdot),s')\ \diff(\strat_{i}(\hat{\rho}s'')(K,\cdot))
    \label{eq:strat:ssg:prob:v}\\
    & =
    \Prob^{\strat_{\maxplay},\strat_{\minplay},n}_{\StochGK,s}(\hat{\rho}s'')\int_{\left(\bigcup_{K\in\psgtrans(s'')} \{K\}{\times}K\right)}\sgtrans(s'',\cdot,s')\ \diff(\strat_{i}(\hat{\rho}s'')(\cdot))
    \label{eq:strat:ssg:prob:vi}\\
    & =
    \Prob^{\strat_{\maxplay},\strat_{\minplay},n+1}_{\StochGK,s}(\hat{\rho}s''s')
    \label{eq:strat:ssg:prob:vii}
  \end{align}
  Equality~(\ref{eq:strat:ssg:prob:i}) is the definition of
  $\Prob^{\strat'_{\maxplay},\strat'_{\minplay},n+1}_{\StochGKp,s}$.
  (\ref{eq:strat:ssg:prob:ii}) follows by calculations and noting that
  $\strat'_{i}(\hat{\rho}s'')\bigg(\sgactions'\setminus\bigg(\bigcup_{\substack{K\in\psgtrans(s'')\\K'\in\Triang(K)}}\{K'\}{\times}f_{K}(K')\bigg)\bigg)=0$.
  (\ref{eq:strat:ssg:prob:iii}) is a consequence of Fubini's theorem.
  (\ref{eq:strat:ssg:prob:iv}) follows by induction hypothesis and the
  easy-to-check equalities
  $\sgtrans'(s'',(K',\mu),s')=\sgtrans(s'',(K,\mu),s')$, for all
  $\mu\in f_K(K')$, and
  $\strat'_{i}(\hat{\rho}s'')(\{K'\}\times B\cap f_K(K'))) = \strat_{i}(\hat{\rho}s'')(\{K\}\times B\cap f_K(K')))$,
  for all $B\in\Salg_{\Dist(\sgnodes)}$, $K\in\psgtrans(s'')$ and
  $K'\in\Triang(K)$.
  (\ref{eq:strat:ssg:prob:v}) follows by calculations and
  (\ref{eq:strat:ssg:prob:vi}) follows by noting that
  $K=\bigcup_{K'\in\Triang(K)}f_K(K')$ and using Fubini's Theorem.
  Finally, (\ref{eq:strat:ssg:prob:vii}) follows by observing that
  $\strat_{i}(\hat{\rho}s'')\left(\sgactions\setminus\left(\bigcup_{K\in\psgtrans(s'')}\{K\}{\times}K\right)\right)=0$
  and by the definition of
  $\Prob^{\strat_{\maxplay},\strat_{\minplay},n+1}_{\StochGK,s}$.

  \medskip
  
  To prove item~\ref{item:PSG:SSG:ii}, first fix a function $f_s$ for
  each state $s\in\sgnodes$ such that
  \begin{enumerate*}[(i)]
  \item%
    $\forall {K\in\psgtrans(s)}\colon {f_s(K)\subseteq\Triang(K)}$,
  \item%
    $\bigcup_{K\in\psgtrans(s)}f_s(K)=\bigcup_{K\in\psgtrans(s)}\Triang(K)$, and
  \item%
    $\forall {K_1,K_2\in\psgtrans(s)}\colon {{f_s(K_1)\cap f_s(K_2)\neq\emptyset} \limp {K_1=K_2}}$.
  \end{enumerate*}

  Let $\strat'_{\maxplay}$ and $\strat'_{\minplay}$ be a pair of
  strategies for $\StochGKp$.  Define $\strat_{i}$,
  $i\in\{\maxplay,\minplay\}$, for all $\hat{\rho}\in\sgnodes^*$,
  $s\in\sgnodes_{i}$, and $A\in\Salg_{\sgactions}$ by
  \begin{equation}\label{eq:def:stratp:psg}
    \strat_{i}(\hat{\rho}s)(A) = \sum_{K\in\psgtrans(s)} \sum_{K'\in f_{s}(K)} \strat'_{i}(\hat{\rho}s)(\{K'\}\times{A\sect{K}})
  \end{equation}
  Notice that, by definition, $K'\in\psgtrans'(s)$.
  Moreover, notice that $f_s$ ensures that a simplex in a
  triangulation of a polytope outgoing $s$ is consdered in exactly one
  summand of~(\ref{eq:def:stratp:psg}).  Thus $\strat_{i}(\hat{\rho})$
  is a well defined probability measure on $\sgnodes$ and hence
  $\strat_{\maxplay}$ and $\strat_{\minplay}$ is a pair of strategies
  for $\StochGK$.

  It is straightforward to check that if $\strat'_{i}$ is memoryless or
  semi-Markov, so is $\strat_{i}$.  Suppose $\strat'_{i}$ is extreme,
  then we have that
  \begin{align}
    \strat_{i}(\hat{\rho}s)(\{(K'',\mu)\in\sgactions(s)\mid\mu\in\vertices(K'')\})\hspace{-16em}
    \notag\\
    & =
    \sum_{K\in\psgtrans(s)} \sum_{K'\in f_{s}(K)} \strat'_{i}(\hat{\rho}s)(\{K'\}\times{\{(K,\mu)\in\sgactions(s)\mid\mu\in\vertices(K)\}\sect{K}})
    \label{eq:strat:psg:extreme:i}\\
    & =
    \sum_{K\in\psgtrans(s)} \sum_{K'\in f_{s}(K)} \strat'_{i}(\hat{\rho}s)(\{K'\}\times\vertices(K'))
    \label{eq:strat:psg:extreme:ii}\\
%%     & = \textstyle
%%     \strat'_{i}(\hat{\rho}s)\bigg(\bigcup_{\substack{K\in\psgtrans(s)\\K'\in f_{s}(K)}} \{K'\}\times\vertices(K')\bigg)
%%     \label{eq:strat:psg:extreme:iii}\\
    & =
    \strat'_{i}(\hat{\rho}s)\bigg(\bigcup_{K\in\psgtrans(s)}\bigcup_{K'\in f_{s}(K)} \{K'\}\times\vertices(K')\bigg)
    \label{eq:strat:psg:extreme:iii}\\
    & =
    \strat'_{i}(\hat{\rho}s)\bigg(\bigcup_{K'\in\bigcup_{K\in\psgtrans(s)}f_{s}(K)} \{K'\}\times\vertices(K')\bigg)
    \label{eq:strat:psg:extreme:iv}\\
    & =
    \strat'_{i}(\hat{\rho}s)\bigg(\bigcup_{K'\in\psgtrans'(s)} \{K'\}\times\vertices(K')\bigg)
    \label{eq:strat:psg:extreme:v}\\
    & =
    \strat'_{i}(\hat{\rho}s)(\{(K',\mu)\in\sgactions'(s)\mid\mu\in\vertices(K')\})
    \label{eq:strat:psg:extreme:vi}\\
    & =
    1
    \label{eq:strat:psg:extreme:vii}
  \end{align}
  Equality~(\ref{eq:strat:psg:extreme:i}) corresponds to the
  definition of $\strat_{i}$ in (\ref{eq:def:stratp:psg}) and
  (\ref{eq:strat:psg:extreme:ii}) follows from the following
  easy-to-check equalities:
  ${\{(K,\mu)\in\sgactions(s)\mid\mu\in\vertices(K)\}\sect{K}} =
  \vertices(K)$ and
  $\strat'_{i}(\hat{\rho}s)(\{K'\}\times(\vertices(K)\setminus{K'}))=0$.
  (\ref{eq:strat:psg:extreme:iii}) follows because
  $\strat'_{i}(\hat{\rho}s)$ is a probability measure and $f_s$
  guarantees the disjointness of sets in the union while
  (\ref{eq:strat:psg:extreme:iv}) follows by calculations.
  (\ref{eq:strat:psg:extreme:v}) is a consequence of
  $\bigcup_{K\in\psgtrans(s)}f_{s}(K)=\bigcup_{K\in\psgtrans(s)}\Triang(K)=\psgtrans'(s)$
  where the first equality is guaranteed by $f_{s}$ and the second one
  is the definition of $\psgtrans'$.
  Finally, (\ref{eq:strat:psg:extreme:vi})~follows by the definition
  of $\sgactions'(s)$ and (\ref{eq:strat:psg:extreme:vii}) because
  $\strat'_{i}$ is extreme.

  Suppose now that $\strat'_{i}$ is deterministic and assume
  $\strat'_{i}(\hat{\rho}s)(\{(K'_\star,\mu)\})=1$ for
  $K'_\star\in\psgtrans'(s)$ and $\mu\in K'_\star$.  Besides, suppose
  that $K'_\star\in f_{s}(K_\star)$.  Then
  \begin{align}
    \strat_{i}(\hat{\rho}s)(\{(K_\star,\mu)\})
    %\hspace{-16em}
    %\notag\\
    & =
    \sum_{K\in\psgtrans(s)} \sum_{K'\in f_{s}(K)} \strat'_{i}(\hat{\rho}s)(\{K'\}\times({\{(K_\star,\mu)\}\sect{K}}\cap K'))
    \label{eq:strat:psg:det:i}\\
    & =
    \sum_{K'\in f_{s}(K_\star)} \strat'_{i}(\hat{\rho}s)(\{K'\}\times({\{(K_\star,\mu)\}\sect{K_\star}}\cap K'))
    \label{eq:strat:psg:det:ii}\\
    & =
    \sum_{K'\in f_{s}(K_\star)} \strat'_{i}(\hat{\rho}s)(\{K'\}\times(\{\mu\}\cap K'))
    \label{eq:strat:psg:det:iii}\\
    & =
    \strat'_{i}(\hat{\rho}s)(\{(K'_\star,\mu)\}) = 1
    \label{eq:strat:psg:det:iv}
  \end{align}
  (\ref{eq:strat:psg:det:i}) follows by (\ref{eq:def:stratp:psg})
  while (\ref{eq:strat:psg:det:ii}) is a consequence that
  $\{(K_\star,\mu)\}\sect{K}=\emptyset$ whenever $K\neq K_\star$.
  (\ref{eq:strat:psg:det:iii}) follows by definition of
  $\sect{K_\star}$ and (\ref{eq:strat:psg:det:iv}) follows from the
  fact that, for $K'\neq K'_\star$, either $\mu\notin K'$ or
  $\strat'_{i}(\hat{\rho}s)(\{(K',\mu)\})=0$.  The last equality on
  (\ref{eq:strat:psg:det:iv}) follows by the assumptions.

  Like before, to prove that
  $\Prob^{\strat_{\maxplay},\strat_{\minplay}}_{\StochGK,s}=\Prob^{\strat'_{\maxplay},\strat'_{\minplay}}_{\StochGKp,s}$
  it sufficies to state that
  $\Prob^{\strat_{\maxplay},\strat_{\minplay},n}_{\StochGK,s}=\Prob^{\strat'_{\maxplay},\strat'_{\minplay},n}_{\StochGKp,s}$
  for all $n\geq0$ which we show by induction.
  For $n=0$,
  $\Prob^{\strat_{\maxplay},\strat_{\minplay},0}_{\StochGK,s}(s')=\Dirac_{s}(s')=\Prob^{\strat'_{\maxplay},\strat'_{\minplay},0}_{\StochGKp,s}(s')$.
  For $n+1>0$ we calculate as follows.  Suppose
  $\hat{\rho}\in\sgnodes^n$, $s''\in\sgnodes$ and $s'\in\sgnodes_{i}$.
  Then,
  \begin{align}
    &\!\!
    \Prob^{\strat_{\maxplay},\strat_{\minplay},n+1}_{\StochGK,s}(\hat{\rho} s'' s')
    \notag\\
    & =
    \Prob^{\strat_{\maxplay},\strat_{\minplay},n}_{\StochGK,s}(\hat{\rho}s'')\int_{\sgactions}\sgtrans(s'',\cdot,s')\ \diff(\strat_{i}(\hat{\rho}s'')(\cdot))
    \label{eq:strat:psg:prob:i}\\
%%     & =
%%     \Prob^{\strat_{\maxplay},\strat_{\minplay},n}_{\StochGK,s}(\hat{\rho}s'')\int_{\bigcup_{K\in\psgtrans(s'')}\{K\}{\times}K}\sgtrans(s'',\cdot,s')\ \diff(\strat_{i}(\hat{\rho}s'')(\cdot))
%%     \label{eq:strat:psg:prob:ii}\\
    & =
    \Prob^{\strat_{\maxplay},\strat_{\minplay},n}_{\StochGK,s}(\hat{\rho}s'')\sum_{K\in\psgtrans(s'')}\int_{K}\sgtrans(s'',(K,\cdot),s')\ \diff(\strat_{i}(\hat{\rho}s'')(K,\cdot))
    \label{eq:strat:psg:prob:ii}\\
    & =
    \Prob^{\strat_{\maxplay},\strat_{\minplay},n}_{\StochGK,s}(\hat{\rho}s'')\sum_{K\in\psgtrans(s'')}\int_{K}\sgtrans(s'',(K,\cdot),s')\ \diff\bigg(\sum_{K'\in f_{s}(K)} \strat'_{i}(\hat{\rho}s'')(K',\cdot)\bigg)
    \label{eq:strat:psg:prob:iii}\\
    & =
    \Prob^{\strat_{\maxplay},\strat_{\minplay},n}_{\StochGK,s}(\hat{\rho}s'')\sum_{K\in\psgtrans(s'')}\sum_{K'\in f_{s}(K)} \int_{K'}\sgtrans(s'',(K,\cdot),s')\ \diff\bigg(\strat'_{i}(\hat{\rho}s'')(K',\cdot)\bigg)
    \label{eq:strat:psg:prob:iv}\\
    & =
    \Prob^{\strat'_{\maxplay},\strat'_{\minplay},n}_{\StochGKp,s}(\hat{\rho}s'')\sum_{K\in\psgtrans(s'')}\sum_{K'\in f_{s}(K)} \int_{K'}\sgtrans'(s'',(K',\cdot),s')\ \diff\bigg(\strat'_{i}(\hat{\rho}s'')(K',\cdot)\bigg)
    \label{eq:strat:psg:prob:v}\\
    & =
    \Prob^{\strat'_{\maxplay},\strat'_{\minplay},n}_{\StochGKp,s}(\hat{\rho}s'')\sum_{K'\in\psgtrans'(s'')} \int_{K'}\sgtrans'(s'',(K',\cdot),s')\ \diff\bigg(\strat'_{i}(\hat{\rho}s'')(K',\cdot)\bigg)
    \label{eq:strat:psg:prob:vi}\\
    & =
    \Prob^{\strat'_{\maxplay},\strat'_{\minplay},n}_{\StochGKp,s}(\hat{\rho}s'') \int_{\left(\bigcup_{K'\in\psgtrans'(s'')}\{K'\}\times K'\right)}\sgtrans'(s'',\cdot,s')\ \diff\bigg(\strat'_{i}(\hat{\rho}s'')(\cdot)\bigg)
    \label{eq:strat:psg:prob:vii}\\
    & =
    \Prob^{\strat'_{\maxplay},\strat'_{\minplay},n}_{\StochGKp,s}(\hat{\rho}s'') \int_{\sgactions'}\sgtrans'(s'',\cdot,s')\ \diff\bigg(\strat'_{i}(\hat{\rho}s'')(\cdot)\bigg)
    \label{eq:strat:psg:prob:viii}\\
    & =
    \Prob^{\strat'_{\maxplay},\strat'_{\minplay},n+1}_{\StochGKp,s}(\hat{\rho}s''s')
    \label{eq:strat:psg:prob:ix}
  \end{align}
  Equality (\ref{eq:strat:psg:prob:i}) applies the definition of
  $\Prob^{\strat_{\maxplay},\strat_{\minplay},n+1}_{\StochGK,s}$.
  (\ref{eq:strat:psg:prob:ii}) follows from the fact that
  $\strat_{i}(\hat{\rho}s'')\left(\sgactions\setminus\left(\bigcup_{K\in\psgtrans(s'')}\{K\}{\times}K\right)\right)=0$
  and from Fubini's theorem.
  (\ref{eq:strat:psg:prob:iii}) follows because,
  by~(\ref{eq:def:stratp:psg}),
  $\strat_{i}(\hat{\rho}s'')(\{K\}{\times}{B})=\sum_{K'\in f_{s}(K)} \strat'_{i}(\hat{\rho}s'')(\{K'\}{\times}{B})$
  since $(\{K''\}{\times}{B})\sect{K}=\emptyset$ for any $K''\neq K$.
  (\ref{eq:strat:psg:prob:iv}) follows from calculations taking into
  account that
  $\strat'_{i}(\hat{\rho}s'')(\{K'\}{\times}(K\setminus K'))=0$.
  (\ref{eq:strat:psg:prob:v}) follows by induction hypothesis and from
  the fact that $\sgtrans'(s'',(K',\mu),s')=\sgtrans(s'',(K,\mu),s')$
  for all $\mu\in K'$ whenever, $K'\in\Triang(K)$.
  (\ref{eq:strat:psg:prob:vi}) follows from the fact that
  $\bigcup_{K\in\psgtrans(s)}f_{s}(K)=\psgtrans'(s)$ and every
  $f_{s}(K)$ is disjoint from any other $f_{s}(K'')$.
  Fubini's theorem yields (\ref{eq:strat:psg:prob:vii}) and the fact
  that
  $\strat'_{i}(\hat{\rho}s'')\left(\sgactions'\setminus\left(\bigcup_{K'\in\psgtrans'(s'')}\{K'\}{\times}K'\right)\right)=0$
  yields (\ref{eq:strat:psg:prob:viii}).
  Finally, by definition of
  $\Prob^{\strat'_{\maxplay},\strat'_{\minplay},n+1}_{\StochGKp,s}$ we
  conclude in~(\ref{eq:strat:psg:prob:ix}).
  \qed
\end{proof}

\begin{proof}[of Lemma~\ref{lm:encoding:reach:and:expectation}]
  First of all, notice that
  $\Prob^{\strat_{\maxplay},\strat_{\minplay}}_{\StochGK,s}(\Finally^k s') =
  \Prob^{\strat_{\maxplay},\strat_{\minplay}}_{\StochGK,s}(\sgnodes^k\times\{s'\}\times\sgnodes^\omega)
  =
  \Prob^{\strat_{\maxplay},\strat_{\minplay},k}_{\StochGK,s}(\sgnodes^k\times\{s'\})$.
  This fact will be used in the following without making explicit the
  justification.

  For item \ref{lm:encoding:reach:and:expectation:i}, we start by
  proving that for all $n\geq 0$ and $\alpha\in\Reals$,
  \begin{equation}\label{eq:encoding:reach:and:expectation:i:prel}
    \sum_{\hat{\rho} \in \sgnodes^{n+1}} \Prob^{\strat_{\maxplay},\strat_{\minplay},n}_{\StochGK,s}(\hat{\rho})\mult\sum^{n}_{i=0}\alpha^i\mult\reward(\hat{\rho}_i)
    =
    \sum^{n}_{i=0} \sum_{s' \in \sgnodes} \Prob^{\strat_{\maxplay},\strat_{\minplay}}_{\StochGK,s}(\Finally^i s')\mult\alpha^i\mult\reward(s').
  \end{equation}
  We proceed by induction on $n$. For $n=0$ we calculate:
  \begin{align*}
    \sum_{\hat{\rho} \in \sgnodes^{0+1}} \Prob^{\strat_{\maxplay},\strat_{\minplay},0}_{\StochGK,s}(\hat{\rho})\mult\sum^{0}_{i=0}\alpha^i\mult\reward(\hat{\rho}_i)
    & =
    \sum_{s' \in \sgnodes} \Prob^{\strat_{\maxplay},\strat_{\minplay},0}_{\StochGK,s}(s')\mult\alpha^0\mult\reward(s') \\
    & =
    \sum_{s' \in \sgnodes} \Prob^{\strat_{\maxplay},\strat_{\minplay}}_{\StochGK,s}(\Finally^0s')\mult\alpha^0\mult\reward(s') \\
    & =
    \sum^{0}_{i=0} \sum_{s' \in \sgnodes} \Prob^{\strat_{\maxplay},\strat_{\minplay}}_{\StochGK,s}(\Finally^i s')\mult\alpha^i\mult\reward(s')
  \end{align*}
  All steps follow by there respective definitions.

  For $n+1$ ($n\geq0$) we proceed as follows:
  \begin{align}
    \sum_{\hat{\rho} \in \sgnodes^{n+2}} \Prob^{\strat_{\maxplay},\strat_{\minplay},n+1}_{\StochGK,s}(\hat{\rho})\mult\sum^{n+1}_{i=0}\alpha^i\mult\reward(\hat{\rho}_i)
    \hspace{-14em} & \notag\\
    & =
    \sum_{\hat{\rho} \in \sgnodes^{n+1}}\sum_{s' \in \sgnodes} \Prob^{\strat_{\maxplay},\strat_{\minplay},n+1}_{\StochGK,s}(\hat{\rho}s')\mult\left(\left(\sum^{n}_{i=0}\alpha^i\mult\reward(\hat{\rho}_i)\right)+\alpha^{n+1}\mult\reward(s')\right)
    \label{eq:encoding:reach:and:expectation:i:prel:proof:i}\\
    & =
    \sum_{\hat{\rho} \in \sgnodes^{n+1}}\left(\sum_{s' \in \sgnodes} \Prob^{\strat_{\maxplay},\strat_{\minplay},n+1}_{\StochGK,s}(\hat{\rho}s')\right)\mult\sum^{n}_{i=0}\alpha^i\mult\reward(\hat{\rho}_i) \notag\\
    & \qquad\qquad
    {}+\sum_{s' \in \sgnodes} \left(\sum_{\hat{\rho} \in \sgnodes^{n+1}}\Prob^{\strat_{\maxplay},\strat_{\minplay},n+1}_{\StochGK,s}(\hat{\rho}s')\right)\mult\alpha^{n+1}\mult\reward(s')
    \notag\\%\label{eq:encoding:reach:and:expectation:i:prel:proof:ii}\\
    & =
    \sum_{\hat{\rho} \in \sgnodes^{n+1}}\Prob^{\strat_{\maxplay},\strat_{\minplay},n}_{\StochGK,s}(\hat{\rho})\mult\sum^{n}_{i=0}\alpha^i\mult\reward(\hat{\rho}_i) \notag\\
    & \qquad\qquad
    {}+\sum_{s' \in \sgnodes} \Prob^{\strat_{\maxplay},\strat_{\minplay},n+1}_{\StochGK,s}(\sgnodes^{n+1}\times\{s'\})\mult\alpha^{n+1}\mult\reward(s')
    \label{eq:encoding:reach:and:expectation:i:prel:proof:iii}\\
    & =
    \sum^{n}_{i=0} \sum_{s' \in \sgnodes} \Prob^{\strat_{\maxplay},\strat_{\minplay}}_{\StochGK,s}(\Finally^i s')\mult\alpha^i\mult\reward(s') \notag\\
    & \qquad\qquad
    {}+\sum_{s' \in \sgnodes} \Prob^{\strat_{\maxplay},\strat_{\minplay}}_{\StochGK,s}(\Finally^{n+1}s')\mult\alpha^{n+1}\mult\reward(s')
    \label{eq:encoding:reach:and:expectation:i:prel:proof:iv}\\
    & =
    \sum^{n+1}_{i=0} \sum_{s' \in \sgnodes} \Prob^{\strat_{\maxplay},\strat_{\minplay}}_{\StochGK,s}(\Finally^i s')\mult\alpha^i\mult\reward(s') \notag
  \end{align}
  Most of the steps follow by simple calculations.  In particular, in
  (\ref{eq:encoding:reach:and:expectation:i:prel:proof:i}) we separate
  the trailing state.  In
  (\ref{eq:encoding:reach:and:expectation:i:prel:proof:iii}), we use
  the fact that
  $\sum_{s'\in\sgnodes}\Prob^{\strat_{\maxplay},\strat_{\minplay},n+1}_{\StochGK,s}(\hat{\rho}s')
  =
  \Prob^{\strat_{\maxplay},\strat_{\minplay},n+1}_{\StochGK,s}(\hat{\rho}\times\sgnodes)
  =
  \Prob^{\strat_{\maxplay},\strat_{\minplay},n}_{\StochGK,s}(\hat{\rho})$
  in the first summand.  Finally, induction hypothesis on the first
  summand of (\ref{eq:encoding:reach:and:expectation:i:prel:proof:iv})
  is applied.

  Item~\ref{lm:encoding:reach:and:expectation:i} of the lemma follows
  for $\FTRewards^n$ and $\FDRewards{\discfactor}^n$ by taking $\alpha=1$
  and $\alpha=\discfactor$
  in~(\ref{eq:encoding:reach:and:expectation:i:prel}),  respectively.
  The case of $\FARewards^n$ follows by observing that
  $\FARewards^n(\hat{\rho})=\frac{\FTRewards^n(\hat{\rho})}{n+1}$
  and calculating as follows
  $\sum_{\hat{\rho} \in \sgnodes^{n+1}} \Prob^{\strat_{\maxplay},\strat_{\minplay},n}_{\StochGK,s}(\hat{\rho})\mult\FARewards^n(\hat{\rho})
  =
  \sum_{\hat{\rho} \in \sgnodes^{n+1}} \Prob^{\strat_{\maxplay},\strat_{\minplay},n}_{\StochGK,s}(\hat{\rho})\mult\frac{\FTRewards^n(\hat{\rho})}{n+1}
  =
  \sum^{n}_{i=0} \sum_{s' \in \sgnodes} \Prob^{\strat_{\maxplay},\strat_{\minplay}}_{\StochGK,s}(\Finally^i s')\mult\frac{1}{n+1}\mult\reward(s')$.

  Item~\ref{lm:encoding:reach:and:expectation:ii} follows from
  item~\ref{lm:encoding:reach:and:expectation:i} as follows:
  \begin{align}
    \Expect^{\strat_{\maxplay},\strat_{\minplay}}_{\StochGK,s}[\GRewards]
    & =
    \lim_{n\to\infty} \Expect^{\strat_{\maxplay},\strat_{\minplay}}_{\StochGK,s}[\FGRewards^n]
    \label{lm:encoding:reach:and:expectation:ii:eq:i}\\
    & =
    \lim_{n\to\infty} \sum_{\hat{\rho} \in \sgnodes^{n+1}} \Prob^{\strat_{\maxplay},\strat_{\minplay},n}_{\StochGK,s}(\hat{\rho})\mult\FGRewards(\hat{\rho})
    \label{lm:encoding:reach:and:expectation:ii:eq:ii}\\
    & =
    \lim_{n\to\infty} \sum^{n}_{i=0} \sum_{s' \in \sgnodes} \Prob^{\strat_{\maxplay},\strat_{\minplay}}_{\StochGK,s}(\Finally^i s')\mult\fgen(i,n)\mult\reward(s')
    \label{lm:encoding:reach:and:expectation:ii:eq:iii}
  \end{align}
  Equality (\ref{lm:encoding:reach:and:expectation:ii:eq:i}) folows by
  convergence properties of random variables,
  (\ref{lm:encoding:reach:and:expectation:ii:eq:ii}) is the definition of
  $\Expect^{\strat_{\maxplay},\strat_{\minplay}}_{\StochGK,s}[\FGRewards^n]$,
  and (\ref{lm:encoding:reach:and:expectation:ii:eq:ii}) is
  item~\ref{lm:encoding:reach:and:expectation:i}.
\qed
\end{proof}

\begin{proof}[of Lemma~\ref{lm:semimarkov}]
  Define $\starredstrat_{\maxplay}$ as follows.
  For $\hat{\rho}\in\sgnodes^*$, $s'\in\sgnodes$, and $A\in\Salg_\sgactions$,
  such that
  $\Prob^{\strat_{\maxplay},\strat_{\minplay}}_{\StochGK,s}(D \Until^n s') > 0$
  and $|\hat{\rho}| = n\geq 0$, let
  \begin{align*}
  \starredstrat_{\maxplay}(\hat{\rho}s')(A)
  & =
  \frac{\displaystyle\sum_{\hat{\rho}'\in D^n}\Prob^{\strat_{\maxplay},\strat_{\minplay},n}_{\StochGK,s}(\hat{\rho}'s')\mult\strat_{\maxplay}(\hat{\rho}' s')(A)}{\displaystyle\Prob^{\strat_{\maxplay},\strat_{\minplay}}_{\StochGK,s}(D \Until^n s')}
  \end{align*}
  For $s'\in\sgnodes$ with
  $\Prob^{\strat_{\maxplay},\strat_{\minplay}}_{\StochGK,s}(D \Until^n s') = 0$
  and $|\hat{\rho}s'| = n$, define
  $\starredstrat_{\maxplay}(\hat{\rho}s')$ to be
  $\Dirac_{\textsf{f}(s')}$ for a globally fixed function $\textsf{f}$
  such that $\textsf{f}(s')\in\enabled(s')$.
  Notice that
  $\starredstrat_{\maxplay}\in\SemiMarkovStrats{\maxplay}$.
  Therefore, we write $\starredstrat_{\maxplay}(n,s')$ for
  $\starredstrat_{\maxplay}(\hat{\rho}s')$ whenever
  $|\hat{\rho}|=n$.

  We focus first on item~\ref{lm:semimarkov:i} and proceed by induction.
  For $n=0$,
  \[\Prob^{\strat_{\maxplay},\strat_{\minplay}}_{\StochGK,s}(D \Until^0 s') =
  \Prob^{\strat_{\maxplay},\strat_{\minplay},0}_{\StochGK,s}(s') = \Dirac_s(s') =
  \Prob^{\starredstrat_{\maxplay},\strat_{\minplay},0}_{\StochGK,s}(s')=
  \Prob^{\starredstrat_{\maxplay},\strat_{\minplay}}_{\StochGK,s}(D \Until^0 s').\]

  For $n + 1 \geq 0$, note that
  $
  \Prob^{\strat_{\maxplay},\strat_{\minplay}}_{\StochGK,s}(D \Until^{n+1} s')
  = \sum_{\hat{\rho} \in D^{n+1}} \Prob^{\strat_{\maxplay},\strat_{\minplay},n+1}_{\StochGK,s}(\hat{\rho}s')
  = \sum_{s'' \in D}\sum_{\hat{\rho} \in D^n} \Prob^{\strat_{\maxplay},\strat_{\minplay},n+1}_{\StochGK,s}(\hat{\rho}s''s')
  $.
  Hence, it suffices to show that, for all $s''\in\sgnodes$,
  \[\sum_{\hat{\rho} \in D^n} \Prob^{\strat_{\maxplay},\strat_{\minplay},n+1}_{\StochGK,s}(\hat{\rho}s''s')
  =
  \sum_{\hat{\rho} \in D^n} \Prob^{\starredstrat_{\maxplay},\strat_{\minplay},n+1}_{\StochGK,s}(\hat{\rho}s''s'),
  \]
  for which we differentiate two cases, depending on the player. So,
  if $s''\in\sgnmin$ we proceed as follows
  \begin{align}
    \sum_{\hat{\rho} \in D^n} \Prob^{\strat_{\maxplay},\strat_{\minplay},n+1}_{\StochGK,s}(\hat{\rho}s''s') \hspace{-4em} &\notag\\
    & =
    \sum_{\hat{\rho} \in D^n} \Prob^{\strat_{\maxplay},\strat_{\minplay},n}_{\StochGK,s}(\hat{\rho}s'')\int_{\sgactions}\sgtrans(s'',\cdot,s')\ \diff(\strat_{\minplay}(\hat{\rho}s'')(\cdot))
    \label{eq:semimarkov:i:proof:min:i}\\
    & =
    \left(\sum_{\hat{\rho} \in D^n} \Prob^{\strat_{\maxplay},\strat_{\minplay},n}_{\StochGK,s}(\hat{\rho}s'')\right)\int_{\sgactions}\sgtrans(s'',\cdot,s')\ \diff(\strat_{\minplay}(n,s'')(\cdot))
    \label{eq:semimarkov:i:proof:min:ii}\\
    & =
    \left(\sum_{\hat{\rho} \in D^n} \Prob^{\starredstrat_{\maxplay},\strat_{\minplay},n}_{\StochGK,s}(\hat{\rho}s'')\right)\int_{\sgactions}\sgtrans(s'',\cdot,s')\ \diff(\strat_{\minplay}(n,s'')(\cdot))
    \label{eq:semimarkov:i:proof:min:iii}\\
    & =
    \sum_{\hat{\rho} \in D^n} \Prob^{\starredstrat_{\maxplay},\strat_{\minplay},n+1}_{\StochGK,s}(\hat{\rho}s''s')
    \label{eq:semimarkov:i:proof:min:iv}
  \end{align}
  Step~(\ref{eq:semimarkov:i:proof:min:i}) follows by definition of
  $\Prob^{\strat_{\maxplay},\strat_{\minplay},n+1}_{\StochGK,s}$.
  Because $\strat_\minplay$ is semi-Markov the integral can be
  factored out of the summation in~(\ref{eq:semimarkov:i:proof:min:ii}).
  Induction hypothesis is applied
  in~(\ref{eq:semimarkov:i:proof:min:iii}), since
  $\Prob^{\strat_{\maxplay},\strat_{\minplay}}_{\StochGK,s}(D \Until^n s'')
  =
  \left(\sum_{\hat{\rho} \in D^n} \Prob^{\strat_{\maxplay},\strat_{\minplay},n}_{\StochGK,s}(\hat{\rho}s'')\right)$
  and similarly for
  $\Prob^{\starredstrat_{\maxplay},\strat_{\minplay}}_{\StochGK,s}(D \Until^n s'')$.
  Finally, step (\ref{eq:semimarkov:i:proof:min:iv}) resolves using
  the same steps as before in reverse order.

  If $s''\in\sgnmax$ we have two subcases. First suppose
  $\Prob^{\strat_{\maxplay},\strat_{\minplay}}_{\StochGK,s}(D \Until^n s'') > 0$.
  Then
  \begin{align}
    \sum_{\hat{\rho} \in D^n} \Prob^{\strat_{\maxplay},\strat_{\minplay},n+1}_{\StochGK,s}(\hat{\rho}s''s') \hspace{-9em} &\notag\\
    & =
    \sum_{\hat{\rho} \in D^n} \Prob^{\strat_{\maxplay},\strat_{\minplay},n}_{\StochGK,s}(\hat{\rho}s'')\int_{\sgactions}\sgtrans(s'',\cdot,s')\ \diff(\strat_{\maxplay}(\hat{\rho}s'')(\cdot))
    \label{eq:semimarkov:i:proof:max:a:i}\\
    & =
    \Prob^{\strat_{\maxplay},\strat_{\minplay}}_{\StochGK,s}(D \Until^n s'')\mult
    \int_{\sgactions}\sgtrans(s'',\cdot,s')\mult
    \frac{\sum_{\hat{\rho} \in D^n} \Prob^{\strat_{\maxplay},\strat_{\minplay},n}_{\StochGK,s}(\hat{\rho}s'')\ \diff(\strat_{\maxplay}(\hat{\rho}s'')(\cdot))}{\Prob^{\strat_{\maxplay},\strat_{\minplay}}_{\StochGK,s}(D \Until^n s'')}
    \label{eq:semimarkov:i:proof:max:a:ii}\\
    & =
    \Prob^{\starredstrat_{\maxplay},\strat_{\minplay}}_{\StochGK,s}(D \Until^n s'')\mult
    \int_{\sgactions}\sgtrans(s'',\cdot,s')\ \diff(\starredstrat_{\maxplay}(n,s'')(\cdot))
    \label{eq:semimarkov:i:proof:max:a:iii}\\
    & =
    \left(\sum_{\hat{\rho}\in D^n}\Prob^{\starredstrat_{\maxplay},\strat_{\minplay},n}_{\StochGK,s}(\hat{\rho}s'')\right)
    \int_{\sgactions}\sgtrans(s'',\cdot,s')\ \diff(\starredstrat_{\maxplay}(n,s'')(\cdot))
    \label{eq:semimarkov:i:proof:max:a:iv}\\
    & =
    \sum_{\hat{\rho}\in D^n}\Prob^{\starredstrat_{\maxplay},\strat_{\minplay},n}_{\StochGK,s}(\hat{\rho}s'')\mult
    \int_{\sgactions}\sgtrans(s'',\cdot,s')\ \diff(\starredstrat_{\maxplay}(\hat{\rho}s'')(\cdot))
    \label{eq:semimarkov:i:proof:max:a:v}\\
    & =
    \sum_{\hat{\rho} \in D^n} \Prob^{\starredstrat_{\maxplay},\strat_{\minplay},n+1}_{\StochGK,s}(\hat{\rho}s''s')
    \label{eq:semimarkov:i:proof:max:a:vi}
  \end{align}
  Step~(\ref{eq:semimarkov:i:proof:max:a:i}) follows by definition
  of $\Prob^{\strat_{\maxplay},\strat_{\minplay},n+1}_{\StochGK,s}$ and
  (\ref{eq:semimarkov:i:proof:max:a:ii}) is obtained by multiplying
  and dividing by
  $\Prob^{\strat_{\maxplay},\strat_{\minplay}}_{\StochGK,s}(D \Until^n s'')$.
  Step (\ref{eq:semimarkov:i:proof:max:a:iii}) follows by induction and
  using the definition of $\starredstrat_{\maxplay}$.
  By noting that
  $\Prob^{\starredstrat_{\maxplay},\strat_{\minplay}}_{\StochGK,s}(D \Until^n s'')=\sum_{\hat{\rho}\in D^n}\Prob^{\starredstrat_{\maxplay},\strat_{\minplay},n}_{\StochGK,s}(\hat{\rho}s'')$
  we obtain (\ref{eq:semimarkov:i:proof:max:a:iv}).
  Calculations and recalling that
  $\starredstrat_{\maxplay}(n,s'')=\starredstrat_{\maxplay}(\hat{\rho}s'')$,
  whenever $|\hat{\rho}|=n$, yields
  (\ref{eq:semimarkov:i:proof:max:a:v}) which, by definition of
  $\Prob^{\starredstrat_{\maxplay},\strat_{\minplay},n+1}_{\StochGK,s}$,
  concludes in (\ref{eq:semimarkov:i:proof:max:a:vi}).

  For the case
  $\Prob^{\strat_{\maxplay},\strat_{\minplay}}_{\StochGK,s}(D \Until^n s'') = 0$,
  we first prove the following claim.
  \begin{claim}
    $\Prob^{\strat_{\maxplay},\strat_{\minplay},n}_{\StochGK,s}(\hat{\rho}s'') = 0$
    implies
    $\Prob^{\starredstrat_{\maxplay},\strat_{\minplay},n}_{\StochGK,s}(\hat{\rho}s'') = 0$,
    for all $n\geq0$, $\hat{\rho}\in D^n$, and $s''\in\sgnodes$.
  \end{claim}
  \begin{proofofclaim}
    We proceed by induction on $n$.  For $n=0$ the claim follows after
    noting that
    $\Prob^{\strat_{\maxplay},\strat_{\minplay},0}_{\StochGK,s}(s') =
    \Dirac_s(s') =
    \Prob^{\starredstrat_{\maxplay},\strat_{\minplay},0}_{\StochGK,s}(s')$
    by the definition of $\Prob$.

    So, let $n+1>0$ and suppose
    $\Prob^{\strat_{\maxplay},\strat_{\minplay},n+1}_{\StochGK,s}(\hat{\rho}s''s') = 0$.
    Suppose $s''\in\sgnmax$. Since, by definition,
    $\Prob^{\strat_{\maxplay},\strat_{\minplay},n+1}_{\StochGK,s}(\hat{\rho}s''s') =
    \Prob^{\strat_{\maxplay},\strat_{\minplay},n}_{\StochGK,s}(\hat{\rho}s'')\int_{\sgactions}\sgtrans(s'',\cdot,s')\ \diff(\strat_{\maxplay}(\hat{\rho}s'')(\cdot))$,
    then either
    $\Prob^{\strat_{\maxplay},\strat_{\minplay},n}_{\StochGK,s}(\hat{\rho}s'') = 0$
    or
    $\int_{\sgactions}\sgtrans(s'',\cdot,s')\ \diff(\strat_{\maxplay}(\hat{\rho}s'')(\cdot)) = 0$.

    If
    $\Prob^{\strat_{\maxplay},\strat_{\minplay},n}_{\StochGK,s}(\hat{\rho}s'') = 0$,
    then
    $\Prob^{\starredstrat_{\maxplay},\strat_{\minplay},n}_{\StochGK,s}(\hat{\rho}s'') = 0$
    by induction hypothesis and hence
    $\Prob^{\starredstrat_{\maxplay},\strat_{\minplay},n+1}_{\StochGK,s}(\hat{\rho}s''s') =
    \Prob^{\starredstrat_{\maxplay},\strat_{\minplay},n}_{\StochGK,s}(\hat{\rho}s'')\int_{\sgactions}\sgtrans(s'',\cdot,s')\ \diff(\strat_{\maxplay}(\hat{\rho}s'')(\cdot))=0$.

    If, instead
    $\Prob^{\strat_{\maxplay},\strat_{\minplay},n}_{\StochGK,s}(\hat{\rho}s'') > 0$,
    then
    $\Prob^{\strat_{\maxplay},\strat_{\minplay},n}_{\StochGK,s}(D \Until^n s'') > 0$
    and hence,
    \begin{align}
      \Prob^{\starredstrat_{\maxplay},\strat_{\minplay},n+1}_{\StochGK,s}(\hat{\rho}s''s')
      \hspace{-6em}\notag\\
      & =
      \Prob^{\starredstrat_{\maxplay},\strat_{\minplay},n}_{\StochGK,s}(\hat{\rho}s'')\int_{\sgactions}\sgtrans(s'',\cdot,s')\ \diff(\starredstrat_{\maxplay}(\hat{\rho}s'')(\cdot))
      \label{eq:semimarkov:i:proof:max:b:claim:i}\\
      & =
      \Prob^{\starredstrat_{\maxplay},\strat_{\minplay},n}_{\StochGK,s}(\hat{\rho}s'')
      \int_{\sgactions}\sgtrans(s'',\cdot,s')\mult
      \frac{\sum_{\hat{\rho} \in D^n} \Prob^{\strat_{\maxplay},\strat_{\minplay},n}_{\StochGK,s}(\hat{\rho}s'')\ \diff(\strat_{\maxplay}(\hat{\rho}s'')(\cdot))}{\Prob^{\strat_{\maxplay},\strat_{\minplay}}_{\StochGK,s}(D \Until^n s'')}
      \label{eq:semimarkov:i:proof:max:b:claim:ii}\\
      & =
      \Prob^{\starredstrat_{\maxplay},\strat_{\minplay},n}_{\StochGK,s}(\hat{\rho}s'')\mult
      \frac{\sum_{\hat{\rho} \in D^n} \Prob^{\strat_{\maxplay},\strat_{\minplay},n}_{\StochGK,s}(\hat{\rho}s'')}{\Prob^{\strat_{\maxplay},\strat_{\minplay}}_{\StochGK,s}(D \Until^n s'')}
      \int_{\sgactions}\sgtrans(s'',\cdot,s')\mult\diff(\strat_{\maxplay}(\hat{\rho}s'')(\cdot))
      \label{eq:semimarkov:i:proof:max:b:claim:iii}\\
      & =
      0
      \label{eq:semimarkov:i:proof:max:b:claim:iv}
    \end{align}
    Step (\ref{eq:semimarkov:i:proof:max:b:claim:i}) follows by
    definition of
    $\Prob^{\starredstrat_{\maxplay},\strat_{\minplay},n+1}_{\StochGK,s}$,
    step (\ref{eq:semimarkov:i:proof:max:b:claim:ii}) by definition
    of $\starredstrat_{\minplay}$ and step
    (\ref{eq:semimarkov:i:proof:max:b:claim:iii}) by standard
    calculations.  Since necessarily
    $\int_{\sgactions}\sgtrans(s'',\cdot,s')\ \diff(\strat_{\minplay}(\hat{\rho}s'')(\cdot)) = 0$,
    the proof concludes with step
    (\ref{eq:semimarkov:i:proof:max:b:claim:iv}).

    For $s''\in\sgnmin$, it follows as before only
    differing in the case of
    $\Prob^{\strat_{\maxplay},\strat_{\minplay},n}_{\StochGK,s}(\hat{\rho}s'') > 0$,
    in which case
    $\int_{\sgactions}\sgtrans(s'',\cdot,s')\ \diff(\strat_{\minplay}(\hat{\rho}s'')(\cdot)) = 0$.
    Thus,
    $\Prob^{\starredstrat_{\maxplay},\strat_{\minplay},n+1}_{\StochGK,s}(\hat{\rho}s''s') =
    \Prob^{\starredstrat_{\maxplay},\strat_{\minplay},n}_{\StochGK,s}(\hat{\rho}s'')\int_{\sgactions}\sgtrans(s'',\cdot,s')\ \diff(\strat_{\minplay}(\hat{\rho}s'')(\cdot)) =
    0$.
    \qedclaim
  \end{proofofclaim}

  Now, notice that
  $\Prob^{\strat_{\maxplay},\strat_{\minplay}}_{\StochGK,s}(D \Until^n s'') = 0$
  implies that
  $\sum_{\hat{\rho}\in D^n}\Prob^{\strat_{\maxplay},\strat_{\minplay},n}_{\StochGK,s}(\hat{\rho}s'') = 0$,
  and therefore
  $\Prob^{\strat_{\maxplay},\strat_{\minplay},n}_{\StochGK,s}(\hat{\rho}s'') = 0$
  for all $\hat{\rho}\in D^n$.
  Recall that $s''\in\sgnmax$. Then,
  \begin{align}
    \sum_{\hat{\rho} \in D^n} \Prob^{\strat_{\maxplay},\strat_{\minplay},n+1}_{\StochGK,s}(\hat{\rho}s''s') %\hspace{-8em} &\notag\\
    & =
    \sum_{\hat{\rho} \in D^n} \Prob^{\strat_{\maxplay},\strat_{\minplay},n}_{\StochGK,s}(\hat{\rho}s'')\int_{\sgactions}\sgtrans(s'',\cdot,s')\ \diff(\strat_{\maxplay}(\hat{\rho}s'')(\cdot))
    \label{eq:semimarkov:i:proof:max:b:i}\\
    & =
    0
    \label{eq:semimarkov:i:proof:max:b:ii}\\
    & =
    \sum_{\hat{\rho} \in D^n} \Prob^{\starredstrat_{\maxplay},\strat_{\minplay},n}_{\StochGK,s}(\hat{\rho}s'')\int_{\sgactions}\sgtrans(s'',\cdot,s')\ \diff(\starredstrat_{\maxplay}(\hat{\rho}s'')(\cdot))
    \label{eq:semimarkov:i:proof:max:b:iii}\\
    & =
    \sum_{\hat{\rho} \in D^n} \Prob^{\starredstrat_{\maxplay},\strat_{\minplay},n+1}_{\StochGK,s}(\hat{\rho}s''s')
    \label{eq:semimarkov:i:proof:max:b:iv}
  \end{align}
  In the above calculations,
  (\ref{eq:semimarkov:i:proof:max:b:i}) and
  (\ref{eq:semimarkov:i:proof:max:b:iv}) follow by definition of
  $\Prob^{\strat_{\maxplay},\strat_{\minplay},n+1}_{\StochGK,s}$ and
  $\Prob^{\starredstrat_{\maxplay},\strat_{\minplay},n+1}_{\StochGK,s}$, respectively,
  step (\ref{eq:semimarkov:i:proof:max:b:ii}) follow from the fact that
  $\Prob^{\strat_{\maxplay},\strat_{\minplay},n}_{\StochGK,s}(\hat{\rho}s'') = 0$
  for all $\hat{\rho}\in D^n$ and because of this, the claim yields
  (\ref{eq:semimarkov:i:proof:max:b:iii}).
  This concludes the proof of item~\ref{lm:semimarkov:i}.

  \medskip
  
  For~\ref{lm:semimarkov:ii}, we have that
  \begin{align*}
    \Prob^{\strat_{\maxplay},\strat_{\minplay}}_{\StochGK,s}(\Finally C)
    & =
    \sum_{s'\in C}\sum_{n\geq0}\Prob^{\strat_{\maxplay},\strat_{\minplay}}_{\StochGK,s}((\sgnodes\setminus C) \Until^n s')
    \\
    & =
    \sum_{s'\in C}\sum_{n\geq0}\Prob^{\starredstrat_{\maxplay},\strat_{\minplay}}_{\StochGK,s}((\sgnodes\setminus C) \Until^n s')
    =
    \Prob^{\starredstrat_{\maxplay},\strat_{\minplay}}_{\StochGK,s}(\Finally C)
  \end{align*}
  where the middle equality follow from item~\ref{lm:semimarkov:i},
  and the other two because, for all $n\neq m$ and $s'\in C$,
  $((\sgnodes\setminus C) \Until^n s') \cap ((\sgnodes\setminus C) \Until^m s') = \emptyset$.

  \medskip
  
  Item~\ref{lm:semimarkov:iii} can be calculated as follows.
  \begin{align}
    \Expect^{\strat_{\maxplay},\strat_{\minplay}}_{\StochGK,s}[\GRewards]
    & =
    \lim_{n\to\infty} \sum^{n}_{i=0} \sum_{s' \in \sgnodes} \Prob^{\strat_{\maxplay},\strat_{\minplay}}_{\StochGK,s}(\Finally^i s')\mult\fgen(i,n)\mult\reward(s')
    \label{eq:semimarkov:i:proof:i}\\
    & =
    \lim_{n\to\infty} \sum^{n}_{i=0} \sum_{s' \in \sgnodes} \Prob^{\starredstrat_{\maxplay},\strat_{\minplay}}_{\StochGK,s}(\Finally^i s')\mult\fgen(i,n)\mult\reward(s')
    \label{eq:semimarkov:i:proof:ii}\\
    & =
    \Expect^{\starredstrat_{\maxplay},\strat_{\minplay}}_{\StochGK,s}[\GRewards]
    \label{eq:semimarkov:i:proof:iii}
  \end{align}
  Steps (\ref{eq:semimarkov:i:proof:i}) and
  (\ref{eq:semimarkov:i:proof:iii}) follow from
  Lemma~\ref{lm:encoding:reach:and:expectation}.\ref{lm:encoding:reach:and:expectation:ii},
  while (\ref{eq:semimarkov:i:proof:ii}) follows by
  item~\ref{lm:semimarkov:ii} of this lemma.

  \medskip
  
  Notice that the proof can be replicated mutatis mutandi with
  $\maxplay$ and $\minplay$ exchanged.  Therefore the last part of the
  lemma also holds.
\qed
\end{proof}

\begin{proof}[of Lemma~\ref{lm:xsemimarkov}]
  \newcommand{\convp}{\mathsf{p}}%
  For any $K\in\DSimp(\sgnodes)$, $\mu\in K$ and
  $\hat{\mu}\in\vertices(K)$ define $\convp^K(\mu,\hat{\mu})\in[0,1]$
  such that
  $\sum_{\hat{\mu}\in\vertices(K)}\convp^K(\mu,\hat{\mu})\mult\hat{\mu}=\mu$.
  That is, all $\convp^K(\mu,\hat{\mu})$, $\hat{\mu}\in\vertices(K)$,
  are the unique factors that define the convex combination for $\mu$
  in the simplex $K$.  Therefore, $\convp^K(\mu,\hat{\mu})$ is well
  defined for all $K\in\DSimp(\sgnodes)$, $\mu\in K$ and
  $\hat{\mu}\in\vertices(K)$.  In any other case, let
  $\convp^K(\mu,\hat{\mu})=0$.

  Let $\convp((K,\mu),(K,\hat{\mu}))=\convp^K(\mu,\hat{\mu})$ for all
  $K\in\DSimp(\sgnodes)$, $\mu\in K$ and $\hat{\mu}\in\vertices(K)$,
  and let $\convp(a,b)=0$ for any other $a,b\in\sgactions$.
  For every $(K,\mu)\in\sgactions$ such that $\mu\in K$, let
  $\vertices(K,\mu)=\{(K,\hat{\mu})\mid\hat{\mu}\in\vertices(K)\}$ and
  let $\vertices(K,\mu)=\emptyset$ otherwise.
  Thus, for every $s\in\sgnodes$ and $a\in\sgactions$,
  \begin{equation}\label{eq:xsemimarkov:sgtrans:convp}
    \sgtrans(s,a,\cdot) =
    \sum_{b\in\vertices(a)}\convp(a,b)\mult\sgtrans(s,b,\cdot).
  \end{equation}

  We also extend $\convp$ to measurable sets $B\in\Salg_\sgactions$
  and $a\in\sgactions$ by
  $\convp(a,B)=\sum_{b\in B\cap\vertices(a)}\convp(a,b)$.
  We observe that $\convp(\cdot,B)$ is measurable for any
  $B\in\Salg_\sgactions$, which is a consequence of the following
  calculation, where $p\in[0,1]$ and $\proj_2(B)=\{\mu\mid(K,\mu)\in B\}$,
  \begin{align*}
    \{a\mid\convp(a,B)\geq p\}
    & = \textstyle
    \left\{(K,\mu)\mid\sum_{\hat{\mu}\in\proj_2(B)\cap\vertices(K)}\convp^K(\mu,\hat{\mu})\geq p \right\} \\ 
    & = \textstyle
    \bigcup_{K\in\bigcup_{s\in\psgnodes}\psgtrans(s)}\{K\}\times\left\{\mu\mid\sum_{\hat{\mu}\in\proj_2(B)\cap\vertices(K)}\convp^K(\mu,\hat{\mu})\geq p \right\}  
  \end{align*}
  %\remarkPC{En la ecuación de arriba, los vertices podrían pertenecer a diferentes politopos?}
  %\remarkPRD{No me queda claro esto. una $\mu$ podría ser un vértice de dos politopos distintos, pero esto no debería traer problemas}
  %
  Notice that, if $p=0$ then $\{a\mid\convp(a,B)\geq p\}=\sgactions$,
  which is measurable. If $p>0$, notice that
  $\bigcup_{s\in\psgnodes}\psgtrans(s)$ is finite and
  $\left\{\mu\mid\sum_{\hat{\mu}\in\proj_2(B)\cap\vertices(K)}\convp^K(\mu,\hat{\mu})\geq p \right\} = \left\{\mu\in K\mid\sum_{\hat{\mu}\in\proj_2(B)\cap\vertices(K)}\convp^K(\mu,\hat{\mu})\geq p \right\}$
  is a convex polytope, hence measurable. Therefore,
  $\{a\mid\convp(a,B)\geq p\}$ is measurable.

  For every $\hat{\rho}\in\sgnodes^*$, $s'\in\sgnodes$ and
  $B\in\Salg_\sgactions$, define $\starredstrat_\maxplay$ by
  \[\starredstrat_\maxplay(\hat{\rho}s')(B) =
  \int_{\sgactions} \convp(\cdot,B)\ \diff(\strat_\maxplay(\hat{\rho}s')(\cdot)) =
  \int_{\enabled(s')} \convp(\cdot,B)\ \diff(\strat_\maxplay(\hat{\rho}s')(\cdot)).\]
  The first equality is the definition and the second equality follows
  from the fact that
  $\strat_\maxplay(\hat{\rho}s')(\sgactions\setminus\enabled(s'))=0$.
  $\starredstrat_\maxplay(\hat{\rho}s')$ is defined so that it assigns
  to each vertex of a simplex the weighted contribution (according to
  $\strat_\maxplay(\hat{\rho}s')$) of each distribution (in the said
  simplex) to such vertex.

  Since $\convp(\cdot,B)$ is measurable, $\starredstrat_\maxplay$ is
  well defined. Moreover, because $\strat_\maxplay$ is semi-Markov, so
  is $\starredstrat_\maxplay$.
  The following calculation shows that $\starredstrat_\maxplay$ is
  also extreme.
  \begin{align}
    \starredstrat_\maxplay(\hat{\rho}s')(\vertices(\sgactions(s')))
    \hspace{-6em}&\notag\\
    & =
    \int_{\enabled(s')} \convp(\cdot,\vertices(\sgactions(s')))\ \diff(\strat_\maxplay(\hat{\rho}s')(\cdot))
    \label{eq:xsemimarkov:starredstrat:is:x:proof:i} \\
    & =
    \int_{\enabled(s')} \convp(\cdot,\{(K',\hat{\mu})\mid K'\in\psgtrans(s'), \hat{\mu}\in\vertices(K')\})\ \diff(\strat_\maxplay(\hat{\rho}s')(\cdot))
    \label{eq:xsemimarkov:starredstrat:is:x:proof:ii} \\
    & =
    \int_{\enabled(s')} \lambda(K,\mu).\left(\sum_{K'\in\psgtrans(s'), \hat{\mu}\in\vertices(K')}\convp((K,\mu),(K',\hat{\mu}))\right)\ \diff(\strat_\maxplay(\hat{\rho}s')(\cdot))
    \label{eq:xsemimarkov:starredstrat:is:x:proof:iii} \\
    & =
    \int_{\enabled(s')} \lambda(K,\mu).\left(\charac_{\psgtrans(s')}(K)\sum_{\hat{\mu}\in\vertices(K)}\convp^K(\mu,\hat{\mu})\right)\ \diff(\strat_\maxplay(\hat{\rho}s')(\cdot))
    \label{eq:xsemimarkov:starredstrat:is:x:proof:iv} \\
    & =
    \int_{\enabled(s')} \diff(\strat_\maxplay(\hat{\rho}s')(\cdot))
    \label{eq:xsemimarkov:starredstrat:is:x:proof:v} \\
    & =
    \strat_\maxplay(\hat{\rho}s')(\enabled(s'))
    \label{eq:xsemimarkov:starredstrat:is:x:proof:vi} \\
    & = 1
    \label{eq:xsemimarkov:starredstrat:is:x:proof:vii}
  \end{align}
  The definition of $\starredstrat_\maxplay$ yields the first step
  (\ref{eq:xsemimarkov:starredstrat:is:x:proof:i}).
  Step (\ref{eq:xsemimarkov:starredstrat:is:x:proof:ii}) follows by
  observing that
  $\vertices(\sgactions(s')) =
  \{(K',\hat{\mu})\in\sgactions(s)\mid\hat{\mu}\in\vertices(K')\} =
  \{(K',\hat{\mu})\mid K'\in\psgtrans(s'), \hat{\mu}\in\vertices(K')\}$
  Using the lambda notation and the definition of $\convp$ on sets, we
  obtain (\ref{eq:xsemimarkov:starredstrat:is:x:proof:iii}).
  By noting that the sum is 0 if $K\notin \psgtrans(s')$, we introduce
  the characteristic function $\charac_{\psgtrans(s')}$ in
  (\ref{eq:xsemimarkov:starredstrat:is:x:proof:iv}), where we also
  apply the deintinion of $\convp$.
  The fact that
  $\sum_{\hat{\mu}\in\vertices(K)}\convp^K(\mu,\hat{\mu})=1$ and that
  $\charac_{\psgtrans(s')}(K)=1$ for all $(K,\mu)\in\enabled(s')$
  (since $K\in\psgtrans(s')$) yields
  (\ref{eq:xsemimarkov:starredstrat:is:x:proof:v}).
  Finally, step (\ref{eq:xsemimarkov:starredstrat:is:x:proof:vi})
  follows by the definition of the integral and step
  (\ref{eq:xsemimarkov:starredstrat:is:x:proof:vii}) because
  $\strat_\maxplay$ is a strategy.

  \medskip

  We now proceed to prove item~\ref{lm:xsemimarkov:i} by induction on
  $n$.  For $n=0$ we calculate:
  \[\Prob^{\strat_{\maxplay},\strat_{\minplay}}_{\StochGK,s}(D \Until^0 s') =
  \Prob^{\strat_{\maxplay},\strat_{\minplay},0}_{\StochGK,s}(s') =
  \Dirac_s(s') =
  \Prob^{\starredstrat_{\maxplay},\strat_{\minplay},0}_{\StochGK,s}(s') =
  \Prob^{\starredstrat_{\maxplay},\strat_{\minplay}}_{\StochGK,s}(D \Until^0 s').\]

  For $n+1>0$, we have that
  \begin{align*}
    \Prob^{\strat_{\maxplay},\strat_{\minplay}}_{\StochGK,s}(D \Until^{n+1} s')
    =
    \sum_{\hat{\rho}\in D^n,s''\in D}\Prob^{\strat_{\maxplay},\strat_{\minplay},n+1}_{\StochGK,s}(\hat{\rho}s''s')
    \hspace{-18em} &
    \\
    & =
    \sum_{\hat{\rho}\in D^n,s''\in(\sgnmax\cap D)}\Prob^{\strat_{\maxplay},\strat_{\minplay},n+1}_{\StochGK,s}(\hat{\rho}s''s') +
    \sum_{\hat{\rho}\in D^n,s''\in(\sgnmin\cap D)}\Prob^{\strat_{\maxplay},\strat_{\minplay},n+1}_{\StochGK,s}(\hat{\rho}s''s')
  \end{align*}
  and similarly for
  $\Prob^{\starredstrat_{\maxplay},\strat_{\minplay}}_{\StochGK,s}(D \Until^{n+1} s')$.
  Therefore, it sufficies to show that
  \begin{align}
    \sum_{\hat{\rho}\in D^n,s''\in\sgnmax}\Prob^{\strat_{\maxplay},\strat_{\minplay},n+1}_{\StochGK,s}(\hat{\rho}s''s')
    & =
    \sum_{\hat{\rho}\in D^n,s''\in\sgnmax}\Prob^{\starredstrat_{\maxplay},\strat_{\minplay},n+1}_{\StochGK,s}(\hat{\rho}s''s') \text{ \ and }
    \label{lm:xsemimarkov:i:obligation:i}\\
    \sum_{\hat{\rho}\in D^n,s''\in\sgnmin}\Prob^{\strat_{\maxplay},\strat_{\minplay},n+1}_{\StochGK,s}(\hat{\rho}s''s')
    & =
    \sum_{\hat{\rho}\in D^n,s''\in\sgnmin}\Prob^{\starredstrat_{\maxplay},\strat_{\minplay},n+1}_{\StochGK,s}(\hat{\rho}s''s').
    \label{lm:xsemimarkov:i:obligation:ii}
  \end{align}

  To prove (\ref{lm:xsemimarkov:i:obligation:i}), we calculate as
  follows:
  \begin{align}
    \sum_{\hat{\rho}\in D^n,s''\in\sgnmax}\Prob^{\starredstrat_{\maxplay},\strat_{\minplay},n+1}_{\StochGK,s}(\hat{\rho}s''s')
    \hspace{-11.4em}\notag\\
    & =
    \sum_{\hat{\rho}\in D^n,s''\in\sgnmax}\Prob^{\starredstrat_{\maxplay},\strat_{\minplay},n}_{\StochGK,s}(\hat{\rho}s'')\int_{\sgactions}\sgtrans(s'',\cdot,s')\ \diff(\starredstrat_{\maxplay}(\hat{\rho}s'')(\cdot))
    \label{eq:xsemimarkov:i:obligation:i:proof:i} \\
    & =
    \sum_{s''\in\sgnmax}\left(\sum_{\hat{\rho}\in D^n}\Prob^{\starredstrat_{\maxplay},\strat_{\minplay},n}_{\StochGK,s}(\hat{\rho}s'')\right)\int_{\sgactions}\sgtrans(s'',\cdot,s')\ \diff(\starredstrat_{\maxplay}(n,s'')(\cdot))
    \label{eq:xsemimarkov:i:obligation:i:proof:ii} \\
    & =
    \sum_{s''\in\sgnmax}\Prob^{\starredstrat_{\maxplay},\strat_{\minplay}}_{\StochGK,s}(D \Until^n s'')\sum_{a\in\vertices(\enabled(s''))}\sgtrans(s'',a,s')\mult\starredstrat_{\maxplay}(n,s'')(\{a\})
    \label{eq:xsemimarkov:i:obligation:i:proof:iii} \\
    & =
    \sum_{s''\in\sgnmax}\Prob^{\strat_{\maxplay},\strat_{\minplay}}_{\StochGK,s}(D \Until^n s'')\sum_{a\in\vertices(\enabled(s''))}\sgtrans(s'',a,s')\int_{\sgactions} \convp(\cdot,\{a\})\ \diff(\strat_\maxplay(n,s'')(\cdot))
    \label{eq:xsemimarkov:i:obligation:i:proof:iv} \\
    & =
    \sum_{s''\in\sgnmax}\left(\sum_{\hat{\rho}\in D^n}\Prob^{\strat_{\maxplay},\strat_{\minplay},n}_{\StochGK,s}(\hat{\rho}s'')\right)
    \notag\\
    &\qquad\qquad\qquad
    \int_{\sgactions} \lambda x.\left(\sum_{a\in\vertices(\enabled(s''))}\convp(x,a)\mult\sgtrans(s'',a,s')\right) \diff(\strat_\maxplay(n,s'')(\cdot))
    \label{eq:xsemimarkov:i:obligation:i:proof:v} \\
    & =
    \sum_{\hat{\rho}\in D^n,s''\in\sgnmax}\Prob^{\strat_{\maxplay},\strat_{\minplay},n}_{\StochGK,s}(\hat{\rho}s'')
    \notag\\
    &\qquad\qquad\qquad 
    \int_{\sgactions} \lambda x.\left(\sum_{a\in\vertices(x)}\convp(x,a)\mult\sgtrans(s'',a,s')\right) \diff(\strat_\maxplay(n,s'')(\cdot))
    \label{eq:xsemimarkov:i:obligation:i:proof:vi} \\
    & =
    \sum_{\hat{\rho}\in D^n,s''\in\sgnmax}\Prob^{\strat_{\maxplay},\strat_{\minplay},n}_{\StochGK,s}(\hat{\rho}s'') \int_{\sgactions} \sgtrans(s'',\cdot,s')\ \diff(\strat_\maxplay(\hat{\rho}s'')(\cdot))
    \label{eq:xsemimarkov:i:obligation:i:proof:vii} \\
    & =
    \sum_{\hat{\rho}\in D^n,s''\in\sgnmax}\Prob^{\strat_{\maxplay},\strat_{\minplay},n+1}_{\StochGK,s}(\hat{\rho}s''s')
    \label{eq:xsemimarkov:i:obligation:i:proof:viii}
  \end{align}
  (\ref{eq:xsemimarkov:i:obligation:i:proof:i}) follows by the
  definition of
  $\Prob^{\starredstrat_{\maxplay},\strat_{\minplay},n+1}_{\StochGK,s}$
  while (\ref{eq:xsemimarkov:i:obligation:i:proof:ii}) follows from
  the fact the $\starredstrat_\maxplay$ is semi-Markov.
  Step (\ref{eq:xsemimarkov:i:obligation:i:proof:iii}) follows from
  the definition of
  $\Prob^{\starredstrat_{\maxplay},\strat_{\minplay}}_{\StochGK,s}(D \Until^n s'')$
  and the fact $\vertices(\enabled(s))$, the support set of
  $\starredstrat_{\maxplay}(n,s'')$, is finite.
  In (\ref{eq:xsemimarkov:i:obligation:i:proof:iv}), induction
  hypothesis is applied as well as the definition of
  $\starredstrat_{\maxplay}(n,s'')(\{a\})$.
  (\ref{eq:xsemimarkov:i:obligation:i:proof:v}) follows by the
  definition of
  $\Prob^{\strat_{\maxplay},\strat_{\minplay}}_{\StochGK,s}(D \Until^n s'')$,
  calculations, the introduction of the $\lambda$ notation and the
  fact that $\convp(x,\{a\})=\convp(x,a)$.
  In (\ref{eq:xsemimarkov:i:obligation:i:proof:vi}), we observe that
  $\convp(x,a)>0$ only if $a\in\vertices(x)$.
  (\ref{eq:xsemimarkov:i:obligation:i:proof:vii}) follows from
  (\ref{eq:xsemimarkov:sgtrans:convp}) and the fact $\enabled(s)$ is
  the support set of $\strat_\maxplay(\hat{\rho}s'')$.
  Finally, the definition of
  $\Prob^{\strat_{\maxplay},\strat_{\minplay},n+1}_{\StochGK,s}$ is
  applied in (\ref{eq:xsemimarkov:i:obligation:i:proof:viii}).

  To prove (\ref{lm:xsemimarkov:i:obligation:ii}), we calculate as
  follows:
  \begin{align}
    \sum_{\hat{\rho}\in D^n,s''\in\sgnmin}\Prob^{\starredstrat_{\maxplay},\strat_{\minplay},n+1}_{\StochGK,s}(\hat{\rho}s''s')
    \hspace{-6em}\notag\\
    & =
    \sum_{\hat{\rho}\in D^n,s''\in\sgnmin}\Prob^{\starredstrat_{\maxplay},\strat_{\minplay},n}_{\StochGK,s}(\hat{\rho}s'')\int_{\sgactions}\sgtrans(s'',\cdot,s')\ \diff(\strat_{\minplay}(\hat{\rho}s'')(\cdot))
    \label{eq:xsemimarkov:i:obligation:ii:proof:i} \\
    & =
    \sum_{s''\in\sgnmin}\left(\sum_{\hat{\rho}\in D^n}\Prob^{\starredstrat_{\maxplay},\strat_{\minplay},n}_{\StochGK,s}(\hat{\rho}s'')\right)\int_{\sgactions}\sgtrans(s'',\cdot,s')\ \diff(\strat_{\minplay}(n,s'')(\cdot))
    \label{eq:xsemimarkov:i:obligation:ii:proof:ii} \\
    & =
    \sum_{s''\in\sgnmin}\Prob^{\starredstrat_{\maxplay},\strat_{\minplay}}_{\StochGK,s}(D \Until^n s'')\int_{\sgactions}\sgtrans(s'',\cdot,s')\ \diff(\strat_{\minplay}(n,s'')(\cdot))
    \label{eq:xsemimarkov:i:obligation:ii:proof:iii} \\
    & =
    \sum_{s''\in\sgnmin}\Prob^{\strat_{\maxplay},\strat_{\minplay}}_{\StochGK,s}(D \Until^n s'')\int_{\sgactions}\sgtrans(s'',\cdot,s')\ \diff(\strat_{\minplay}(n,s'')(\cdot))
    \label{eq:xsemimarkov:i:obligation:ii:proof:iv} \\
    & =
    \sum_{\hat{\rho}\in D^n,s''\in\sgnmin}\Prob^{\strat_{\maxplay},\strat_{\minplay},n+1}_{\StochGK,s}(\hat{\rho}s''s')
    \label{eq:xsemimarkov:i:obligation:ii:proof:v}
  \end{align}
  Step (\ref{eq:xsemimarkov:i:obligation:ii:proof:i}) follows by the
  definition of
  $\Prob^{\starredstrat_{\maxplay},\strat_{\minplay},n+1}_{\StochGK,s}$,
  (\ref{eq:xsemimarkov:i:obligation:ii:proof:ii}) follows from the
  fact the $\strat_\minplay$ is semi-Markov, and
  (\ref{eq:xsemimarkov:i:obligation:ii:proof:iii}) follows from the
  definition of
  $\Prob^{\starredstrat_{\maxplay},\strat_{\minplay}}_{\StochGK,s}(D \Until^n s'')$.
  Induction hypothesis is applied in
  (\ref{eq:xsemimarkov:i:obligation:ii:proof:iv}).
  Finally (\ref{eq:xsemimarkov:i:obligation:ii:proof:v}) follows like
  the first three steps in the inverse order.

  \medskip
  
  For~\ref{lm:xsemimarkov:ii}, we have that
  \begin{align*}
    \Prob^{\strat_{\maxplay},\strat_{\minplay}}_{\StochGK,s}(\Finally C)
    & =
    \sum_{s'\in C}\sum_{n\geq0}\Prob^{\strat_{\maxplay},\strat_{\minplay}}_{\StochGK,s}((\sgnodes\setminus C) \Until^n s')
    \\
    & =
    \sum_{s'\in C}\sum_{n\geq0}\Prob^{\starredstrat_{\maxplay},\strat_{\minplay}}_{\StochGK,s}((\sgnodes\setminus C) \Until^n s')
    =
    \Prob^{\starredstrat_{\maxplay},\strat_{\minplay}}_{\StochGK,s}(\Finally C)
  \end{align*}
  where the middle equality follow from item~\ref{lm:xsemimarkov:i},
  and the other two because, for all $n\neq m$ and $s'\in C$,
  $((\sgnodes\setminus C) \Until^n s') \cap ((\sgnodes\setminus C) \Until^m s') = \emptyset$.

  \medskip

  Item~\ref{lm:xsemimarkov:iii} can be calculated as follows:
  \begin{align*}
    \Expect^{\strat_{\maxplay},\strat_{\minplay}}_{\StochGK,s}[\GRewards]
    & =
    \lim_{n\to\infty} \sum^{n}_{i=0} \sum_{s' \in \sgnodes} \Prob^{\strat_{\maxplay},\strat_{\minplay}}_{\StochGK,s}(\Finally^i s')\mult\fgen(i,n)\mult\reward(s') 
    \tag{by Lemma~\ref{lm:encoding:reach:and:expectation}.\ref{lm:encoding:reach:and:expectation:ii}}\\
    & =
    \lim_{n\to\infty} \sum^{n}_{i=0} \sum_{s' \in \sgnodes} \Prob^{\starredstrat_{\maxplay},\strat_{\minplay}}_{\StochGK,s}(\Finally^i s')\mult\fgen(i,n)\mult\reward(s') 
    \tag{by item~\ref{lm:xsemimarkov:ii}}\\
    & =
    \Expect^{\starredstrat_{\maxplay},\strat_{\minplay}}_{\StochGK,s}[\GRewards]
    \tag{by Lemma~\ref{lm:encoding:reach:and:expectation}.\ref{lm:encoding:reach:and:expectation:ii}}
  \end{align*}

  \medskip
  
  Notice that the proof can be replicated mutatis mutandi with
  $\maxplay$ and $\minplay$ exchanged.  Therefore the last part of the
  lemma also holds.
\qed
\end{proof}

\begin{proposition}\label{prop:StochGK:StochHK}
  Let $\StochGK$ and $\StochHK$ be respectively the interpretation and
  the extreme interpretation of $\StochK$. Then
  \begin{enumerate}
  \item\label{prop:StochGK:StochHK:i}%
    For every
    $\strat_\maxplay\in\XSemiMarkovStrats{\StochGK,\maxplay}$ and
    $\strat_\minplay\in\XSemiMarkovStrats{\StochGK,\minplay}$,
    \begin{enumerate*}
    \item\label{prop:StochGK:StochHK:i:a}%
      $\Prob^{\strat_{\maxplay},\strat_{\minplay}}_{\StochGK,s}(\Finally C) =
      \Prob^{\stratv_{\maxplay},\stratv_{\minplay}}_{\StochHK,s}(\Finally C)$,
      for all $C\subseteq\sgnodes$, and
    \item\label{prop:StochGK:StochHK:i:b}%
      $\Expect^{\strat_{\maxplay},\strat_{\minplay}}_{\StochGK,s}[\GRewards] =
      \Expect^{\stratv_{\maxplay},\stratv_{\minplay}}_{\StochHK,s}[\GRewards]$; and
    \end{enumerate*}
  \item\label{prop:StochGK:StochHK:ii}%
    For every
    $\strat_\maxplay\in\SemiMarkovStrats{\StochHK,\maxplay}$ and
    $\strat_\minplay\in\SemiMarkovStrats{\StochHK,\minplay}$,
    \begin{enumerate*}
    \item\label{prop:StochGK:StochHK:ii:a}%
      $\Prob^{\stratx_{\maxplay},\stratx_{\minplay}}_{\StochGK,s}(\Finally C) =
      \Prob^{\strat_{\maxplay},\strat_{\minplay}}_{\StochHK,s}(\Finally C)$,
      for all $C\subseteq\sgnodes$, and
    \item\label{prop:StochGK:StochHK:ii:b}%
      $\Expect^{\stratx_{\maxplay},\stratx_{\minplay}}_{\StochGK,s}[\GRewards] =
      \Expect^{\strat_{\maxplay},\strat_{\minplay}}_{\StochHK,s}[\GRewards]$.
    \end{enumerate*}
  \end{enumerate}
\end{proposition}

\begin{proof}[of Proposition~\ref{prop:StochGK:StochHK}]
  To prove item \ref{prop:StochGK:StochHK:i}, we first prove by
  induction that for all $n\geq0$ and $\hat{\rho}\in\sgnodes^{n+1}$,
  \begin{equation}\label{eq:StochGK:StochHK:i:aux}
    \Prob^{\strat_{\maxplay},\strat_{\minplay},n}_{\StochGK,s}(\hat{\rho}) =
    \Prob^{\stratv_{\maxplay},\stratv_{\minplay},n}_{\StochHK,s}(\hat{\rho})
  \end{equation}

  The case $n=0$ is direct. For $n+1>0$, $\hat{\rho}\in\sgnodes^n$,
  $s'\in\sgnodes$ and $s''\in\sgnmax$ we calculate as follows
  \begin{align}
    \Prob^{\strat_{\maxplay},\strat_{\minplay},n+1}_{\StochGK,s}(\hat{\rho}s''s')
    & =
    \Prob^{\strat_{\maxplay},\strat_{\minplay},n}_{\StochGK,s}(\hat{\rho}s'')\int_{\sgactions}\sgtrans(s'',\cdot,s')\ \diff(\strat_{\maxplay}(\hat{\rho}s'')(\cdot))
    \label{eq:StochGK:StochHK:i:aux:i}\\
    & =
    \Prob^{\strat_{\maxplay},\strat_{\minplay},n}_{\StochGK,s}(\hat{\rho}s'')\sum_{a\in\vertices(\enabled(s''))}\sgtrans(s'',a,s')\mult\strat_{\maxplay}(\hat{\rho}s'')(\{a\})
    \label{eq:StochGK:StochHK:i:aux:ii}\\
    & =
    \Prob^{\stratv_{\maxplay},\stratv_{\minplay},n}_{\StochHK,s}(\hat{\rho}s'')\sum_{a\in\enabled_\StochHK(s'')}\sgtrans_\StochHK(s'',a,s')\mult\stratv_{\maxplay}(\hat{\rho}s'')(\{a\})
    \label{eq:StochGK:StochHK:i:aux:iii}\\
    & =
    \Prob^{\stratv_{\maxplay},\stratv_{\minplay},n+1}_{\StochHK,s}(\hat{\rho}s''s')
    \label{eq:StochGK:StochHK:i:aux:iv}
  \end{align}
  Step (\ref{eq:StochGK:StochHK:i:aux:i}) follows from definition of
  $\Prob^{\strat_{\maxplay},\strat_{\minplay},n+1}_{\StochGK,s}(\hat{\rho}s''s')$.
  To obtain (\ref{eq:StochGK:StochHK:i:aux:ii}) we use the fact that
  the support set $\vertices(\enabled(s''))$ of
  $\strat_\maxplay(\hat{\rho}s'')$ is finite.
  In (\ref{eq:StochGK:StochHK:i:aux:iii}) induction hypothesis is
  applied, together with the definitions of
  $\sgtrans_\StochHK(s'',a,s')$ and
  $\stratv_{\maxplay}(\hat{\rho}s'')(\{a\})$.
  Finally, the definition of
  $\Prob^{\stratv_{\maxplay},\stratv_{\minplay},n+1}_{\StochHK,s}(\hat{\rho}s''s')$
  yields (\ref{eq:StochGK:StochHK:i:aux:iv}).
  
  The calculations follow similarly for the case of $s''\in\sgnmin$,
  which proves (\ref{eq:StochGK:StochHK:i:aux}).

  Now item \ref{prop:StochGK:StochHK:i:a} follows from the following
  calculations:
  \begin{align}
    \Prob^{\strat_{\maxplay},\strat_{\minplay}}_{\StochGK,s}(\Finally C)
    & =
    \sum_{n\geq0}\Prob^{\strat_{\maxplay},\strat_{\minplay}}_{\StochGK,s}((\sgnodes\setminus C) \Until^n C)
    \label{eq:StochGK:StochHK:i:a:proof:i}\\
    & =
    \sum_{n\geq0}\sum_{s'\in C}\sum_{\hat{\rho}\in (\sgnodes\setminus C)^n}\Prob^{\strat_{\maxplay},\strat_{\minplay},n}_{\StochGK,s}(\hat{\rho}s')
    \label{eq:StochGK:StochHK:i:a:proof:ii}\\
    & =
    \sum_{n\geq0}\sum_{s'\in C}\sum_{\hat{\rho}\in (\sgnodes\setminus C)^n}\Prob^{\stratv_{\maxplay},\stratv_{\minplay},n}_{\StochHK,s}(\hat{\rho}s')
    \label{eq:StochGK:StochHK:i:a:proof:iii}\\
    & =
    \Prob^{\stratv_{\maxplay},\stratv_{\minplay}}_{\StochHK,s}(\Finally C)
    \label{eq:StochGK:StochHK:i:a:proof:iv}
  \end{align}
  Equality~(\ref{eq:StochGK:StochHK:i:a:proof:i}) follows from the fact
  $((\sgnodes\setminus C) \Until^n C) \cap ((\sgnodes\setminus C) \Until^m C) = \emptyset$
  whenever $n\neq m$ and
  equality~(\ref{eq:StochGK:StochHK:i:a:proof:ii}), from the definition
  of
  $\Prob^{\strat_{\maxplay},\strat_{\minplay}}_{\StochGK,s}((\sgnodes\setminus C) \Until^n C)$.
  The auxiliary result (\ref{eq:StochGK:StochHK:i:aux}) yields
  (\ref{eq:StochGK:StochHK:i:a:proof:iii}) and
  (\ref{eq:StochGK:StochHK:i:a:proof:iv}) follows as the previous steps
  in inverse direction.

  For item~\ref{prop:StochGK:StochHK:i:b}, we calculate as follows
  \begin{align*}
    \Expect^{\strat_{\maxplay},\strat_{\minplay}}_{\StochGK,s}[\GRewards]
    \hspace{-3em} &
    \\
    & =    
    \lim_{n\to\infty} \sum^{n}_{i=0}\sum_{s' \in \sgnodes} \Prob^{\strat_{\maxplay},\strat_{\minplay}}_{\StochGK,s}(\Finally^i s')\mult\fgen(i,n)\mult\reward(s')
    \tag{by Lemma~\ref{lm:encoding:reach:and:expectation}.\ref{lm:encoding:reach:and:expectation:ii}}\\
    & =    
    \lim_{n\to\infty} \sum^{n}_{i=0}\sum_{s' \in \sgnodes}\sum_{\hat{\rho} \in \sgnodes^i} \Prob^{\strat_{\maxplay},\strat_{\minplay},i}_{\StochGK,s}(\hat{\rho}s')\mult\fgen(i,n)\mult\reward(s')
    \tag{by def. of $\Prob^{\strat_{\maxplay},\strat_{\minplay}}_{\StochGK,s}(\Finally^i s')$}\\
    & =    
    \lim_{n\to\infty} \sum^{n}_{i=0}\sum_{s' \in \sgnodes}\sum_{\hat{\rho} \in \sgnodes^i} \Prob^{\stratv_{\maxplay},\stratv_{\minplay},i}_{\StochGK,s}(\hat{\rho}s')\mult\fgen(i,n)\mult\reward(s')
    \tag{because of (\ref{eq:StochGK:StochHK:i:aux})}\\
    & =
    \lim_{n\to\infty} \sum^{n}_{i=0}\sum_{s' \in \sgnodes} \Prob^{\stratv_{\maxplay},\stratv_{\minplay}}_{\StochGK,s}(\Finally^i s')\mult\fgen(i,n)\mult\reward(s')
    \tag{by def. of $\Prob^{\stratv_{\maxplay},\stratv_{\minplay}}_{\StochGK,s}(\Finally^i s')$}\\
    & =
    \Expect^{\stratv_{\maxplay},\stratv_{\minplay}}_{\StochGK,s}[\GRewards]
    \tag{by Lemma~\ref{lm:encoding:reach:and:expectation}.\ref{lm:encoding:reach:and:expectation:ii}}
  \end{align*}

  To prove item \ref{prop:StochGK:StochHK:ii} we proceed similarly.
  Thus the proof boils down to show that for all $n\geq0$ and
  $\hat{\rho}\in\sgnodes^{n+1}$,
  $\Prob^{\stratx_{\maxplay},\stratx_{\minplay},n}_{\StochGK,s}(\hat{\rho})
  =
  \Prob^{\strat_{\maxplay},\strat_{\minplay},n}_{\StochHK,s}(\hat{\rho})$
  which can be done just like for (\ref{eq:StochGK:StochHK:i:aux}).
  \qed
\end{proof}

\begin{proof}[of Proposition~\ref{prop:infsup:supinf:StochGK:StochHK}]
  Let $\strat_i\in\XSemiMarkovStrats{\StochGK,i}$.  Then, for all
  $\hat{\rho}\in\sgnodes^*$, $s\in\sgnodes$, and
  $A\in\Salg_\sgactions$,
  $\supx{(\stratv_i)}(\hat{\rho}s)(A)=\stratv_i(\hat{\rho}s)(A\cap\vertices(\sgactions))=\strat_i(\hat{\rho}s)(A\cap\vertices(\sgactions))=\strat_i(\hat{\rho}s)(A)$.
  All equalities follows from definitions except the last one that
  follows from the fact that $\strat_i$ is extreme.
  Similarly, for $\strat_i\in\SemiMarkovStrats{\StochHK,i}$, and all
  $\hat{\rho}\in\sgnodes^*$, $s\in\sgnodes$, and
  $A\subseteq\vertices(\sgactions)$,
  $\supv{(\stratx_i)}(\hat{\rho}s)(A)=\stratx_i(\hat{\rho}s)(A)=\strat_i(\hat{\rho}s)(A\cap\vertices(\sgactions))=\strat_i(\hat{\rho}s)(A)$.
  Again, all equalities follows from definitions except the last one
  that follows from the fact that $A\subseteq\vertices(\sgactions)$.
  From these observations, it follows that
  \begin{equation}\label{eq:observation:semimarkov:xsemimarkov}
    \SemiMarkovStrats{\StochHK,i}=\{\stratv_i\mid\strat_i\in\XSemiMarkovStrats{\StochGK,i}\}
    \text{ and }
    \XSemiMarkovStrats{\StochGK,i}=\{\stratx_i\mid\strat_i\in\SemiMarkovStrats{\StochHK,i}\}.
  \end{equation}
  These equalities and Proposition~\ref{prop:StochGK:StochHK} leads to
  the result.
  \qed
\end{proof}

\begin{proof}[Proof of Proposition~\ref{prop:stopping:irreducible:StochGK:StochHK}]
  For item~\ref{prop:stopping:irreducible:StochGK:StochHK:i} we have that
  \begin{align*}
    1
    & =
    \inf_{\strat_\minplay\in\Strategies{\StochGK,\minplay}}\inf_{\strat_\maxplay\in\Strategies{\StochGK,\maxplay}}\Prob^{\strat_{\maxplay},\strat_{\minplay}}_{\StochGK,s}(\Finally T)
    \tag{def. of stopping}\\
    & \leq
    \inf_{\strat_\minplay\in\XSemiMarkovStrats{\StochGK,\minplay}}\inf_{\strat_\maxplay\in\XSemiMarkovStrats{\StochGK,\maxplay}}\Prob^{\strat_{\maxplay},\strat_{\minplay}}_{\StochGK,s}(\Finally T)
    \tag{$\XSemiMarkovStrats{\StochGK,i}\subseteq\Strategies{\StochGK,\minplay}$}\\
    & =
    \inf_{\strat_\minplay\in\SemiMarkovStrats{\StochHK,\minplay}}\inf_{\strat_\maxplay\in\SemiMarkovStrats{\StochHK,\maxplay}}\Prob^{\strat_{\maxplay},\strat_{\minplay}}_{\StochHK,s}(\Finally T)
    \tag{by observation~(\ref{eq:observation:semimarkov:xsemimarkov})}\\
    & =
    \inf_{\strat_\minplay\in\Strategies{\StochHK,\minplay}}\inf_{\strat_\maxplay\in\Strategies{\StochHK,\maxplay}}\Prob^{\strat_{\maxplay},\strat_{\minplay}}_{\StochHK,s}(\Finally T)
  \end{align*}
  Since $\StochHK$ is finite, the last step follows from standard
  results on MDP~\cite{Puterman94}.
  
  With a similar reasoning we obtain
  $0 <
  \inf_{\strat_\minplay\in\Strategies{\StochGK,\minplay}}\inf_{\strat_\maxplay\in\Strategies{\StochGK,\maxplay}}\Prob^{\strat_{\maxplay},\strat_{\minplay}}_{\StochGK,s}(\Finally s')
  \leq
  \inf_{\strat_\minplay\in\Strategies{\StochHK,\minplay}}\inf_{\strat_\maxplay\in\Strategies{\StochHK,\maxplay}}\Prob^{\strat_{\maxplay},\strat_{\minplay}}_{\StochHK,s}(\Finally s')$,
  proving thus also item~\ref{prop:stopping:irreducible:StochGK:StochHK:ii}.
  \qed
\end{proof}

\begin{proof}[of Theorem~\ref{th:determinacy:and:discretazation}]
  For item \ref{th:determinacy:and:discretazation:i}, we calculate as
  follows:
  \begin{align*}
    \inf_{\strat_\minplay\in\Strategies{\StochGK,\minplay}}\sup_{\strat_\maxplay\in\Strategies{\StochGK,\maxplay}}\Prob^{\strat_{\maxplay},\strat_{\minplay}}_{\StochGK,s}(\Finally C)
    \hspace{-8em} &
    \\
    & \leq
    \inf_{\strat_\minplay\in\SemiMarkovStrats{\StochGK,\minplay}}\sup_{\strat_\maxplay\in\Strategies{\StochGK,\maxplay}}\Prob^{\strat_{\maxplay},\strat_{\minplay}}_{\StochGK,s}(\Finally C)
    \tag{$\SemiMarkovStrats{\StochGK,\minplay}\subseteq\Strategies{\StochGK,\minplay}$}\\
    & =
    \inf_{\strat_\minplay\in\SemiMarkovStrats{\StochGK,\minplay}}\sup_{\strat_\maxplay\in\SemiMarkovStrats{\StochGK,\maxplay}}\Prob^{\strat_{\maxplay},\strat_{\minplay}}_{\StochGK,s}(\Finally C)
    \tag{by Lemma~\ref{lm:semimarkov}.\ref{lm:semimarkov:ii}}\\
    & =
    \inf_{\strat_\minplay\in\XSemiMarkovStrats{\StochGK,\minplay}}\sup_{\strat_\maxplay\in\XSemiMarkovStrats{\StochGK,\maxplay}}\Prob^{\strat_{\maxplay},\strat_{\minplay}}_{\StochGK,s}(\Finally C)
    \tag{by Corollary~\ref{cor:xsemimarkov}.\ref{cor:xsemimarkov:i}}\\
    & =
    \inf_{\strat_\minplay\in\SemiMarkovStrats{\StochHK,\minplay}}\sup_{\strat_\maxplay\in\SemiMarkovStrats{\StochHK,\maxplay}}\Prob^{\strat_{\maxplay},\strat_{\minplay}}_{\StochHK,s}(\Finally C)
    \tag{by Prop.~\ref{prop:infsup:supinf:StochGK:StochHK}.\ref{prop:infsup:supinf:StochGK:StochHK:i}}\\
    & \leq
    \inf_{\strat_\minplay\in\DetMemorylessStrats{\StochHK,\minplay}}\sup_{\strat_\maxplay\in\SemiMarkovStrats{\StochHK,\maxplay}}\Prob^{\strat_{\maxplay},\strat_{\minplay}}_{\StochHK,s}(\Finally C)
    \tag{$\DetMemorylessStrats{\StochHK,\minplay}\subseteq\SemiMarkovStrats{\StochHK,\minplay}$}\\
    & =
    \inf_{\strat_\minplay\in\DetMemorylessStrats{\StochHK,\minplay}}\sup_{\strat_\maxplay\in\DetMemorylessStrats{\StochHK,\maxplay}}\Prob^{\strat_{\maxplay},\strat_{\minplay}}_{\StochHK,s}(\Finally C)
    \tag{by Prop.~\ref{prop:mdp:results}.\ref{prop:mdp:results:i}}\\
    & =
    \sup_{\strat_\maxplay\in\DetMemorylessStrats{\StochHK,\maxplay}}\inf_{\strat_\minplay\in\DetMemorylessStrats{\StochHK,\minplay}}\Prob^{\strat_{\maxplay},\strat_{\minplay}}_{\StochHK,s}(\Finally C)
    \tag{by \cite[Lemma 6]{Condon92}}\\
    & =
    \sup_{\strat_\maxplay\in\DetMemorylessStrats{\StochHK,\maxplay}}\inf_{\strat_\minplay\in\SemiMarkovStrats{\StochHK,\minplay}}\Prob^{\strat_{\maxplay},\strat_{\minplay}}_{\StochHK,s}(\Finally C)
    \tag{by Prop.~\ref{prop:mdp:results}.\ref{prop:mdp:results:ii}}\\
    & \leq
    \sup_{\strat_\maxplay\in\SemiMarkovStrats{\StochHK,\maxplay}}\inf_{\strat_\minplay\in\SemiMarkovStrats{\StochHK,\minplay}}\Prob^{\strat_{\maxplay},\strat_{\minplay}}_{\StochHK,s}(\Finally C)
    \tag{$\DetMemorylessStrats{\StochHK,\maxplay}\subseteq\SemiMarkovStrats{\StochHK,\maxplay}$}\\
    & =
    \sup_{\strat_\maxplay\in\XSemiMarkovStrats{\StochGK,\maxplay}}\inf_{\strat_\minplay\in\XSemiMarkovStrats{\StochGK,\minplay}}\Prob^{\strat_{\maxplay},\strat_{\minplay}}_{\StochGK,s}(\Finally C)
    \tag{by Prop.~\ref{prop:infsup:supinf:StochGK:StochHK}.\ref{prop:infsup:supinf:StochGK:StochHK:ii}}\\
    & =
    \sup_{\strat_\maxplay\in\SemiMarkovStrats{\StochGK,\maxplay}}\inf_{\strat_\minplay\in\SemiMarkovStrats{\StochGK,\minplay}}\Prob^{\strat_{\maxplay},\strat_{\minplay}}_{\StochGK,s}(\Finally C)
    \tag{by Corollary~\ref{cor:xsemimarkov}.\ref{cor:xsemimarkov:i}}\\
    & =
    \sup_{\strat_\maxplay\in\SemiMarkovStrats{\StochGK,\maxplay}}\inf_{\strat_\minplay\in\Strategies{\StochGK,\minplay}}\Prob^{\strat_{\maxplay},\strat_{\minplay}}_{\StochGK,s}(\Finally C)
    \tag{by Lemma~\ref{lm:semimarkov}.\ref{lm:semimarkov:ii}}\\
    & \leq
    \sup_{\strat_\maxplay\in\Strategies{\StochGK,\maxplay}}\inf_{\strat_\minplay\in\Strategies{\StochGK,\minplay}}\Prob^{\strat_{\maxplay},\strat_{\minplay}}_{\StochGK,s}(\Finally C)
    \tag{$\SemiMarkovStrats{\StochGK,\maxplay}\subseteq\Strategies{\StochGK,\maxplay}$}\\
    & \leq
    \inf_{\strat_\minplay\in\Strategies{\StochGK,\minplay}}\sup_{\strat_\maxplay\in\Strategies{\StochGK,\maxplay}}\Prob^{\strat_{\maxplay},\strat_{\minplay}}_{\StochGK,s}(\Finally C)
    \tag{by prop. of $\sup$ and $\inf$}
  \end{align*}
  Since the last term is equal to the first term in the calculation,
  item~\ref{th:determinacy:and:discretazation:i} is concluded.

  For item~\ref{th:determinacy:and:discretazation:ii} we calculate as follows:
  \begin{align*}
    \inf_{\strat_\minplay\in\Strategies{\StochGK,\minplay}}\sup_{\strat_\maxplay\in\Strategies{\StochGK,\maxplay}}\Expect^{\strat_{\maxplay},\strat_{\minplay}}_{\StochGK,s}(\GRewards)
    \hspace{-8em} &
    \\
    & \leq
    \inf_{\strat_\minplay\in\SemiMarkovStrats{\StochGK,\minplay}}\sup_{\strat_\maxplay\in\Strategies{\StochGK,\maxplay}}\Expect^{\strat_{\maxplay},\strat_{\minplay}}_{\StochGK,s}(\GRewards)
    \tag{$\SemiMarkovStrats{\StochGK,\minplay}\subseteq\Strategies{\StochGK,\minplay}$}\\
    & =
    \inf_{\strat_\minplay\in\SemiMarkovStrats{\StochGK,\minplay}}\sup_{\strat_\maxplay\in\SemiMarkovStrats{\StochGK,\maxplay}}\Expect^{\strat_{\maxplay},\strat_{\minplay}}_{\StochGK,s}(\GRewards)
    \tag{by Lemma~\ref{lm:semimarkov}.\ref{lm:semimarkov:iii}}\\
    & =
    \inf_{\strat_\minplay\in\XSemiMarkovStrats{\StochGK,\minplay}}\sup_{\strat_\maxplay\in\XSemiMarkovStrats{\StochGK,\maxplay}}\Expect^{\strat_{\maxplay},\strat_{\minplay}}_{\StochGK,s}(\GRewards)
    \tag{by Corollary~\ref{cor:xsemimarkov}.\ref{cor:xsemimarkov:ii}}\\
    & =
    \inf_{\strat_\minplay\in\SemiMarkovStrats{\StochHK,\minplay}}\sup_{\strat_\maxplay\in\SemiMarkovStrats{\StochHK,\maxplay}}\Expect^{\strat_{\maxplay},\strat_{\minplay}}_{\StochHK,s}(\GRewards)
    \tag{by Prop.~\ref{prop:infsup:supinf:StochGK:StochHK}.\ref{prop:infsup:supinf:StochGK:StochHK:iii}}\\
    & \leq
    \inf_{\strat_\minplay\in\DetMemorylessStrats{\StochHK,\minplay}}\sup_{\strat_\maxplay\in\SemiMarkovStrats{\StochHK,\maxplay}}\Expect^{\strat_{\maxplay},\strat_{\minplay}}_{\StochHK,s}(\GRewards)
    \tag{$\DetMemorylessStrats{\StochHK,\minplay}\subseteq\SemiMarkovStrats{\StochHK,\minplay}$}\\
    & =
    \inf_{\strat_\minplay\in\DetMemorylessStrats{\StochHK,\minplay}}\sup_{\strat_\maxplay\in\DetMemorylessStrats{\StochHK,\maxplay}}\Expect^{\strat_{\maxplay},\strat_{\minplay}}_{\StochHK,s}(\GRewards)
    \tag{by Prop.~\ref{prop:mdp:results}.\ref{prop:mdp:results:iii}}\\
    & =
    \sup_{\strat_\maxplay\in\DetMemorylessStrats{\StochHK,\maxplay}}\inf_{\strat_\minplay\in\DetMemorylessStrats{\StochHK,\minplay}}\Expect^{\strat_{\maxplay},\strat_{\minplay}}_{\StochHK,s}(\GRewards)
    \tag{*}\\
    & =
    \sup_{\strat_\maxplay\in\DetMemorylessStrats{\StochHK,\maxplay}}\inf_{\strat_\minplay\in\SemiMarkovStrats{\StochHK,\minplay}}\Expect^{\strat_{\maxplay},\strat_{\minplay}}_{\StochHK,s}(\GRewards)
    \tag{by Prop.~\ref{prop:mdp:results}.\ref{prop:mdp:results:iv}}\\
    & \leq
    \sup_{\strat_\maxplay\in\SemiMarkovStrats{\StochHK,\maxplay}}\inf_{\strat_\minplay\in\SemiMarkovStrats{\StochHK,\minplay}}\Expect^{\strat_{\maxplay},\strat_{\minplay}}_{\StochHK,s}(\GRewards)
    \tag{$\DetMemorylessStrats{\StochHK,\maxplay}\subseteq\SemiMarkovStrats{\StochHK,\maxplay}$}\\
    & =
    \sup_{\strat_\maxplay\in\XSemiMarkovStrats{\StochGK,\maxplay}}\inf_{\strat_\minplay\in\XSemiMarkovStrats{\StochGK,\minplay}}\Expect^{\strat_{\maxplay},\strat_{\minplay}}_{\StochGK,s}(\GRewards)
    \tag{by Prop.~\ref{prop:infsup:supinf:StochGK:StochHK}.\ref{prop:infsup:supinf:StochGK:StochHK:iv}}\\
    & =
    \sup_{\strat_\maxplay\in\SemiMarkovStrats{\StochGK,\maxplay}}\inf_{\strat_\minplay\in\SemiMarkovStrats{\StochGK,\minplay}}\Expect^{\strat_{\maxplay},\strat_{\minplay}}_{\StochGK,s}(\GRewards)
    \tag{by Corollary~\ref{cor:xsemimarkov}.\ref{cor:xsemimarkov:ii}}\\
    & =
    \sup_{\strat_\maxplay\in\SemiMarkovStrats{\StochGK,\maxplay}}\inf_{\strat_\minplay\in\Strategies{\StochGK,\minplay}}\Expect^{\strat_{\maxplay},\strat_{\minplay}}_{\StochGK,s}(\GRewards)
    \tag{by Lemma~\ref{lm:semimarkov}.\ref{lm:semimarkov:iii}}\\
    & 
    \leq
    \sup_{\strat_\maxplay\in\Strategies{\StochGK,\maxplay}}\inf_{\strat_\minplay\in\Strategies{\StochGK,\minplay}}\Expect^{\strat_{\maxplay},\strat_{\minplay}}_{\StochGK,s}(\GRewards)
    \tag{$\SemiMarkovStrats{\StochGK,\maxplay}\subseteq\Strategies{\StochGK,\maxplay}$}\\
    & \leq
    \inf_{\strat_\minplay\in\Strategies{\StochGK,\minplay}}\sup_{\strat_\maxplay\in\Strategies{\StochGK,\maxplay}}\Expect^{\strat_{\maxplay},\strat_{\minplay}}_{\StochGK,s}(\GRewards)
    \tag{by prop. of $\sup$ and $\inf$}
  \end{align*}
  Since the last term is equal to the first term in the calculation,
  item~\ref{th:determinacy:and:discretazation:ii} is concluded.
  In particular, step (*) is justified as follows, depending on
  $\GRewards$:
  \begin{itemize}
  \item%
    For $\GRewards=\TRewards$, (*) follows by
    \cite[Theorem~4.2.6]{FilarV96} since, by
    Proposition~\ref{prop:stopping:irreducible:StochGK:StochHK}.\ref{prop:stopping:irreducible:StochGK:StochHK:i},
    the game $\StochHK$ is also almost surely stopping.
  \item%
    For $\GRewards=\DRewards{\gamma}$ (*) follows by
    \cite[Theorem~4.3.2]{FilarV96}.
  \item%
    For $\GRewards=\ARewards$ (*) follows by
    \cite[Theorem~5.1.5]{FilarV96} since, by
    Proposition~\ref{prop:stopping:irreducible:StochGK:StochHK}.\ref{prop:stopping:irreducible:StochGK:StochHK:ii},
    the game $\StochHK$ is also irreducible.
  \qed
  \end{itemize}
\end{proof}

\end{document}